%% file: Shapes.tex
\documentclass[10pt]{article}
\usepackage{graphicx}
\usepackage{cite}
\usepackage{hyperref}
\usepackage{amsmath}
\usepackage{amssymb}
\usepackage{amsfonts}
\usepackage{amsthm}
\usepackage{fullpage}
\usepackage{multirow}
\usepackage{amsthm}
\usepackage{subfig}
\usepackage{sidecap}

\input{commands-tam.sty}

\theoremstyle{definition}
\newtheorem{lemma}{Lemma}[section]
\newtheorem{theorem}{Theorem}[section]

\begin{document}

\title{\textbf{Optimal self-assembly of finite shapes at temperature 1 in 3D}}

\author{%
David Furcy\thanks{Computer Science Department, University of Wisconsin--Oshkosh, Oshkosh, WI 54901, USA,\protect\url{furcyd@uwosh.edu}.}
\and
Scott M. Summers\thanks{Computer Science Department, University of Wisconsin--Oshkosh, Oshkosh, WI 54901, USA,\protect\url{summerss@uwosh.edu}. This author's research was supported in part by UWO Faculty Development Research grant FDR881.}
}


\date{}
\maketitle

\input{abstract}


\input{introduction}
\input{definitions}

\input{main_theorem}

\input{construction}
\input{conclusion}
\section*{Acknowledgement}
We thank Matthew Patitz for offering helpful improvements to the presentation of our main construction.

\bibliographystyle{amsplain}
\bibliography{tam}

\newpage
\input{Turing_machine_simulation}

\clearpage
\input{appendix}

\end{document}

%% file: abstract.tex
\begin{abstract}
Working in a three-dimensional variant of Winfree's abstract Tile Assembly Model, we show that, for an arbitrary finite, connected shape $X \subset \mathbb{Z}^2$, there is a tile set that uniquely self-assembles into a 3D representation of $X$ at temperature 1 in 3D with optimal program-size complexity (the program-size complexity, also known as tile complexity, of a shape is the minimum number of tile types required to uniquely self-assemble it). Moreover, our construction is ``just barely'' 3D in the sense that it only places tiles in the $z = 0$ and $z = 1$ planes. Our result is essentially a just-barely 3D temperature 1 simulation of a similar 2D temperature 2 result by Soloveichik and Winfree (SICOMP 2007).
\end{abstract} 

%% file: introduction.tex
\section{Introduction}
Self-assembly is intuitively defined as the process through which simple, unorganized components spontaneously combine, according to local interaction rules, to form some kind of organized final structure. While examples of self-assembly in nature are abundant, Seeman \cite{Seem82} was the first to demonstrate the feasibility of self-assembling man-made DNA tile molecules. Since then, self-assembly researchers have used principles from DNA tile self-assembly to self-assemble a wide variety of nanoscale structures, such as regular arrays~\cite{WinLiuWenSee98}, fractal structures~\cite{RoPaWi04,FujHarParWinMur07}, smiling faces~\cite{rothemund2006folding}, DNA tweezers~\cite{yurke2000dna}, logic circuits~\cite{qian2011scaling}, neural networks~\cite{qian2011neural}, and molecular robots\cite{DNARobotNature2010}.
What is more, over roughly the past decade, researchers have dramatically reducing the tile placement error rate (the percentage of incorrect tile placements) for DNA tile self-assembly from 10\% to 0.05\% \cite{RoPaWi04,FujHarParWinMur07,BarSchRotWin09,ConstEvansPhD}.

In 1998, Winfree \cite{Winf98} introduced the abstract Tile Assembly Model (aTAM) as an over-simplified, combinatorial, error-free model of experimental DNA tile self-assembly. The aTAM is a constructive version of mathematical Wang tiling \cite{Wang61} in that the former bestows upon the latter a mechanism for sequential ``growth'' of a tile assembly starting from an initial seed. Very briefly, in the aTAM, the fundamental components are un-rotatable, translatable square ``tile types'' whose sides are labeled with (alpha-numeric) glue ``colors'' and (integer) ``strengths''. Two tiles that are placed next to each other \emph{bind} if both the glue colors and the strengths on their abutting sides match and the sum of their matching strengths sum to at least a certain (integer) ``temperature''. Self-assembly starts from a ``seed'' tile type, typically assumed to be placed at the origin, and proceeds nondeterministically and asynchronously as tiles bind to the seed-containing assembly one at a time. In this paper, we work in a three-dimensional variant of the aTAM in which tile types are unit cubes and tiles are placed in a \emph{non-cooperative} manner.

Tile self-assembly in which tiles may bind to an existing assembly in a non-cooperative fashion is often referred to as ``temperature~1 self-assembly'' or simply ``non-cooperative self-assembly''. In this type of self-assembly, a tile may non-cooperatively bind to an assembly via (at least) one of its sides, unlike in \emph{cooperative} self-assembly, in which some tiles may be required to bind on two or more sides. It is worth noting that cooperative self-assembly leads to highly non-trivial theoretical behavior, e.g., Turing universality \cite{Winf98} and the efficient self-assembly of $N \times N$ squares \cite{AdlemanCGH01,RotWin00} and other algorithmically specified shapes \cite{SolWin07}.

Despite its theoretical algorithmic capabilities, when cooperative self-assembly is implemented using DNA tiles in the laboratory \cite{RoPaWi04,BarSchRotWin09,SchWin07,MaoLabReiSee00,WinLiuWenSee98}, tiles may (and do) erroneously bind in a non-cooperative fashion, which usually results in the production of undesired final structures. In order to completely avoid the erroneous effects of tiles unexpectedly binding in a non-cooperative fashion, the experimenter should only build nanoscale structures using constructions that are guaranteed to work correctly in non-cooperative self-assembly. Thus, characterizing the theoretical power of non-cooperative self-assembly has significant practical implications.

Although no characterization of the power of non-cooperative self-assembly exists at the time of this writing, Doty, Patitz and Summers conjecture \cite{jLSAT1} that 2D non-cooperative self-assembly is weaker than 2D cooperative self-assembly because a certain technical condition, known as ``pumpability'', is true for any 2D non-cooperative tile set. If the pumpability conjecture is true, then non-cooperative 2D self-assembly can only produce simple, highly-regular shapes and patterns, which are too simple and regular to be the result of complex computation.

In addition to the pumpability conjecture, there are a number of results that study the suspected weakness of non-cooperative self-assembly. For example, Rothemund and Winfree \cite{RotWin00} proved that, if the final assembly must be fully connected, then the minimum number of unique tile types required to self-assemble an $N \times N$ square (i.e., its \emph{tile complexity}) is exactly $2N - 1$. Manuch, Stacho and Stoll \cite{ManuchSS10} showed that the previous tile complexity is also true when the final assembly cannot contain even any glue mismatches. Moreover, at the time of this writing, the only way in which non-cooperative self-assembly has been shown to be unconditionally weaker than cooperative self-assembly is in the sense of \emph{intrinsic universality} \cite{iusa,usa}. First, Doty, Lutz, Patitz, Schweller, Summers and Woods \cite{iusa} proved the existence of a universal cooperative tile set that can be programmed to simulate the behavior of any tile set (i.e., the aTAM is intrinsically universal for itself). Then, Meunier, Patitz, Summers, Theyssier, Winslow and Woods \cite{WindowMovieLemma} showed, via a combinatorial argument, that there is no universal non-cooperative tile set that can be programmed to simulate the behavior of an arbitrary (cooperative) tile set. Thus, in the sense of intrinsic universality, non-cooperative self-assembly is strictly weaker than cooperative self-assembly.

While non-cooperative self-assembly is suspected of being strictly weaker than cooperative self-assembly, in general, it is interesting to note that 3D non-cooperative self-assembly (where the tile types are unit cubes) and 2D cooperative self-assembly share similar capabilities. For instance, Cook, Fu and Schweller \cite{CookFuSch11} proved that it is possible to deterministically simulate an arbitrary Turing machine using non-cooperative self-assembly, even if tiles are only allowed to be placed in the $z = 0$ and $z = 1$ planes (Winfree \cite{Winf98} proved this for the 2D aTAM). Cook, Fu and Schweller \cite{CookFuSch11} also proved that it is possible to deterministically self-assemble an $N \times N$ 3D ``square'' shape $S_N \subseteq \{0, \ldots, N - 1\} \times \{0, \ldots, N - 1\} \times \{0, 1\}$ using non-cooperative self-assembly with $O(\log N)$ tile complexity (Rothemund and Winfree \cite{RotWin00} proved this for the 2D aTAM). Furcy, Micka and Summers \cite{FurcyMickaSummers} reduced the tile complexity of deterministically assembling an $N \times N$ square in 3D non-cooperative self-assembly to $O\left(\frac{\log N}{\log \log N}\right)$ (Adleman, Cheng, Goel and Huang \cite{AdlemanCGH01} proved this for the 2D aTAM), which is optimal for all algorithmically random values of $N$. Given that it is possible to optimally self-assemble an $N \times N$ square in 3D using non-cooperative self-assembly, the following is a natural question: Is it possible to self-assemble an arbitrary finite shape in 3D using non-cooperative self-assembly with optimal tile complexity?

Note that the previous question was answered affirmatively by Soloveichik and Winfree \cite{SolWin07} for the 2D aTAM, assuming the shape of the final assembly can be a scaled-up version of the input shape (i.e., each point in the input shape is replaced by a $c \times c$ block of points, where $c$ is the \emph{scaling factor}). Specifically, Soloveichik and Winfree gave a construction that takes as input an algorithmic description of an arbitrary finite, connected shape $X \subset \mathbb{Z}^2$ and outputs a cooperative (temperature 2) tile set $T_X$ that deterministically self-assembles into a scaled-up version of $X$ and $\left| T_X \right| = O\left( \frac{|M|}{\log |M|}\right)$, where $|M|$ is the size of (i.e., number of bits needed to describe) the Turing machine $M$, which outputs the list of points in $X$. In the main result of this paper, using a combination of 3D, temperature 1 self-assembly techniques from Furcy, Micka and Summers \cite{FurcyMickaSummers} and Cook, Fu and Schweller \cite{CookFuSch11}, we show how the optimal construction of Soloveichik and Winfree can be simulated in 3D using non-cooperative self-assembly with optimal tile complexity. Thus, our main result represents a Turing-universal way of guiding the self-assembly of a scaled-up, just-barely 3D version of an arbitrary finite shape $X$ at temperature 1 with optimal tile complexity.

%% file: definitions.tex
\section{Definitions}\label{sec-definitions}
In this section, we give a brief sketch of a $3$-dimensional version of the aTAM along with some definitions of scaled finite shapes and the complexities thereof.

\subsection{3D abstract Tile Assembly Model}

Let $\Sigma$ be an alphabet. A $3$-dimensional \emph{tile type} is a tuple $t \in (\Sigma^* \times \N)^{6}$, e.g., a unit cube with six sides listed in some standardized order, each side having a \emph{glue} $g \in \Sigma^* \times \N$ consisting of a finite string \emph{label} and a non-negative integer \emph{strength}. In this paper, all glues have strength $1$.
There is a finite set $T$ of $3$-dimensional tile types but an infinite number of copies of each tile type, with each copy being referred to as a \emph{tile}.

A $3$-dimensional \emph{assembly} is a positioning of tiles on the
integer lattice $\Z^3$ and is described formally as a partial function
$\alpha:\Z^3 \dashrightarrow T$. Two adjacent tiles in an assembly
\emph{bind} if the glue labels on their abutting sides are equal and
have positive strength.  Each assembly induces a \emph{binding graph},
i.e., a ``grid graph'' (sometimes called the \emph{adjacency graph}) whose vertices are (positions of) tiles and whose
edges connect any two vertices whose corresponding tiles bind.  If
$\tau$ is an integer, we say that an assembly is
\emph{$\tau$-stable} if every cut of its binding graph has strength at
least~$\tau$, where the strength of a cut is the sum of all of the
individual glue strengths in~the~cut.

A $3$-dimensional \emph{tile assembly system} (TAS) is a triple $\calT
= \left(T,\sigma,\tau\right)$, where $T$ is a finite set of tile
types, $\sigma:\Z^3 \dashrightarrow T$ is a finite, $\tau$-stable
\emph{seed assembly}, and $\tau$ is the \emph{temperature}. In this
paper, we assume that $|\dom{\sigma}| = 1$ and $\tau=1$. An assembly
$\alpha$ is \emph{producible} if either $\alpha = \sigma$ or if
$\beta$ is a producible assembly and $\alpha$ can be obtained from
$\beta$ by the stable binding of a single tile.  In this case we write
$\beta\to_1^{\calT} \alpha$ (to mean~$\alpha$ is producible from
$\beta$ by the binding of one tile), and we write $\beta\to^{\calT}
\alpha$ if $\beta \to_1^{\calT^*} \alpha$ (to mean $\alpha$ is
producible from $\beta$ by the binding of zero or more tiles).
When~$\calT$ is clear from context, we may write $\to_1$ and $\to$
instead.  We let $\mathcal{A}\left[\mathcal{T}\right]$ denote the set
of producible assemblies of~$\calT$.  An assembly is \emph{terminal}
if no tile can be $\tau$-stably bound to it.  We
let~$\mathcal{A}_{\Box}\left[\mathcal{T}\right] \subseteq
\mathcal{A}\left[\mathcal{T}\right]$ denote the set of
producible, terminal assemblies of $\calT$.

A TAS~$\calT$ is \emph{directed} if $\left|\mathcal{A}_{\Box}\left[\calT\right]\right| = 1$. Hence, although a directed system may be nondeterministic in terms of the order of tile placements,  it is deterministic in the sense that exactly one terminal assembly is producible. For a set $X \subseteq \Z^3$, we say that $X$ is uniquely produced if there is a directed TAS $\mathcal{T}$, with $\mathcal{A}_{\Box}\left[\calT\right] = \{\alpha\}$, and $\dom{\alpha} = X$.

\subsection{Complexities of (scaled) finite shapes}

The following definitions are based on the definitions found in \cite{SolWin07}. We include these definitions for the sake of completeness.

A \emph{coordinated shape} is a finite set $X \subset \mathbb{Z}^2$ such that $X$ is connected, i.e., the grid graph induced by $X$ is connected. For some $c \in \mathbb{Z}^+$, we say that a $c$-\emph{scaling} of $X$, denoted as $X^c$, is the set $X^c = \left\{ (a,b) \mid \left(\lfloor a/c \rfloor, \lfloor b/c \rfloor\right) \in X \right\}$. Intuitively, $X^c$ is the coordinated shape obtained by taking $X$ and replacing each point in $X$ with a $c \times c$ block of points. Here, the constant $c$ is known as the \emph{scale factor} (or \emph{resolution loss}). Note that a $c$-scaling of an actual shape is itself a coordinated shape.

Let $X_1$ and $X_2$ be two coordinated shapes. We say that $X_1$ and $X_2$ are \emph{scale-equivalent} if $X^a_1 = X^b_2$, for some $a,b\in\mathbb{Z}^+$. We say that $X_1$ and $X_2$ are \emph{translation-equivalent} if they are equal up to translation. We write $X^a_1 \cong X^b_2$ if $X^a_1$ is translation-equivalent to $X^b_2$, for some $a,b\in\mathbb{Z}^+$. Note that the three previously defined relations are all equivalence relations (see the appendix of \cite{SolWin07}). We will use the notation $\widetilde{X}$ to denote the equivalence class containing $X$ under the equivalence relation $\cong$. We say that $\widetilde{X}$ is the \emph{shape} of $X$. While $\widetilde{X}$ is technically a set of coordinate shapes, we will abuse the notation $\left|\widetilde{X}\right|$ and say that it represents the size of coordinate shape $X \in \widetilde{X}$, i.e., $\left|X^1\right|$.

We will now define the tile complexity of a 3D shape. However, we will first briefly define Kolmogorov complexity of a binary string $x$, relative to a universal Turing machine $U$. We say that the \emph{Kolmogorov complexity} of $x$ relative to $U$ is $K_U(x) = \min\left\{|p| \mid U(p) = x\right\}$, where, for any Turing machine $M$, $|M|$ denotes the number of bits used to describe $M$, with respect to some fixed encoding scheme. In other words, $K_U(x)$ is the smallest program that outputs $x$ (see \cite{Li:1997:IKC} for a comprehensive discussion of Kolmogorov complexity).

Relative to a fixed universal Turing machine $U$, we say that the \emph{Kolmogorov complexity of a shape} $\widetilde{X}$ is the size of the smallest program that outputs some $X \in \widetilde{X}$ as a list of locations, i.e., $K_U\left(\widetilde{X}\right) = \min\left\{|p| \; \left| \;U(p) = \langle X \rangle\ \textmd{ for some } X \in \widetilde{X}\right.\right\}$. Beyond this point, we will assume $U$ is a fixed universal Turing machine and therefore will be omitted from our notation.

The \emph{3D tile complexity} at temperature $\tau$ (often referred to as \emph{program-size complexity}) of a shape $\widetilde{X}$ at temperature $\tau$ is

$$K^{\tau}_{3DSA}\left(\widetilde{X}\right) = \min\setr{n}{
\begin{split}
    \textmd{$\mathcal{T} = (T,\sigma,\tau)$, $|T|=n$ and there exists\;\;\;} \\
    \textmd{$X \in \widetilde{X}$ such that $\mathcal{T}$ uniquely produces $\alpha$}\\
    \textmd{such that $X \times \{0\} \subseteq \dom{\alpha} \subseteq X \times \{0,1\}$}
\end{split}}.$$

%% file: main_theorem.tex
\section{Main theorem}
\label{sec:main_result}

The main theorem of this paper describes the relationship between the quantities $K\left(\widetilde{X}\right)$ and $K^1_{3DSA}\left(\widetilde{X}\right)$. This relationship is formally stated in Theorem~\ref{thm:main-theorem}. Note that the main result of \cite{SolWin07} describes the relationship between $K\left(\widetilde{X}\right)$ and $K^2_{SA}\left(\widetilde{X}\right)$, where $K^2_{SA}\left(\widetilde{X}\right)$ (see \cite{SolWin07}) is the tile complexity of the 2D shape $\widetilde{X}$ at temperature 2. In this section, assume that $\widetilde{X}$ is an arbitrary finite shape.

\begin{theorem}
\label{thm:main-theorem}
The following hold: 
\begin{enumerate}
    \item $K\left(\tilde{X}\right) = O\left(K^{1}_{3DSA}\left(\widetilde{X}\right)\log K^{1}_{3DSA}\left(\widetilde{X}\right)\right)$
    \item $K^{1}_{3DSA}\left(\widetilde{X}\right)\log K^{1}_{3DSA}\left(\widetilde{X}\right) = O\left(K\left(\widetilde{X}\right)\right)$
\end{enumerate}
\end{theorem}

We will prove Theorem~\ref{thm:main-theorem} in Lemmas~\ref{lem:main-lower} and~\ref{lem:main-upper}.

\begin{lemma}
\label{lem:main-lower}
$K\left(\tilde{X}\right) = O\left(K^{1}_{3DSA}\left(\widetilde{X}\right)\log K^{\tau}_{3DSA}\left(\widetilde{X}\right)\right)$.
\end{lemma}

\begin{proof}
Soloveichik and Winfree \cite{SolWin07} showed that $K\left(\widetilde{X}\right) = O\left(K^{\tau}_{SA}\left(\widetilde{X}\right)\log K^{\tau}_{SA}\left(\widetilde{X}\right)\right)$, which still holds when $K^{\tau}_{SA}$ is replaced with $K^{\tau}_{3DSA}$, for any $\tau \in \mathbb{Z}^+$.
\end{proof}

The main contribution of this paper is the following lemma, the proof of which mimics the proof of $K^2_{SA}\left( \widetilde{X}\right) \log K^2_{SA}\left( \widetilde{X}\right) = O\left(K\left(\widetilde{X}\right)\right)$ from \cite{SolWin07}. We give the details of the proof here for the sake of completeness.

\begin{lemma}
\label{lem:main-upper}
$K^{1}_{3DSA}\left(\widetilde{X}\right)\log K^{1}_{3DSA}\left(\widetilde{X}\right) = O\left(K\left(\widetilde{X}\right)\right)$.
\end{lemma}

\begin{proof}
Let $s$ be a program (Turing machine) that outputs a list of points contained in some coordinated shape $X \in \widetilde{X}$. We develop a temperature 1 3D construction that takes $s$ as input and outputs a TAS $\mathcal{T}_{\widetilde{X}} = \left(T_{\widetilde{X}},\sigma,1\right)$ such that $\mathcal{T}_{\widetilde{X}}$ uniquely produces an assembly whose domain is some shape $X \in \widetilde{X}$ and $\left|T_{\widetilde{X}}\right| = O\left(\frac{\left|s\right|}{\log \left|s\right|}\right)$. This construction is discussed in Section~\ref{sec:main-construction}.

Now suppose that $s$ is the smallest Turing machine that outputs the list of points in some coordinated shape $X \in \widetilde{X}$. In other words, $s$ is such that $K\left( \widetilde{X} \right) = \left| s \right|$. Let $\mathcal{T}^*_{\widetilde{X}} = \left(T^*_{\widetilde{X}},\sigma,1\right)$ be the TAS produced by our construction, when given $s$ as input. Observe that, for any TAS $\mathcal{T}=(T,\sigma,1)$ in which the shape $\widetilde{X}$ uniquely self-assembles, we have $K^1_{3DSA}\left(\widetilde{X}\right) \leq \left|T\right|$. Then, for some constant $c \in \mathbb{Z}^+$, the following is true:
\begin{eqnarray*}
   K^1_{3DSA}\left(\widetilde{X}\right) \log K^1_{3DSA}\left(\widetilde{X}\right) & \leq & \left|T^*_{\widetilde{X}}\right| \log\left|T^*_{\widetilde{X}}\right| \leq c\frac{\left|s\right|}{\log\left|s\right|} \log \frac{\left|s\right|}{\log\left|s\right|} \\
   & = & c\frac{\left|s\right|}{\log\left|s\right|}\left( \log \left|s\right| - \log \log \left|s\right| \right) \leq c\left|s\right|.
\end{eqnarray*}
Thus, $K^1_{3DSA}\left(\widetilde{X}\right) \log K^1_{3DSA}\left(\widetilde{X}\right) = O\left(K\left(\widetilde{X}\right)\right)$.
\end{proof}

%% file: construction.tex
\section{Main construction}
\label{sec:main-construction}
In this section, we give an overview of our main construction. Full details can be found in Sections~\ref{sec:Turing_machine_simulation} and~\ref{sec:appendix}.

\subsection{Setup}
\label{sec:setup}
Our construction represents a Turing-universal way of guiding the self-assembly of a scaled-up, just-barely 3D version of an arbitrary input shape $X$ at temperature 1 with optimal tile complexity. Therefore, we let $U$ be a fixed universal Turing machine over a binary alphabet, with a one-way infinite tape (to the right), such that, upon termination, $U$ contains the output -- and only the output -- of the Turing machine being simulated on its tape, perhaps padded to the left and right with $0$ bits (the tape alphabet symbol $0$) and the tape head is reading the last bit of its output. We also assume that, to the left of the leftmost tape cell, there is a special left marker symbol \#, which can be read by $U$ but can neither be overwritten nor written elsewhere on the tape. In general, if $M$ is a Turing machine and $w$ is an input string such that $M(w) = y$, then, upon termination, $U(\langle M,w\rangle)$ leaves exactly a string of the form $\#0^*y0^*$ on its tape with its tape head reading the last bit of $y$ (\# is eventually converted to a 0 bit). Our construction is programmed by specifying the program to be executed by $U$.

\begin{figure}[htp]%
\vspace{-10pt}
\centering
    \subfloat[][The input shape $X$.]{%
        \label{fig:input_shape}%
        \includegraphics[width=.5in]{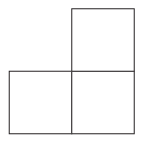}}%
        \hspace{10pt}%
    \subfloat[][$X^2$]{%
        \label{fig:input_scaled}%
        \includegraphics[width=.5in]{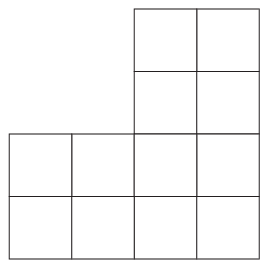}}
        \hspace{10pt}
    \subfloat[][Put a wicket in each $2 \times 2$ block.]{%
        \label{fig:input_scaled_spanning_tree}%
        \includegraphics[width=.75in]{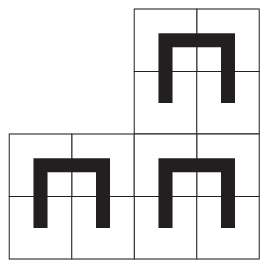}}
        \hspace{10pt}
    \subfloat[][Connect the wickets to get a spanning tree of $X^2$.]{%
        \label{fig:input_scaled_spanning_tree}%
        \includegraphics[width=.75in]{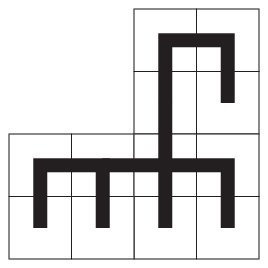}}
        \hspace{10pt}
    \subfloat[][Do a modified depth-first search to get a Hamiltonian cycle of $X^4$ that contains three consecutive, collinear points.]{%
        \label{fig:nice_ham_cycle}%
        \includegraphics[width=1.0in]{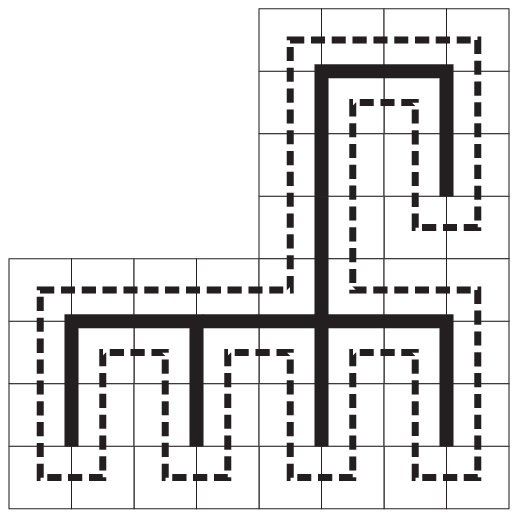}}
    \caption{\label{fig:input_shape_overview} An overview of the algorithm for plotting out a Hamiltonian cycle of an input shape (scaled up by a factor of $4$), such that the cycle has three consecutive, collinear points. Although a Hamiltonian cycle is hard to compute, in general, it is clear that, for an arbitrary finite shape $X$, a Hamiltonian cycle, as described above, of $X^4$, can be computed in polynomial time. }
\end{figure}

The input to $U$ is a program $p$, appropriately encoded as $\langle p \rangle$, using some fixed encoding scheme. The output of $p$ is used to guide the self-assembly of our construction. We assume $p$ is actually the concatenation of two programs $s$ and $p_{hc}$, where $s$, the input to our construction, is a program that outputs the list of points in $X$ and $p_{hc}$ is a fixed program (independent of $X$) that uses the output of $s$ as its input and builds a special Hamiltonian cycle $H$ of $X^4$, along which there are three consecutive, collinear points. The \emph{seed block}, which is described in the next subsection, is defined as the middle point of an arbitrarily chosen triplet of consecutive, collinear points of $H$ (by the way we construct $H$, there is always at least one such triplet of points). We further assume that $p$ outputs $H$ -- and only $H$ -- as a sequence of pairs of bits, such that, $00$, $10$, $01$ and $11$ correspond to ``no-move'', ``left'', ``right'' and ``straight'', respectively and possibly padded to the left with at least two no-moves and to the right with an even number of no-moves (see Figure~\ref{fig:input_shape_overview} for an overview of how $H$ is constructed). Thus, $p$ satisfies $|p| = |s| + \left|p_{hc}\right|$.

\subsection{Seed block}
\label{sec:seed_block}
The (upper portion of what will eventually become the) seed block of our construction grows from a single seed tile and carries out the following three logical phases: decoding, simulation and output. These three phases are depicted in Figure~\ref{fig:seed_block} with vertical, zig-zag and diagonal patterns, respectively.

\begin{figure}[htp]
\begin{minipage}{\textwidth}
\centering
 \includegraphics[width=\textwidth]{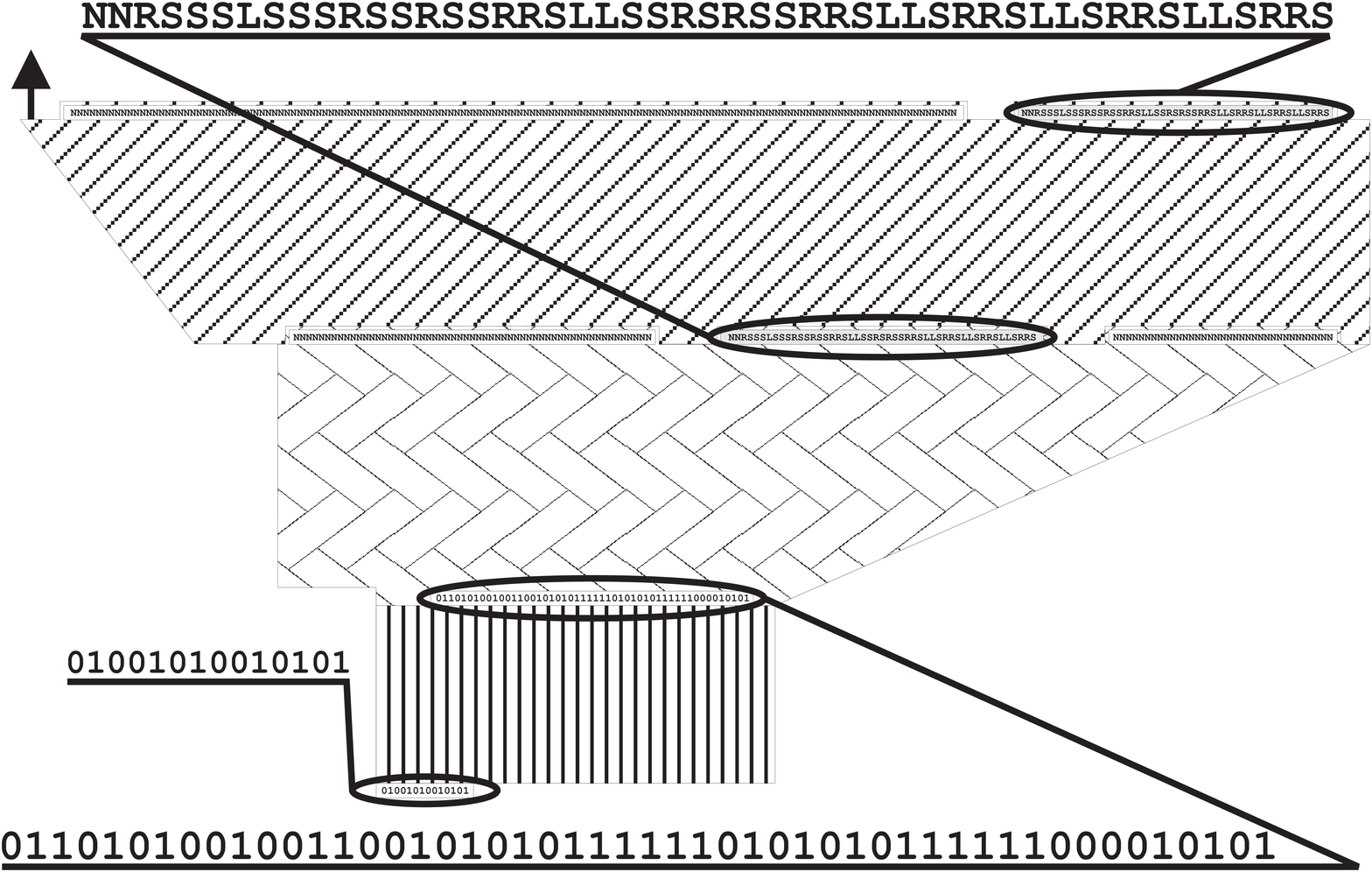}
\end{minipage}
\caption{Self-assembly of the upper portion of the seed block consists of three logical phases. In the first phase (the region filled with vertical lines), the bits of $p$ are decoded using the 3D, temperature 1 optimal encoding scheme of Furcy, Micka and Summers \cite{FurcyMickaSummers} (the encoded bits of $p$ are depicted as the shorter binary string). The decoded bits of $p$ (the longer binary string) are input to a fixed universal Turing machine $U$. Then, in the second phase, the simulation of $p$ on $U$ is carried out (in the region with the zig-zag pattern). We require that $U(\langle p\rangle)$ evaluates precisely to the sequence of moves in the Hamiltonian cycle of $X^4$, padded to the left and right with an even number of $0$ bits (the boxes that are not encircled in this figure). In other words, we require that $U(\langle p \rangle)$ evaluates to a string of the form $(00)^*(00|10|01|11)^*(00)^*$, with the tape head of $U$ reading the second bit in the last move of the Hamiltonian cycle. Finally, in the third phase (in the region with the diagonal line pattern), the moves in the Hamiltonian cycle are shifted to the right. Self-assembly of the first growth block begins from the upward-pointing arrow. Note that the moves in the Hamiltonian cycle are listed in the grey boxes and we use the characters `N', `L', `R' and `S' to represent ``no-move'', ``left'', ``right'' and ``straight'', respectively.}
\label{fig:seed_block}
\vspace{-10pt}
\end{figure}

\textbf{Decoding}. The first phase is the decoding phase. In the decoding phase, the bits of $\langle p \rangle$ are decoded from a $O(\log |\langle p \rangle|)$-bits-per-tile representation to a $1$-bit-per-tile representation (actually, we end up with a $1$-bit-per-gadget representation, which is sufficient to maintain the optimality of our construction). To accomplish this, we use the 3D, temperature 1 optimal encoding scheme of Furcy, Micka and Summers \cite{FurcyMickaSummers}. When the decoding phase completes, the decoded bits of $\langle p \rangle$ are advertised in a one-bit-per-gadget representation along the top of the optimal encoding region (the rectangle with the vertical lines in Figure~\ref{fig:seed_block}).

\textbf{Simulation}. Once the bits of $\langle p \rangle$ are decoded, the simulation phase begins. In this phase, $p$ is simulated on $U$ using a specialized temperature 1, just-barely 3D Turing machine simulation (the region with the zig-zag pattern in Figure~\ref{fig:seed_block}). Our specialized Turing machine simulation, which is described in detail in Section~\ref{sec:Turing_machine_simulation}, assumes an input Turing machine $M$ with (1) a binary alphabet, (2) a one-way infinite tape (to the right), the leftmost tape cell of which contains a special left marker symbol \#, which can be read by $M$ but can neither be overwritten nor written elsewhere on the tape. We simulate $M$ in a zig-zag fashion, similar to the temperature 1, just-barely 3D Turing machine simulation by Cook, Fu and Schweller \cite{CookFuSch11}. However, unlike that of Cook, Fu and Schweller, our simulation represents the contents of each tape cell of $M$ using a six-tile-wide gadget. This gives a more compact geometric representation of the output of $M$ and, as a result, simplifies the construction of the \emph{growth blocks} (see Section~\ref{sec:growth_block}).

By the definition of $U$ and $p$ and because of the compact geometry of our simulation of $p$ on $U$, the output of the simulation phase, i.e., $U(\langle p\rangle)$, is an even number of geometrically-encoded bits (each bit is represented by a six-tile-wide bit-bump gadget), possibly padded to the left and right with an even number of $0$ bits, such that each pair of bits corresponds to a move in the Hamiltonian cycle $H$ (not counting occurrences of the pair $00$, which represents a no-move).

\textbf{Output}. In order to satisfy certain geometric constraints, which are required by the growth blocks, after the simulation phase of $p$ on $U$ is complete, the (final) output phase begins. In the output phase, we use a special, constant-size tile set to shift the geometrically-encoded bits of $H$ to the right, so that the bits of $H$ are in a right-justified position along the top of the seed block (the region with diagonal lines in Figure~\ref{fig:seed_block}). For each right-shift, we add a pair of $0$ bits to the left, which ensures that the upper portion of the seed block will be wider than it is taller. After the output phase of the seed block, self-assembly of the first growth block, which is always to the north in our construction, as guaranteed by $p$, begins from the left side of the top of the upper portion of the seed block (see the upward-pointing arrow in Figure~\ref{fig:seed_block}).

\textbf{Scale factor}. Let $W_{decode}$ and $H_{decode}$ be the maximum horizontal and vertical extent, respectively, of the seed block after the decoding phase (the rectangle with vertical lines pattern in Figure~\ref{fig:seed_block}) completes. From \cite{FurcyMickaSummers}, we know that $W_{decode} > H_{decode}$. Next, in the simulation phase, the tape grows to the right by two tape cells for each transition. Each transition is comprised of two rows of gadgets, which are wider than they are tall (six tiles versus four tiles). Also, we may assume that, during the simulation phase, $p$ is programmed to initially scan the input from left-to-right and then from right-to-left before beginning. Let $W_{sim}$ and $H_{sim}$ be the maximum horizontal and vertical extent, respectively, of the seed block after the simulation phase (the zig-zag pattern in Figure~\ref{fig:seed_block}) completes. Then we have $W_{sim} > H_{sim}$. Finally, in the output phase, as the output bits of the simulation phase are shifted to the right, two tape cells are added to the left for each shift. Each shift is comprised of two rows of gadgets, which, like the simulation gadgets, are wider than they are tall (six versus four). Therefore, let $W_{sb}$ and $H_{sb}$ be the maximum horizontal and vertical extent, respectively, of the seed block after the simulation phase (the diagonal pattern in Figure~\ref{fig:seed_block}) completes. Then we have $W_{sb} > H_{sb}$. From this we may conclude that the seed block, once completely filled in by the last growth block (see Figure~\ref{fig:putting_it_all_together_big}) will be a square. The scale factor of our construction is $W_{sb}$.

\subsection{Growth blocks}
\label{sec:growth_block}
\begin{figure}[htp]
\begin{minipage}{\textwidth}
\centering
 \includegraphics[width=\textwidth]{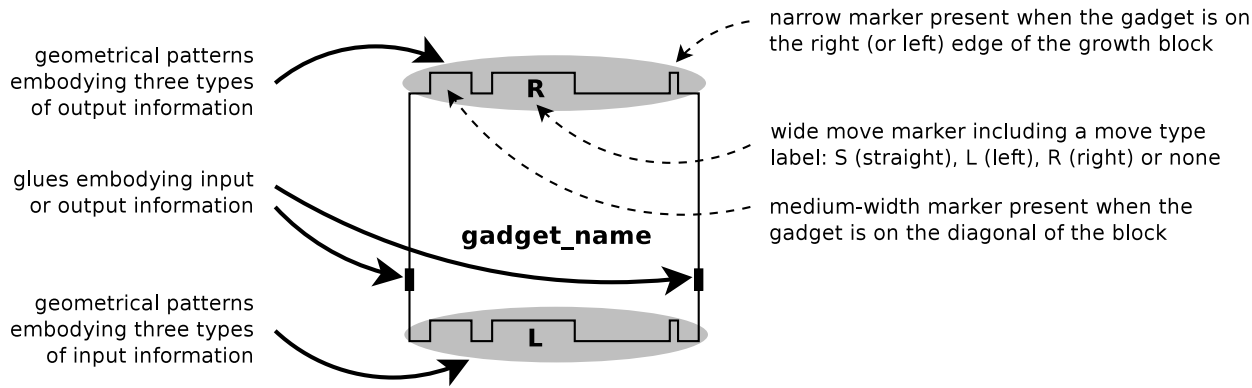}
\end{minipage}
\caption{Key to the representation of our gadgets}
\label{fig:generic_gadget}
\vspace{-10pt}
\end{figure}
Each growth block has a single input side, which reads the remaining
moves in the Hamiltonian cycle and a single output side, which
advertises the same remaining path but with its first move
erased. This first move determines the position of the output side in
relation to the input side. In this section, we assume that the input
side of the growth block is its south side (the construction simply
needs to be rotated for the three other possible positions of the
input side). So, if the first (erased) move in the remaining path is a
right turn, then the output side of the growth block is its east
side. We describe the construction for this case here (see
Figure~\ref{fig:Rgrowth_construction}). The ``left turn'' and
``straight move'' cases are described in Section~\ref{sec:appendix}. The overview
figure for the growth blocks uses gadgets whose structure is explained
in Figure~\ref{fig:generic_gadget}. The constructions of the individual
gadgets are described in the appendix.
\begin{figure}[htp]
\begin{minipage}{\textwidth}
\centering
 \includegraphics[width=\textwidth]{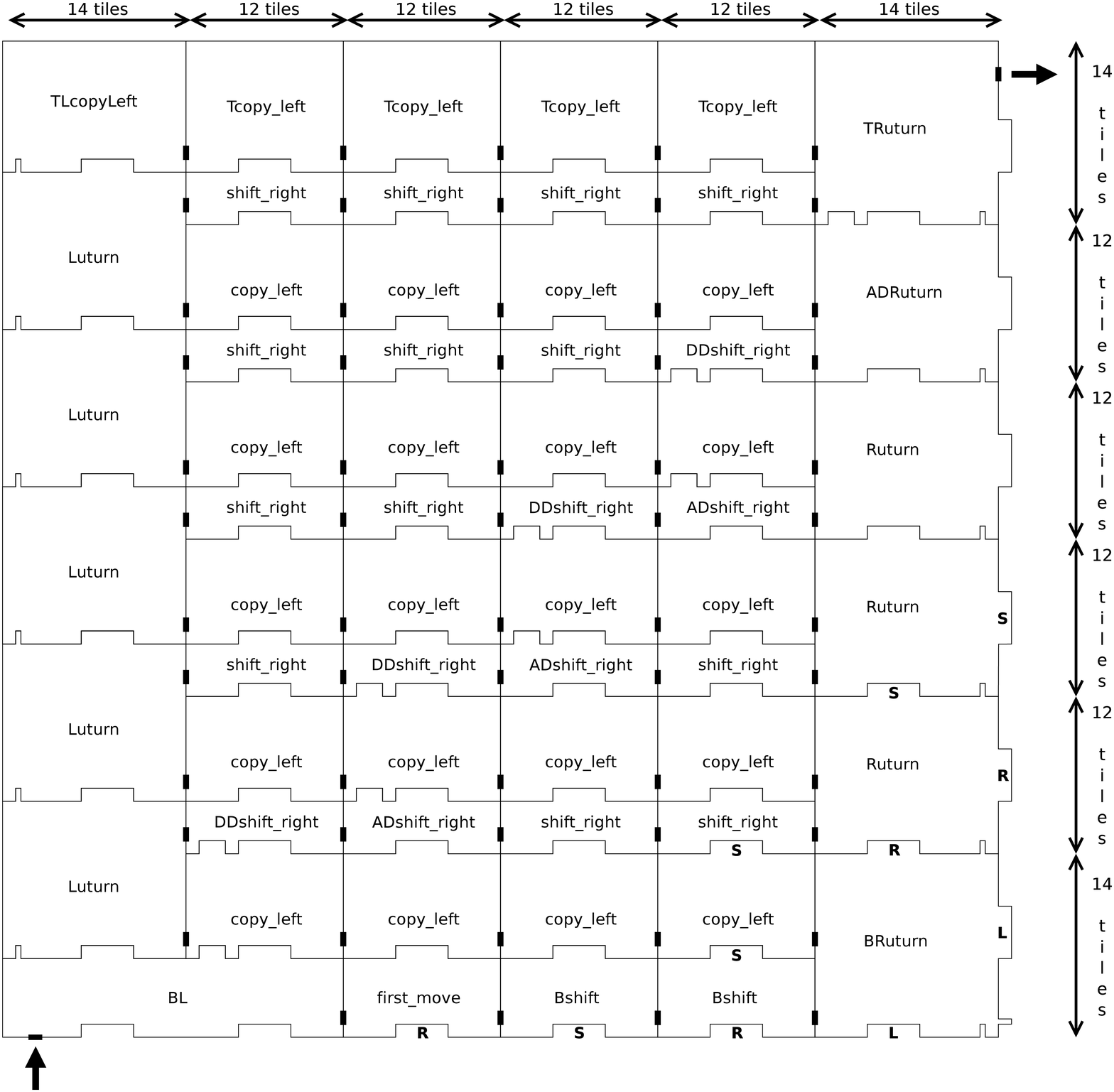}
\end{minipage}
\caption{Overall construction of a growth block whose input path
  starts with a right turn}
\label{fig:Rgrowth_construction}
\vspace{0pt}
\end{figure}

The growth block starts assembling in its southwest corner and
progresses in a zig-zag pattern. The first row of gadgets, moving from
left to right, starts by copying all of the leading no-moves, of which
there are exactly two in Figure~\ref{fig:Rgrowth_construction} but
generally many more. Once the first actual move is found, a set of
gadgets specific to its type is activated. In the case of a right
turn, all of the moves are shifted by one position to the right (and
the first move is replaced by a no-move). The last move in the
remaining path is advertised at the bottom of the output (or east)
side of the block. Then the construction switches direction and moves
from right to left, simply copying the shifted path, which completes
the first iteration. In each subsequent iteration, the left-to-right
pass shifts the whole path to the right by one position and advertises
one more move on the output (or east) side. In addition to the right
shift, each zig-zag iteration moves a diagonal marker by one position
to the right, starting from the southwest corner. Once this diagonal
marker reaches the east side of the block, the top row, moving from
right to left, can complete the block. Note that in this case, the
remaining path is not advertised on the west nor the north sides. If
the first move in the remaining path were a straight move, the
remaining moves would not be shifted but simply copied at each
iteration and eventually advertised on the north side of the block. If
the first move in the remaining path were a left turn, then the
remaining moves would be shifted to the left and advertised on the
west side instead. In other words, for a straight move, Figure~\ref{fig:Rgrowth_construction} would have smooth east and west sides, with bit-bumps along its north side and for a left move, Figure~\ref{fig:Rgrowth_construction} would have smooth east and north sides, with bit-bumps along its west side. Section~\ref{sec:appendix} describes these two cases.


\subsection{Putting it all together}
\begin{figure}[htp]%
\centering
    \subfloat[][In the upper portion (pattern-filed regions) of the seed block, the points in $X$ are decoded and a special Hamiltonian cycle $H$ of $X^4$. As moves of $H$ are carried out within each growth block, the bits of $H$ are propagated to the next growth block but the move that was just executed is erased (depicted as dashes). The last growth block assembles the remaining portion of the seed block with a sequence of single-tile-wide paths that assemble northward until running into the existing portion of the seed block.]{%
        \label{fig:putting_it_all_together_big}%
        \includegraphics[width=.7\textwidth]{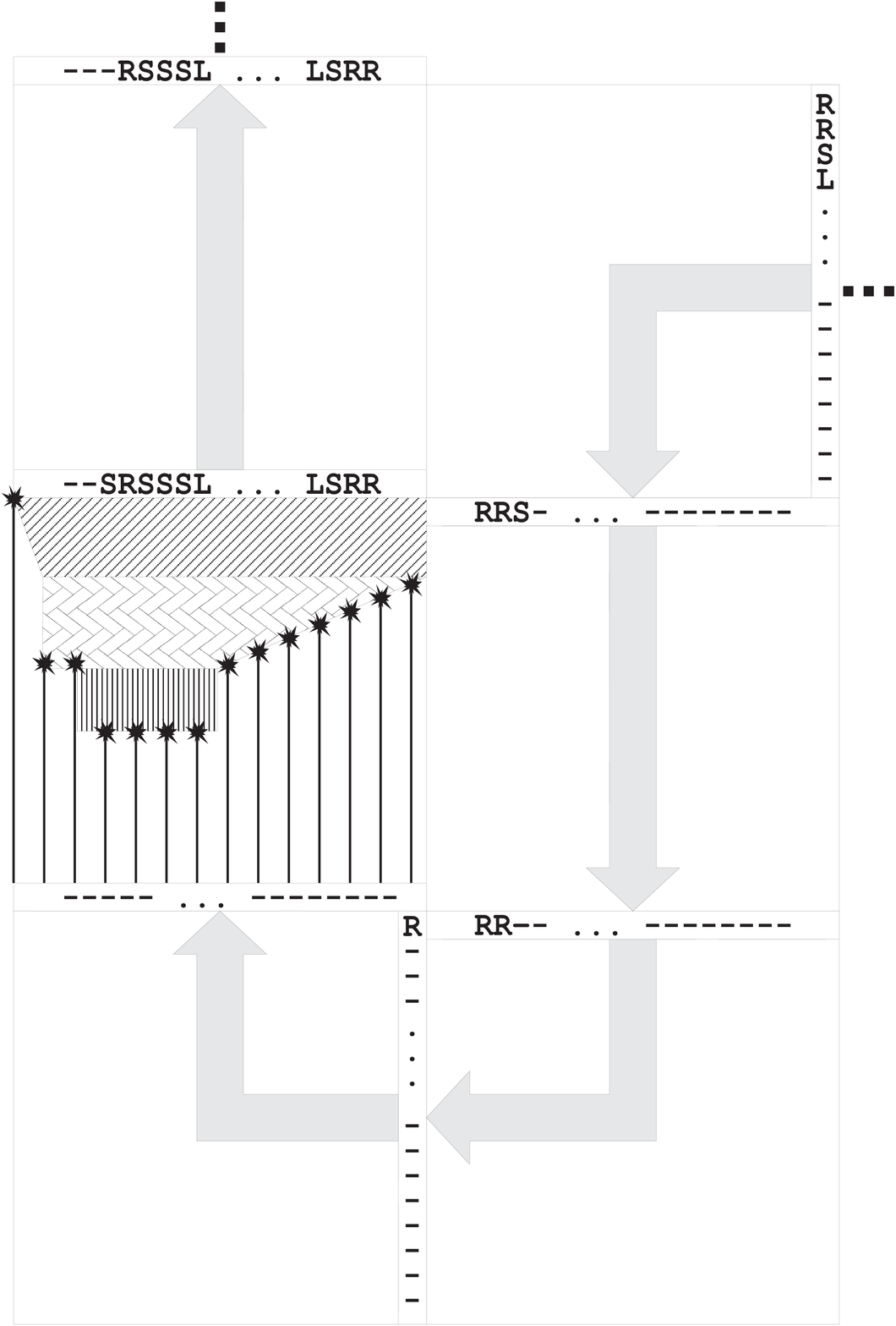}}%
        \hspace{5pt}%
    \subfloat[][The grey squares indicate the portion being assembled in Figure~\ref{fig:putting_it_all_together_big}. The darker square is the seed block. Self-assembly always proceeds north from the seed block.]{%
        \label{fig:putting_it_all_together_small}%
        \includegraphics[width=.2\textwidth]{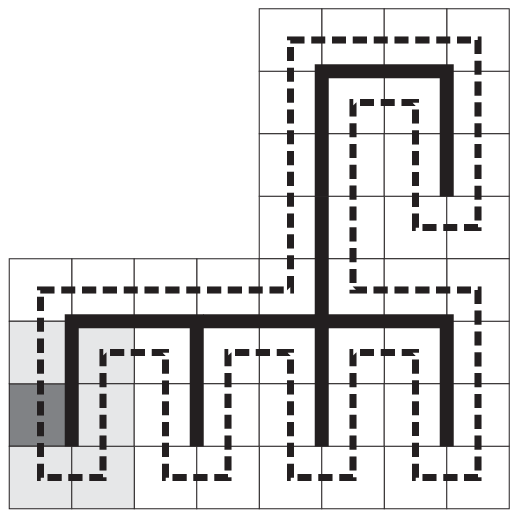}}
        \hspace{5pt}
    \caption{\label{fig:putting_it_all_together} Putting it all together. In this figure, we depict the moves in the Hamiltonian cycle of $X^4$ as `-', `L', `R' and `S' for ``no-move'', ``left'', ``right'' and ``straight'', respectively. In our construction, these moves are represented using the pairs of bits 00, 10, 01 and 11, for ``no-move'', ``left'', ``right'' and ``straight'', respectively.}
\end{figure}
After the final growth block completes, the remaining portion of the seed block, i.e., its lower portion, is assembled. Note that, up until this point, the seed block is not a $c \times c$ square. However, the horizontal extent of the upper portion of the seed block defines the scale factor $c$ of our construction. This scale factor is dominated by running time of $p$ on $U$, which is the sum of the running times of $s$ and $p_{hc}$. The final growth block fills in the remaining portion of the seed block by initiating the assembly of a sequence of $c$ single-tile-wide, vertically and uncontrollably assembling paths that are inhibited only by existing portions of the seed block (see the explosion icons in Figure~\ref{fig:putting_it_all_together_big}). Thus, the final, uniquely-produced terminal assembly of our construction is an assembly made up of $c \times c$ blocks of tiles, where each block is mapped to some point in $X$. Figure~\ref{fig:putting_it_all_together} gives a high-level overview of how all of the major components of our construction work together.

%% file: conclusion.tex
\section{Conclusion}
In this paper, we develop a Turing-universal way of guiding the self-assembly of a scaled-up, just-barely 3D version of an arbitrary input shape $X$ at temperature 1 with optimal tile complexity. This result is essentially a just-barely 3D temperature 1 simulation of a similar 2D temperature 2 result by Soloveichik and Winfree \cite{SolWin07}. One possibility for future research is to resolve the tile complexity of an arbitrary shape $\widetilde{X}$ at temperature 1 in 2D, i.e., what is the quantity $K^1_{SA}\left(\widetilde{X}\right)$?

%% file: Turing_machine_simulation.tex
\clearpage
\section{Appendix: Turing machine simulation}
\label{sec:Turing_machine_simulation}
In this section, we describe our Turing machine simulation.

\subsection{Notation for figures}

In the figures of the gadgets, we use big squares to represent tiles placed in the $z=0$ plane and small squares to represent tiles placed in the $z=1$ plane. A glue between a $z=0$ tile and $z=1$ tile is denoted as a small black disk. Glues between $z=0$ tiles are denoted as thick lines. Glues between $z=1$ tiles are denoted as thin lines. This convention is also used in Section~\ref{sec:appendix}.

In what follows, we assume that the glues on each tile are implicitly defined to ensure deterministic assembly.

\subsection{Overview}
\begin{figure}[htp]
\centering
\includegraphics[width=\textwidth]{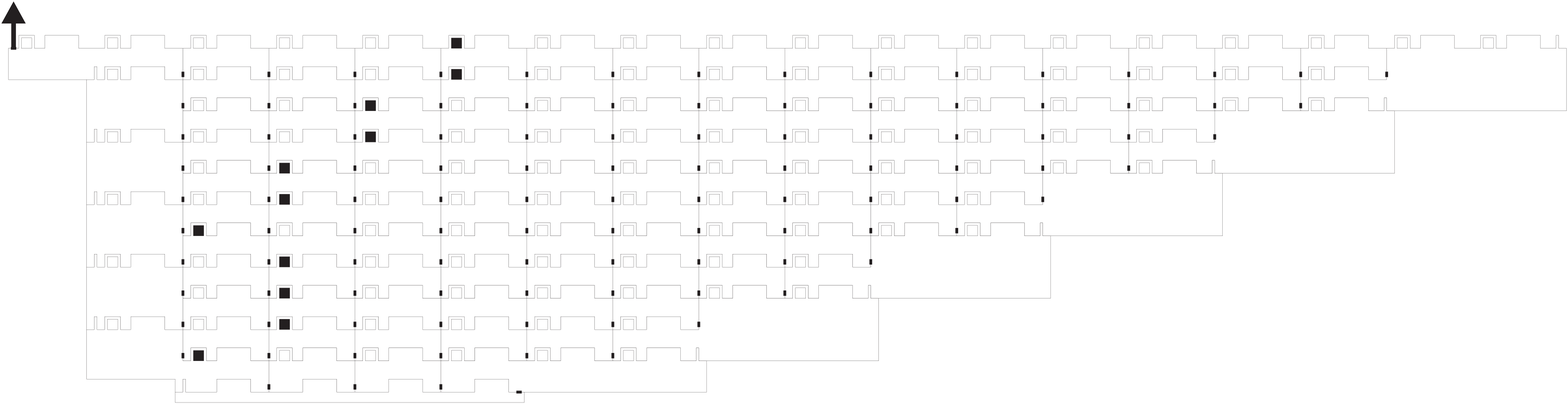}
\caption{Our simulation assembles the configuration history of the computation of a Turing machine $M$ on input $w$. Tape cells are represented as pairs of bits. Most gadgets represent one tape cell, but some gadgets represent two or three tape cells. In this, and other high-level overview figures, a wide bump on a gadget indicates a binary value for a tape cell, a medium bump with a black (white) square indicates the tape head is (not) reading that tape cell and a very narrow bump indicates the left or right boundary. As the Turing machine is being simulated, the tape is expanded to the right by two blank tape cells for each transition. At the end of the simulation, either one or two blank cells are added to the left, depending on the parity of the length of the input (in this example, two extra tape cells are added to the left because there are an even number of input bits). The north-facing glue of the last tile placed is indicated by the upward pointing arrow. Note that the post-processing phase is not shown (but it is shown in Figure~\ref{fig:overview_post_process}).}
\label{fig:overview}
\vspace{0pt}
\end{figure}
Our Turing machine simulation assumes an input Turing machine $M = \left(Q,\{0,1\},\{0,1\},\delta,q_0,q_{halt}\right)$, such that, $M$ has a one-way infinite tape (to the right), the leftmost tape cell of which contains a special left marker symbol \#, which can be read by $M$ but can neither be overwritten nor written elsewhere on the tape (the left marker symbol is added by our construction so that it is left of the leftmost input bit). We further require that $M$ begins with its tape head reading the leftmost bit of the input word $w \in \{0,1\}^*$ and $M$ halts with its tape head reading the rightmost bit of output.

We simulate $M$ on $w$ in a zig-zag fashion, similar to Cook, Fu and Schweller \cite{CooFuSch11}, where each transition in our simulation is carried out by a pair of rows of gadgets. First, a row of gadgets assembles left-to-right, followed by a row of gadgets that assembles right-to-left. Right-moving transitions are performed in the former whereas the latter handles left-moving transitions.

Our simulation differs from that of Cook, Fu and Schweller in the following way: we represent the contents of each tape cell with one six-tile-wide gadget (Cook, Fu and Schweller represent each tape cell with a non-constant number of gadgets, depending on the size of the transition function of $M$). Consequently, in our construction, and unlike in Cook, Fu and Schweller, the transition function of $M$ is stored in the glues of the tiles. Also, since the tape may grow unbounded, we must store the position of the tape head geometrically. Therefore, in our construction, each tape cell gadget, in addition to storing its binary value using a bit-bump, also stores whether or not the tape head is reading that cell in a similar kind of bit-bump (the tape head is either reading this cell or not). This approach has a more compact geometric representation of the output, which is more convenient for our purposes.

The first phase in the simulation is \emph{preprocessing}. In the preprocessing phase, each input bit is augmented with a bit indicating whether or not the tape head is reading that bit, where 0 (1) represents the absence (presence) of the tape head. The preprocessing gadgets assemble from right-to-left. After the preprocessing phase completes, the \emph{transition} phase begins.

During the transition phase, the computation of $M$ on $w$ is simulated according to the transition function of $M$, which is encoded into the glues of our simulation. The gadgets that implement the transition phase automatically grow the tape to the right by two tape cells in each time step. The tape is expanded to the left by one tape cell in the preprocessing phase of the simulation. Then, in the last step of the transition phase, one additional tape cell is added to the left if there is an even number of input bits.

The last phase of the simulation is \emph{post-processing}, in which the output bits are shifted to the right. In this phase, the tape does not expand to the right, but is expanded to the left by two tape cells for each shift. Note that, in this phase, each shift is carried out in two steps. First, a row of gadgets, assembling left-to-right, shifting all the tape cells over to the right by one (gadget) position. Then, a row of gadgets, assembling right-to-left, copying the (binary values of the) just-shifted tape cells up and, after growing the tape to the left by two tape cells, restarts another left-to-right-assembling shift row. The shifting is complete when the geometrically-encoded bit that represents the position of the tape head reaches the rightmost tape cell. The last row of post-processing gadgets assembles from right-to-left and erases the tape-head-indicator, leaving only the binary value of the tape cell.

When the post-processing phase completes, thus completing the simulation of $M$ on $w$, the number of bits along the top of the simulation assembly is even. This is because: (1) the preprocessing phase computes the parity of the length of the input, which causes an appropriate number of tape cells to be added to the left immediately before the output phase begins, (2) the transition phase grows the tape to the right by two tape cells for each transition, and (3) the post-processing phase grows the tape to the left by two tape cells for every shift operation. Moreover, because of the way the tape is expanded in the transition and post-processing phases, the completed simulation of $M$ on $w$ is wider than it is tall.

We use the following notation: For $b,c \in \{0,1\}$ and $p,q\in Q - \{q_{halt}\}$, $move(p,b) = L$ if $\delta(p,b) = (q,c,L)$ and $move(p,b) = R$ if $\delta(p,b) = (q,c,R)$, i.e., $move(p,b)$ gives the direction $M$ will move in when in state $p$ reading bit $b$.


\subsection{Input}
\begin{figure}[htp]
\centering
\includegraphics[width=\textwidth]{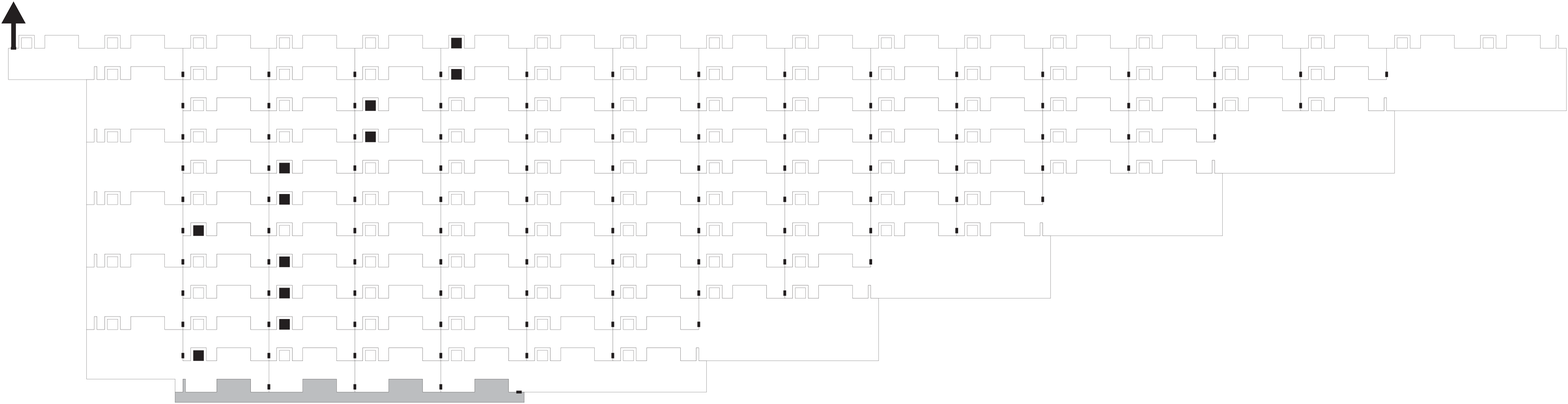}
\caption{The input bits are indicated with dark grey.}
\label{fig:overview_input}
\vspace{0pt}
\end{figure}

The input to the simulation is a sequence of appropriately-spaced, geometrically-encoded bits (bit-bumps), where a bump in the $z=0$ ($z=1$) plane represents the bit 0 (1). Figure~\ref{fig:Turing_machine_input} gives an example of an input sequence.

\begin{figure}[htp]
\begin{minipage}{\textwidth}
\centering
\fbox{
 \includegraphics[width=0.85\textwidth]{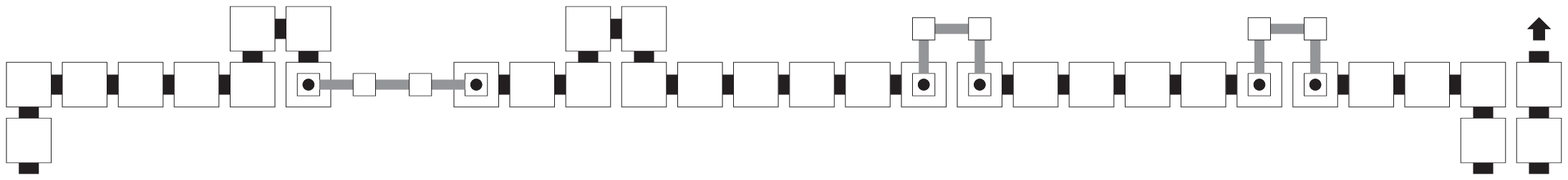}
}
\end{minipage}

\caption{An example of a 4-bit input to the Turing machine simulation, which matches the geometry of the output of the optimal encoding region from \cite{FurcyMickaSummers}. The bit string depicted here is $w=0011$. Each bit is encoded with a gadget that is six tiles wide and two tiles tall (with a bump in the two middle tiles). Observe that the leftmost bit is specially marked by the two missing $z=0$ tiles (seventh and eighth tiles from the left in the middle row). A subsequent gadget will detect this gap and thus learn the location of the leftmost edge of the input. }
\label{fig:Turing_machine_input}
\vspace{0pt}
\end{figure}

\subsection{Preprocessing}
\begin{figure}[htp]
\centering
\includegraphics[width=\textwidth]{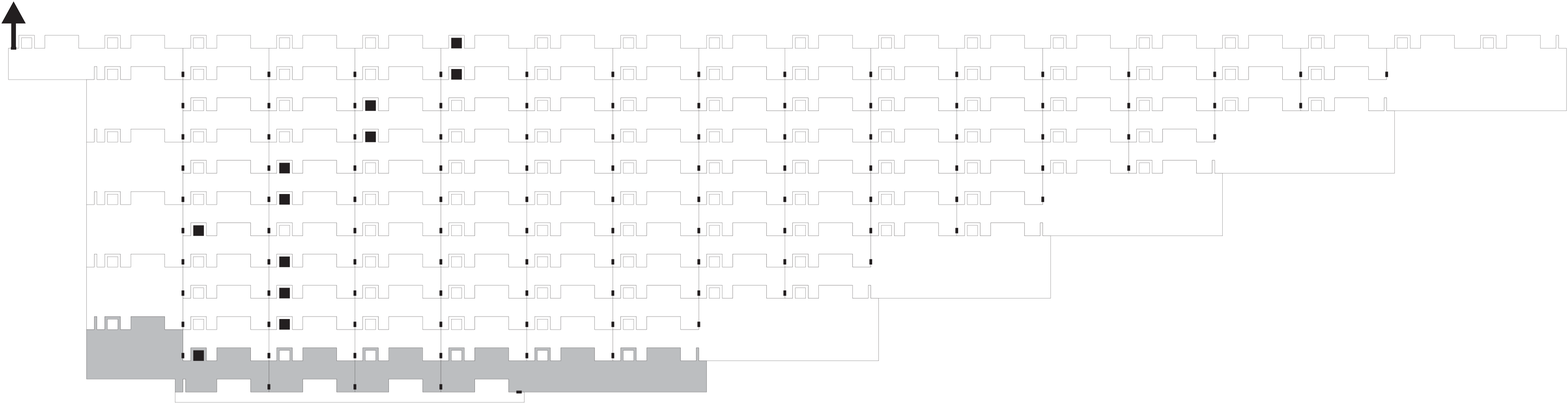}
\caption{The preprocessing gadgets are indicated with dark grey.}
\label{fig:overview_preprocessing}
\vspace{0pt}
\end{figure}
The first phase in our Turing machine simulation is \emph{preprocessing}. In this phase, the input bits are read from right-to-left and copied toward the north by bit-reading gadgets. Each input bit is augmented with another (geometrically-represented) bit, which indicates whether or not the tape head is currently reading that cell. The tape cell that contains the first input bit receives a 1 bit, so as to indicate that the tape head is currently reading the first bit of input. The rest of the tape cells receive a 0 bit. The gadgets that carry out the preprocessing phase compute the parity of the length of the input. If there is an even number of input bits, then two additional tape cells are added to the left. If there is an odd number of input bits, then only one additional tape cell will be added to the left. Technically, in either case, only one additional tape cell is added to the left by the preprocessing stage (this additional tape cell logically contains the left tape marker symbol $\#$, which is geometrically encoded). However, the leftmost tape cell gadget will encode the parity of the length of the input and, in the case of even parity, an additional 0 bit will be added in the post-processing phase of the simulation.

The gadgets in this subsection are shown in Figure~\ref{fig:pre_rightmost},~\ref{fig:pre_middle} and~\ref{fig:pre_leftmost}.

\begin{figure}[htp] 
\fbox{
\begin{minipage}[t][40mm][b]{0.98\textwidth}
\centering
 \includegraphics[width=0.7\textwidth]{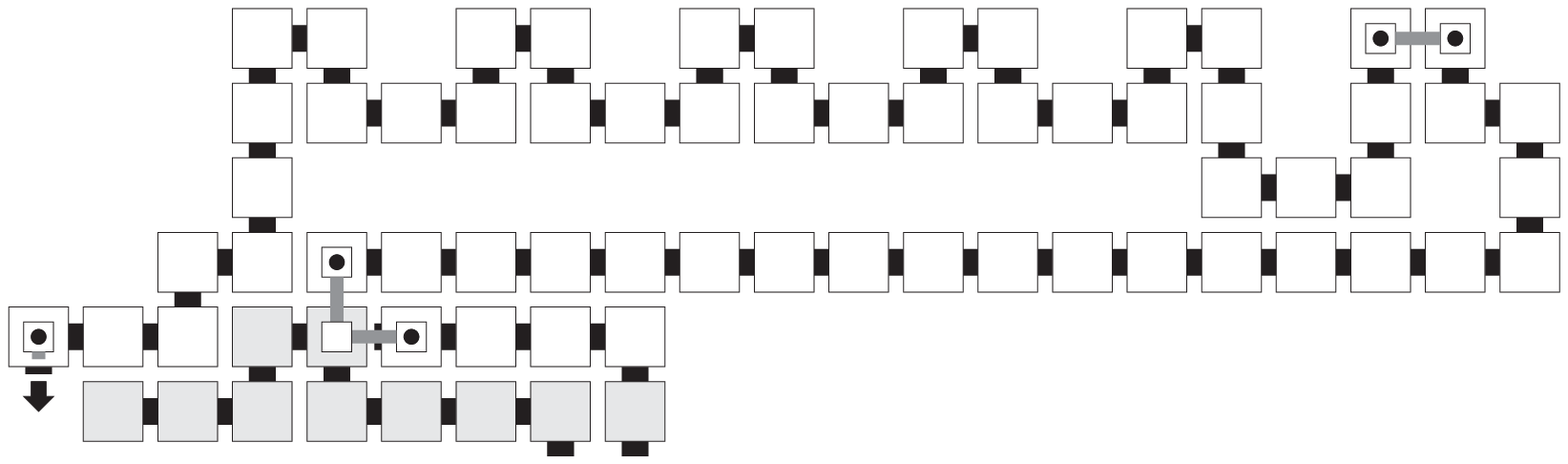}
\medskip

\begin{minipage}[t][][t]{\textwidth}
\centering\small
(a) {\tt Pre\_rightmost\_0}: the last input bit is 0.
\end{minipage}
\end{minipage}}
\medskip

\begin{minipage}[t]{0.98\textwidth}
\fbox{
\begin{minipage}[t][40mm][b]{\textwidth}
\vspace{0pt}
\centering
 \includegraphics[width=0.7\textwidth]{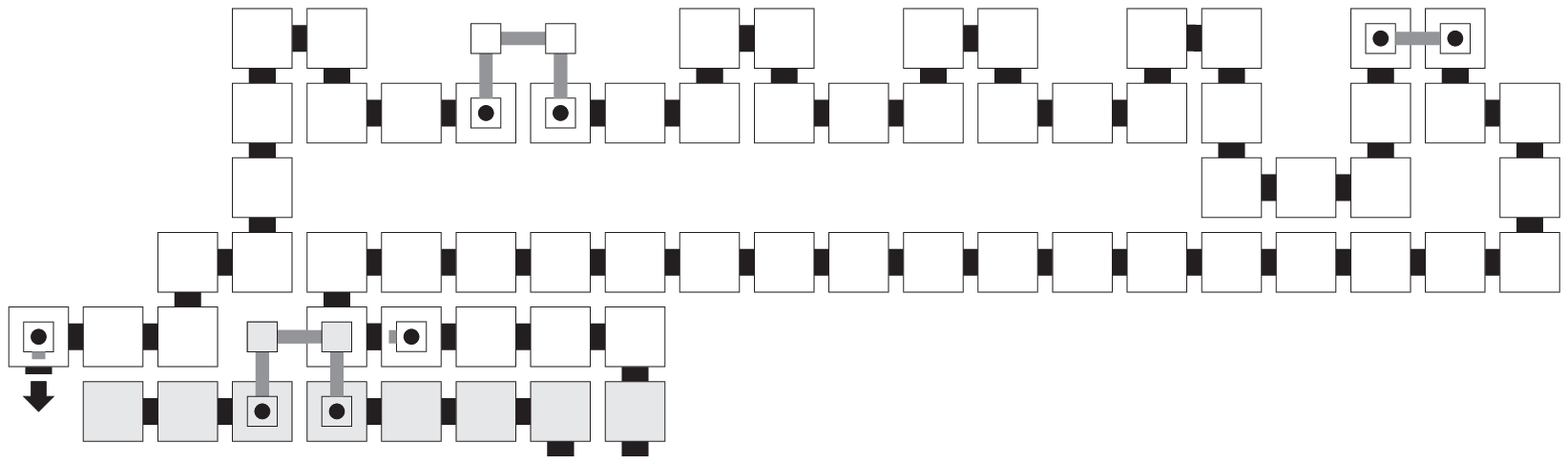}
\medskip

\begin{minipage}[t][][t]{\textwidth}
\centering\small
(b) {\tt Pre\_rightmost\_1}: the last input bit is 1.
\end{minipage}
\end{minipage}}

\end{minipage}
\caption{The {\tt Pre\_rightmost} gadgets, which copy the rightmost input bit and grow the tape to the right by two additional tape cells. These gadgets actually shift the input bit over to the right by three tiles. Note the first two (leftmost) bumps along the top of this gadget represent the last input bit, the first (leftmost) of which is in the $z=0$ plane to indicate that the tape head is currently not reading the last bit of input (this convention is followed in subsequent gadgets). The deep dent between the final two bit-bumps indicates the position of the rightmost tape cell (this convention is followed in subsequent gadgets). In this, and subsequent figures, the grey tiles represent previously-placed tiles that this gadget is using as geometrically-encoded input. The arrow pointing away from the gadget represents its output glue. The empty regions in these gadgets (and subsequent gadgets) can be filled in with a constant number of filler tiles. }
\label{fig:pre_rightmost}
\end{figure}

\begin{figure}[htp] 
\centering
\begin{minipage}[t]{\textwidth}
\centering
\fbox{
\begin{minipage}[t][45mm][b]{.45\textwidth}
\vspace{0pt}
\centering
 \includegraphics[width=.6\textwidth]{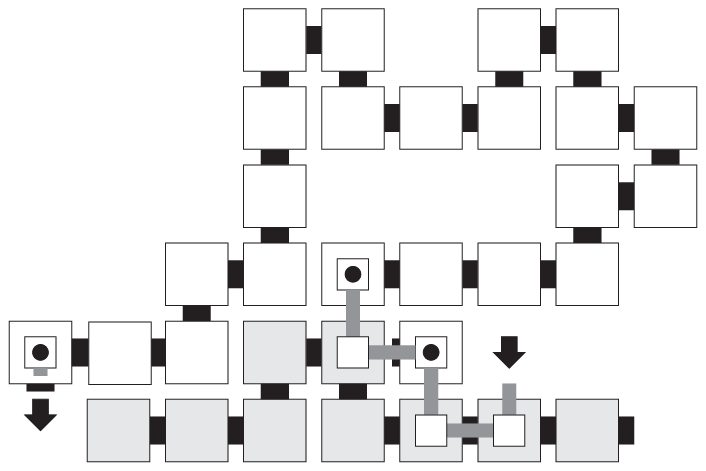}
\medskip
\begin{minipage}[t][][t]{\textwidth}
\centering\small
(a) {\tt Pre\_middle\_0}
\end{minipage}
\end{minipage}}
\fbox{
\begin{minipage}[t][45mm][b]{0.45\textwidth}
\centering
 \includegraphics[width=.6\textwidth]{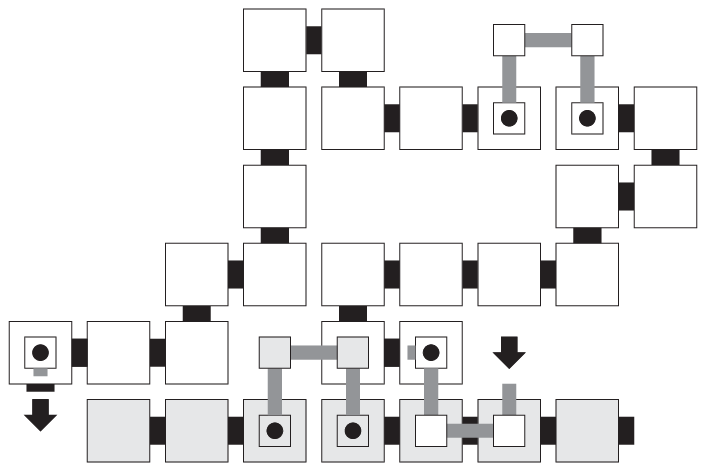}
\medskip
\begin{minipage}[t][][t]{\textwidth}
\centering\small
(b) {\tt Pre\_middle\_1}
\end{minipage}
\end{minipage}}
\end{minipage}
\caption{The {\tt Pre\_middle} gadgets. These gadgets copy the middle input bits (i.e., neither the leftmost nor the rightmost input bits). Observe that, in either case, the tape head bit-bump is in the $z=0$ plane and therefore indicates that the tape head is not present. The input glues of these gadgets bind to the output glues of {\tt Pre\_middle} or {\tt Pre\_right} gadgets.}
\label{fig:pre_middle}
\end{figure}

\begin{figure}[htp]
\centering


\begin{minipage}[t]{\textwidth}
\centering

\fbox{
\begin{minipage}[t][65mm][b]{0.45\textwidth}
\centering
 \includegraphics[width=.9\textwidth]{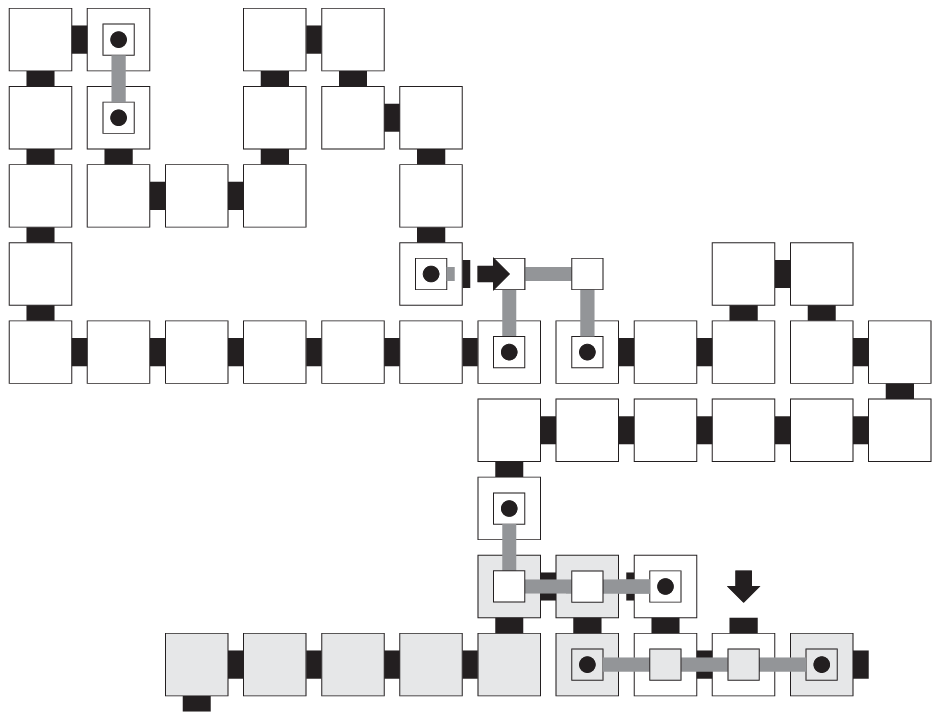}

(a) {\tt Pre\_leftmost\_even\_0}
\end{minipage}}
\fbox{
\begin{minipage}[t][65mm][b]{0.45\textwidth}
\centering
 \includegraphics[width=.9\textwidth]{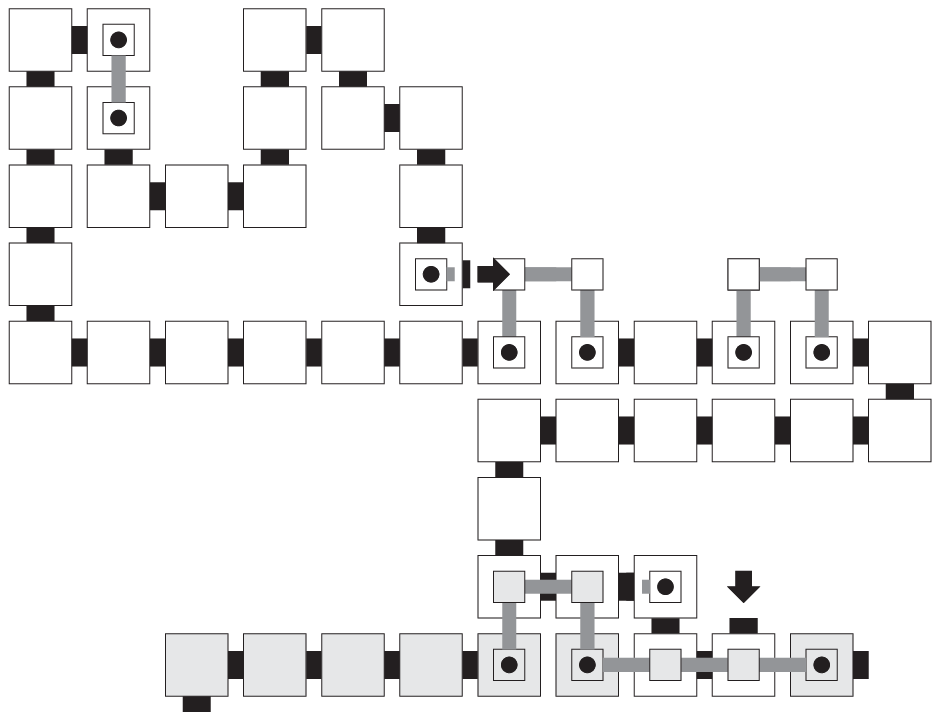}

(b) {\tt Pre\_leftmost\_even\_1}
\end{minipage}}
\end{minipage}

\bigskip


\begin{minipage}[t]{\textwidth}
\centering
\fbox{
\begin{minipage}[t][65mm][b]{0.45\textwidth}
\centering
 \includegraphics[width=.9\textwidth]{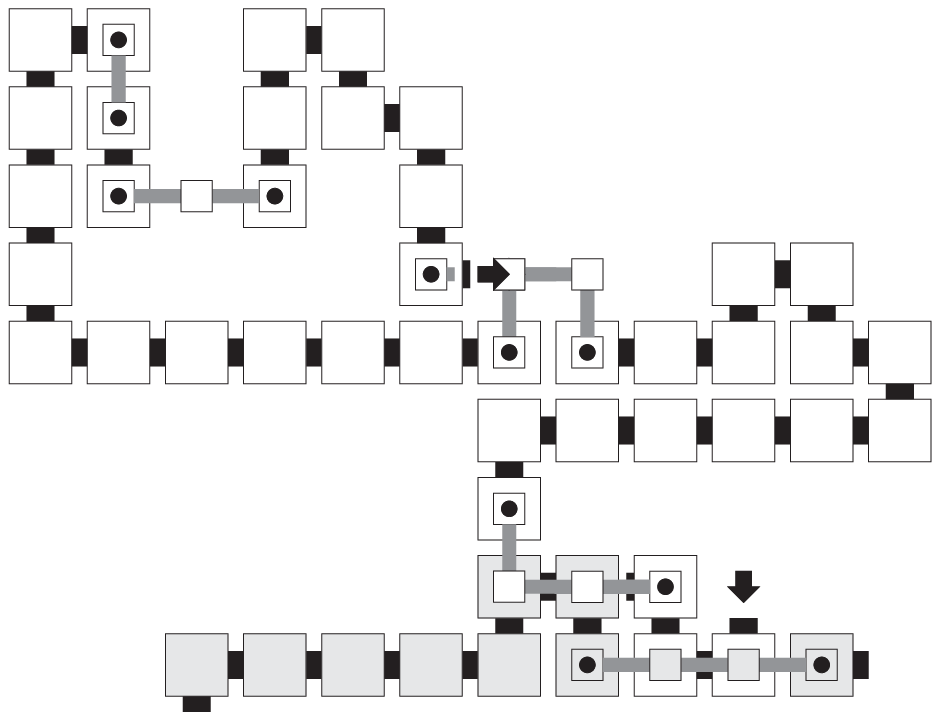}

(c) {\tt Pre\_leftmost\_odd\_0}
\end{minipage}}
\fbox{
\begin{minipage}[t][65mm][b]{0.45\textwidth}
\centering
 \includegraphics[width=.9\textwidth]{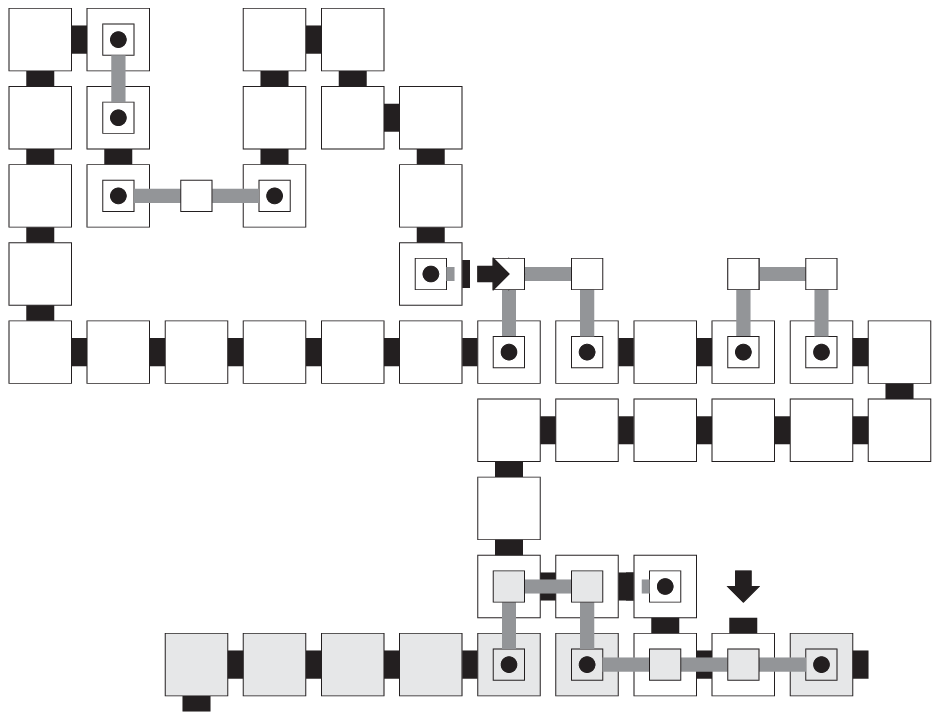}

(d) {\tt Pre\_leftmost\_odd\_1}
\end{minipage}}
\end{minipage}

\caption{The {\tt Pre\_leftmost} gadgets. These gadgets preprocess the leftmost input bit by copying its value as well as adding the bit indicating the presence of the tape head. They also add an additional tape cell to the left and it is assumed that, if $M$ is ever reading this write most tape cell, it is actually reading a special left marker symbol \#. These gadgets initiate the first left-to-right transition row. The deep dent in the upper leftmost portion of the gadget indicates the leftmost position of the tape. This dent also encodes the parity of the length of the input (see Figure~\ref{fig:RL_copy_leftmost_no_halt}, which contains gadgets that read and propagate this parity information). The input glues of these gadgets are not defined with a state of $M$. However, the output glue is $q_0$, i.e., the start state of the Turing machine being simulated. The output glue also encodes whether $\delta(q_0,w_0)$ is a left or right transition. Thus, the first row of left-to-right transition gadgets will assemble with the knowledge that $M$ is in state $q_0$ and which way $M$ is going to move. For example, if $\delta(q_0,w_0)$ is a left-moving transition, then the output glue is $\langle q_0, L\rangle$ (the arrow for the output glue overlaps with a tile in the $z=1$ plane). In general, the input and output glues for subsequent gadgets will encode, among other things, the state and direction of the next transition of $M$. The input glues for these gadgets bind to the output glues of {\tt Pre\_middle} gadgets.}
\label{fig:pre_leftmost}
\vspace{0pt}
\end{figure}

\clearpage
\subsection{Simulation (left-to-right)}
\begin{figure}[htp]
\centering
\includegraphics[width=\textwidth]{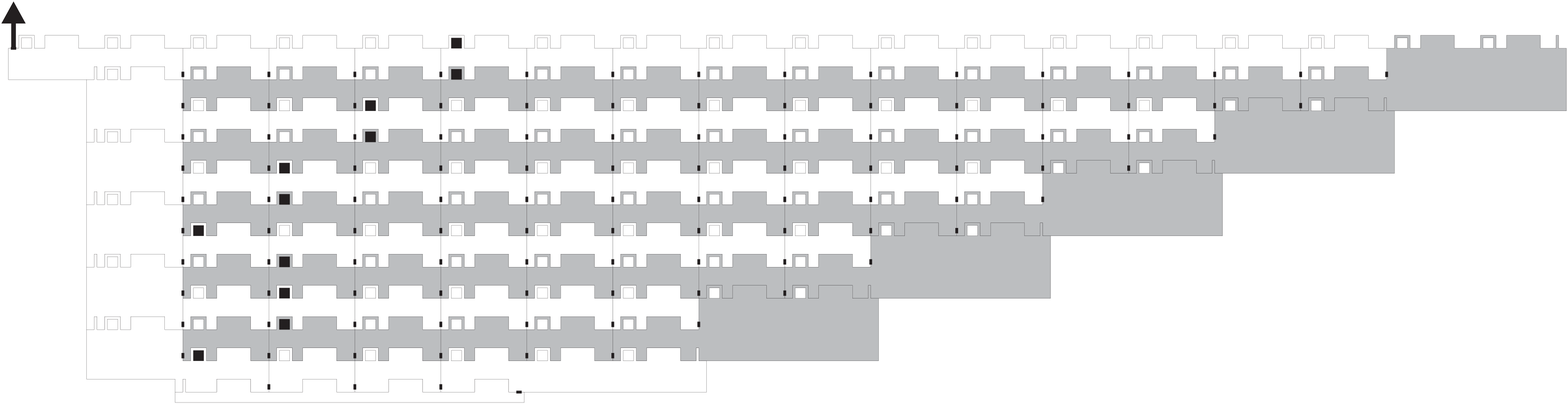}
\caption{The gadgets that assemble in a left-to-right fashion are indicated with dark grey.}
\label{fig:LR_overview}
\vspace{0pt}
\end{figure}
All the gadgets in this subsection have names that start with {\tt LR} because they all assemble in a left-to-right fashion. These gadgets are indicated in Figure~\ref{fig:LR_overview}.

\subsubsection{Move the tape head}
The gadgets in this subsection are shown in Figures~\ref{fig:LR_send_tape_head} and~\ref{fig:LR_receive_tape_head}.
\begin{figure}[htp]
\centering


\begin{minipage}[t]{\textwidth}
\centering

\fbox{
\begin{minipage}[t][55mm][b]{0.45\textwidth}
\centering
 \includegraphics[width=.55\textwidth]{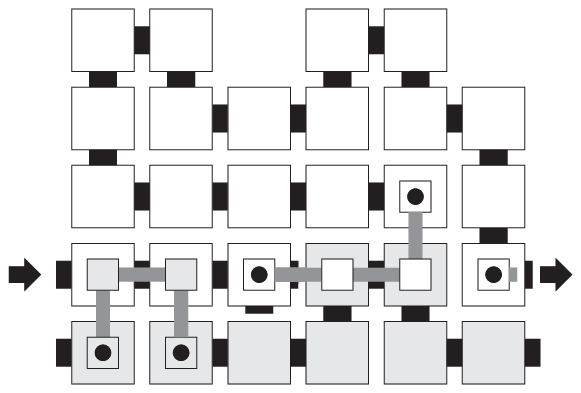}

(a) {\tt LR\_0\_0}: For each transition of the form $\delta(p,0) = (q,0,R)$, where $p \in Q - \{q_{halt}\}$, create one of these gadgets with input glue $\langle p, R\rangle$ and output glue $\langle q, R\rangle$.
\end{minipage}}
\fbox{
\begin{minipage}[t][55mm][b]{0.45\textwidth}
\centering
 \includegraphics[width=.55\textwidth]{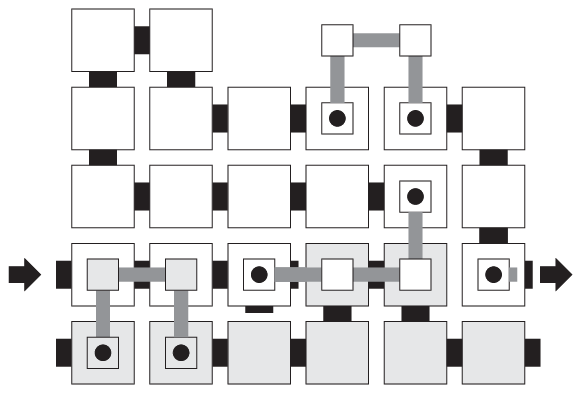}

(b) {\tt LR\_0\_1}: For each transition of the form $\delta(p,0) = (q,1,R)$, where $p \in Q - \{q_{halt}\}$, create one of these gadgets with input glue $\langle p, R\rangle$ and output glue $\langle q, R\rangle$.
\end{minipage}}
\end{minipage}

\bigskip


\begin{minipage}[t]{\textwidth}
\centering
\fbox{
\begin{minipage}[t][55mm][b]{0.45\textwidth}
\centering
 \includegraphics[width=.55\textwidth]{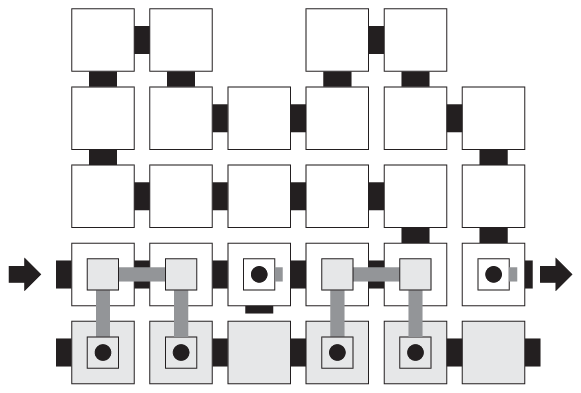}

(c) {\tt LR\_1\_0}: For each transition of the form $\delta(p,1) = (q,0,R)$, where $p \in Q - \{q_{halt}\}$, create one of these gadgets with input glue $\langle p, R\rangle $ and output glue $\langle q, R\rangle$.
\end{minipage}}
\fbox{
\begin{minipage}[t][55mm][b]{0.45\textwidth}
\centering
 \includegraphics[width=.55\textwidth]{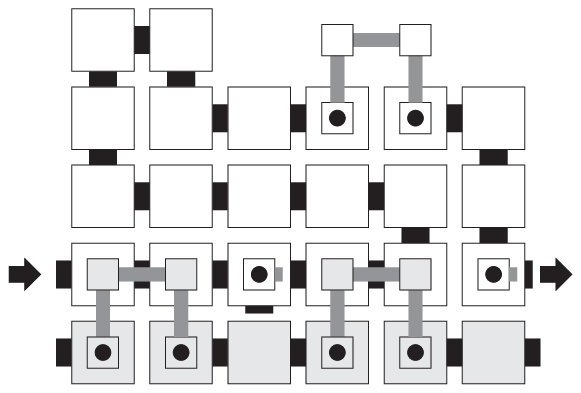}

(d) {\tt LR\_1\_1}: For each transition of the form $\delta(p,1) = (q,1,R)$, where $p \in Q - \{q_{halt}\}$, create one of these gadgets with input glue $\langle p, R\rangle $ and output glue $\langle q, R\rangle$.
\end{minipage}}
\end{minipage}

\caption{The {\tt LR\_to\_right} gadgets. These gadgets execute the first half of a right-moving transition, i.e., they ``send'' the tape head to the right. Their input glues bind to the output glues of {\tt LR\_copy} (see Figure~\ref{fig:LR_copy}), {\tt Pre\_leftmost} or {\tt RL\_copy\_leftmost\_no\_halt} (see Figure~\ref{fig:RL_copy_leftmost_no_halt}) gadgets.}
\label{fig:LR_send_tape_head}
\vspace{0pt}
\end{figure}

\begin{figure}[htp] 
\centering
\begin{minipage}[t]{\textwidth}
\centering
\fbox{
\begin{minipage}[t][65mm][b]{.45\textwidth}
\vspace{0pt}
\centering
 \includegraphics[width=.55\textwidth]{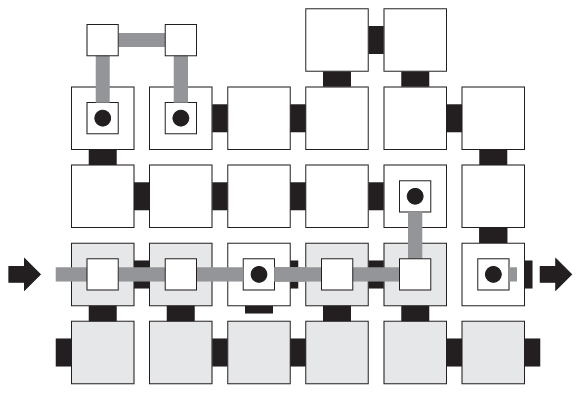}
\medskip
\begin{minipage}[t][][t]{\textwidth}
\centering\small
(a) {\tt LR\_from\_left\_0}: Receive the tape head from the left (now reading a 0). For each $q \in Q$, create one of these gadgets with input glue $\langle q,R\rangle$ and output glue $\langle q, R\rangle$ if $move(q,0) = R$ and $\langle q, L^*\rangle$ otherwise. We will use the notation $L^*$ to indicate that the next move is to the left, but the previous move was to the right (thus, the left move must wait until the next transition).
\end{minipage}
\end{minipage}}
\fbox{
\begin{minipage}[t][65mm][b]{0.45\textwidth}
\centering
 \includegraphics[width=.55\textwidth]{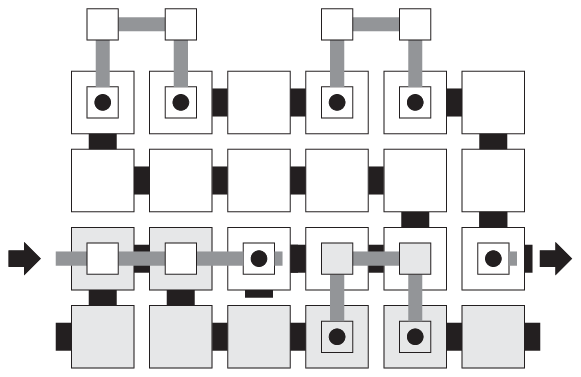}
\medskip
\begin{minipage}[t][][t]{\textwidth}
\centering\small
(b) {\tt LR\_from\_left\_1}: Receive the tape head from the left (now reading a 1). For each $q \in Q$, create one of these gadgets with input glue $\langle q,R\rangle$ and output glue $\langle q, R\rangle$ if $move(q,1) = R$ and $\langle q, L^*\rangle$ otherwise.
\end{minipage}
\end{minipage}}
\end{minipage}
\caption{The {\tt LR\_from\_left} gadgets. These gadgets execute the second half of a right-moving transition, i.e., they ``receive'' the tape head from the left. The input glues for these gadgets bind to the output glues of the {\tt LR\_to\_right} gadgets.}
\label{fig:LR_receive_tape_head}
\end{figure}

\subsubsection{Copy}
The gadgets in this subsection are shown in Figure~\ref{fig:LR_copy}.
\begin{figure}[htp]
\centering


\begin{minipage}[t]{\textwidth}
\centering

\fbox{
\begin{minipage}[t][55mm][b]{0.45\textwidth}
\centering
 \includegraphics[width=.55\textwidth]{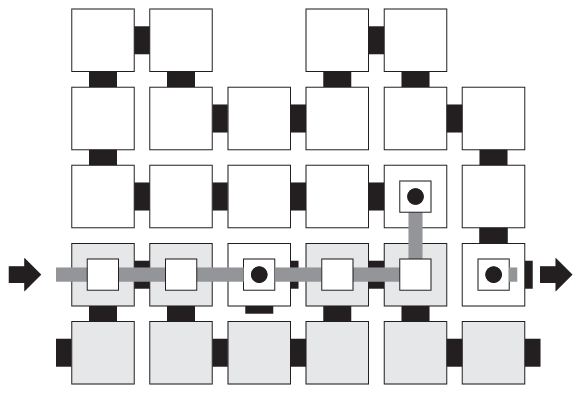}

(a) {\tt LR\_copy\_no\_state\_0}: Copy a 0 bit that is not being read by the tape head. For each $q \in Q$ and $d \in \{L,R,L^*\}$, create one of these gadgets with input glue $\langle q,d\rangle$ and output glue $\langle q,d\rangle$.
\end{minipage}}
\fbox{
\begin{minipage}[t][55mm][b]{0.45\textwidth}
\centering
 \includegraphics[width=.55\textwidth]{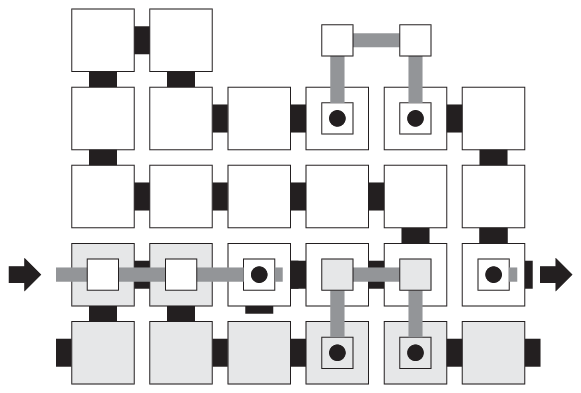}

(b) {\tt LR\_copy\_no\_state\_1}: Copy a 1 bit that is not being read by the tape head. For each $q \in Q$ and $d \in \{L,R,L^*\}$, create one of these gadgets with input glue $\langle q,d\rangle$ and output glue $\langle q,d\rangle$.
\end{minipage}}
\end{minipage}

\bigskip


\begin{minipage}[t]{\textwidth}
\centering
\fbox{
\begin{minipage}[t][55mm][b]{0.45\textwidth}
\centering
 \includegraphics[width=.55\textwidth]{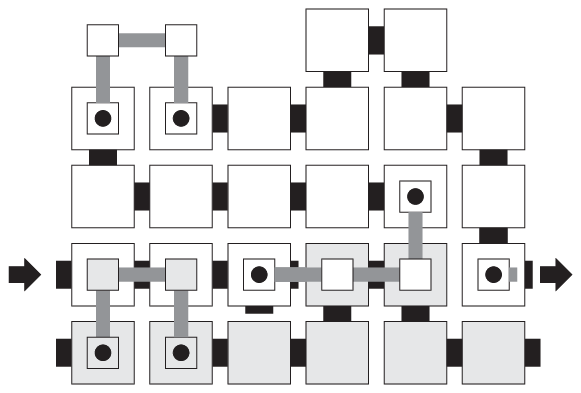}

(c) {\tt LR\_copy\_with\_state\_0}: Copy a 0 bit that is being read by the tape head. For each $q \in Q$, create one of these gadgets with input glue $\langle q,L\rangle$ and output glue $\langle q, L\rangle$.
\end{minipage}}
\fbox{
\begin{minipage}[t][55mm][b]{0.45\textwidth}
\centering
 \includegraphics[width=.55\textwidth]{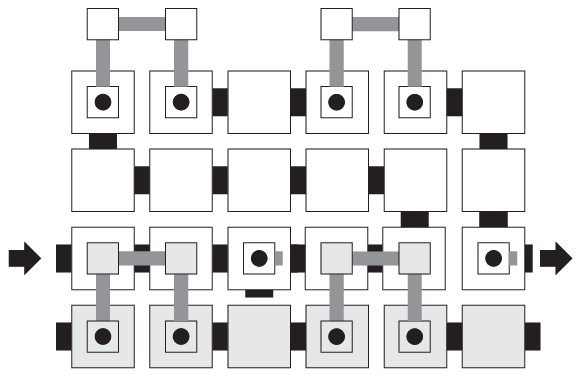}

(d) {\tt LR\_copy\_with\_state\_1}: Copy a 1 bit that is being read by the tape head. For each $q \in Q$, create one of these gadgets with input glue $\langle q,L\rangle$ and output glue $\langle q, L\rangle$.
\end{minipage}}
\end{minipage}

\caption{The {\tt LR\_copy} gadgets. These gadgets propagate the contents of a tape cell up to the next (right-to-left) transition row. The input glues of these gadgets bind to the output glues of either {\tt LR\_from\_left}, {\tt LR\_copy}, {\tt Pre\_leftmost} or {\tt RL\_copy\_leftmost} (see Figure~\ref{fig:RL_copy_leftmost_no_halt}) gadgets. }
\label{fig:LR_copy}
\vspace{0pt}
\end{figure}

\subsubsection{Grow the tape}
The gadget in this subsection is shown in Figure~\ref{fig:LR_grow_the_tape}.
\begin{figure}[htp]
\begin{minipage}{\textwidth}
\centering
\fbox{
 \includegraphics[width=0.65\textwidth]{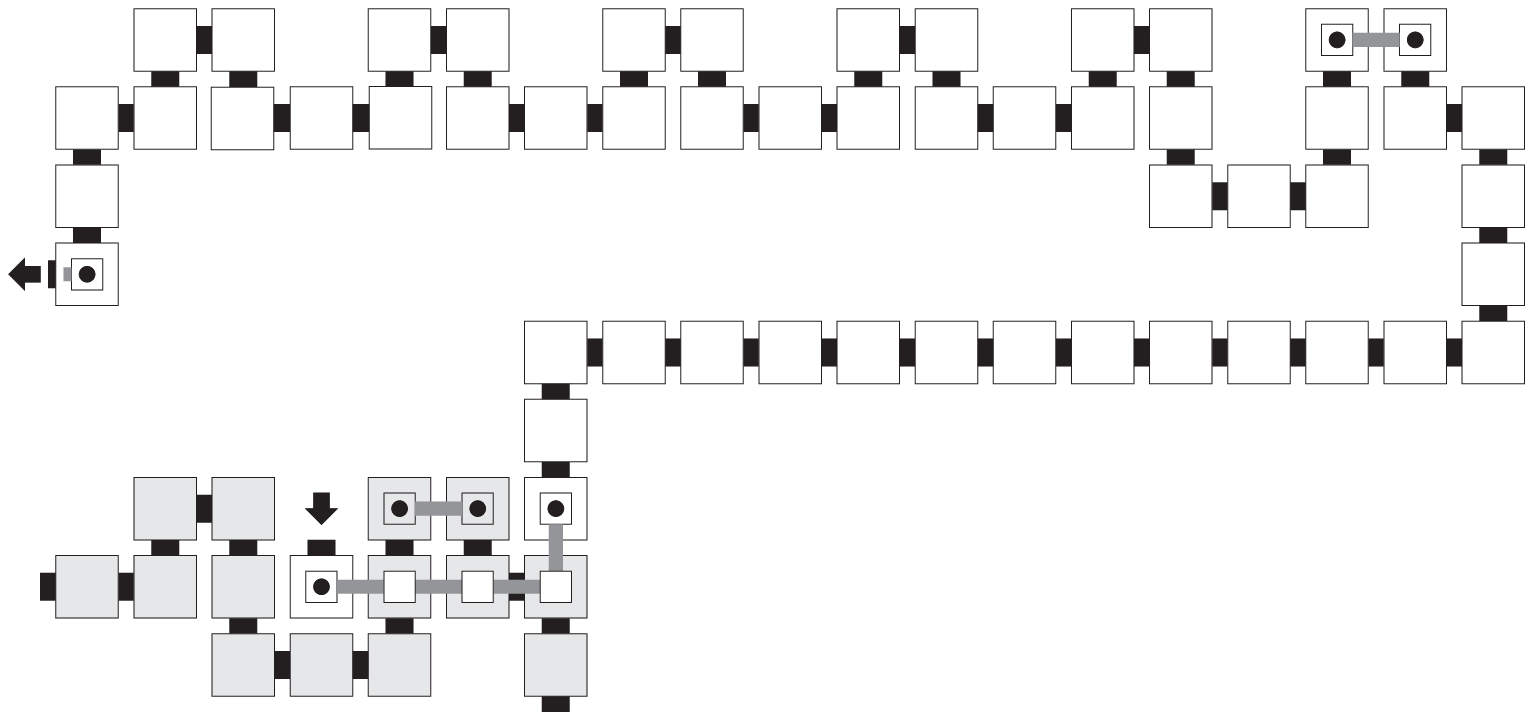}
}
\end{minipage}

\caption{The {\tt Grow\_tape} gadget. This gadget expands the tape to the right by two cells. Note that the deep dent indicates the position of the rightmost tape cell. Therefore, for each $q\in Q$ and $d \in \{L,R,L^*\}$, create one of these gadgets with input and output glues $\langle q, d\rangle$. The input glue of this gadget binds to the south-facing output glues of {\tt LR\_copy} gadgets (these glues face south from the tile that binds after the {\tt LR\_copy} gadget reads the tape head bit but before it reads the tape cell value bit.}
\label{fig:LR_grow_the_tape}
\vspace{0pt}
\end{figure}

\subsection{Simulation (right-to-left)}
\begin{figure}[htp]
\centering
\includegraphics[width=\textwidth]{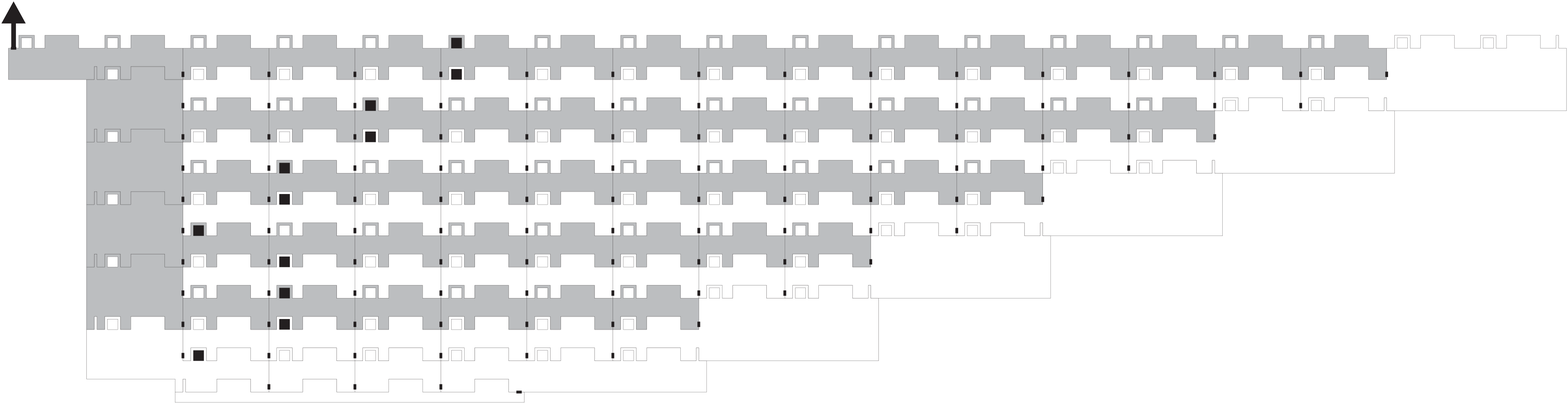}
\caption{The gadgets that assemble in a right-to-left fashion are indicated with dark grey.}
\label{fig:RL_overview}
\vspace{0pt}
\end{figure}
All the gadgets in this subsection have names that start with {\tt RL} because they all assemble in a right-to-left fashion. These gadgets are indicated in Figure~\ref{fig:RL_overview}.

\subsubsection{Guess a bit}
Since the tape head indicator bit is to the left of the tape cell value indicator bit in each gadget, we use generic bit-guessing gadgets to first guess the value of a tape cell, before determining whether the tape head is reading that tape cell.

The gadgets in this subsection are shown in Figure~\ref{fig:RL_guess}.

\begin{figure}[htp] 
\centering
\begin{minipage}[t]{\textwidth}
\centering
\fbox{
\begin{minipage}[t][30mm][b]{.45\textwidth}
\vspace{0pt}
\centering
 \includegraphics[width=.55\textwidth]{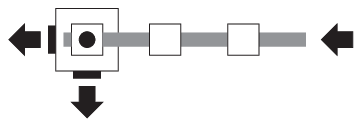}
\medskip
\begin{minipage}[t][][t]{\textwidth}
\centering\small
(a) {\tt RL\_guess\_0}: Guess a 0 bit.
\end{minipage}
\end{minipage}}
\fbox{
\begin{minipage}[t][30mm][b]{0.45\textwidth}
\centering
 \includegraphics[width=.55\textwidth]{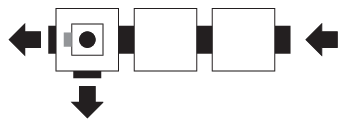}
\medskip
\begin{minipage}[t][][t]{\textwidth}
\centering\small
(b) {\tt RL\_guess\_1}: Guess a 1 bit.
\end{minipage}
\end{minipage}}
\end{minipage}
\caption{The {\tt RL\_guess} gadgets. These gadgets guess the binary value of a tape cell. For each $q\in Q$ and $d \in \{L,R,L^*\}$, create one of each of these gadgets with input glue $\langle q,d\rangle$ and output glue $\langle q,d\rangle$. The input glues for these gadgets bind to {\tt Grow\_tape}, {\tt RL\_copy} (see Figure~\ref{fig:RL_copy}) and {\tt RL\_from\_right} (see Figure~\ref{fig:RL_receive_tape_head}) gadgets.}
\label{fig:RL_guess}
\end{figure}

\subsubsection{Move the tape head}
The gadgets in this subsection are shown in Figures~\ref{fig:RL_send_tape_head} and~\ref{fig:RL_receive_tape_head}.
\begin{figure}[htp]
\centering


\begin{minipage}[t]{\textwidth}
\centering

\fbox{
\begin{minipage}[t][55mm][b]{0.45\textwidth}
\centering
 \includegraphics[width=.55\textwidth]{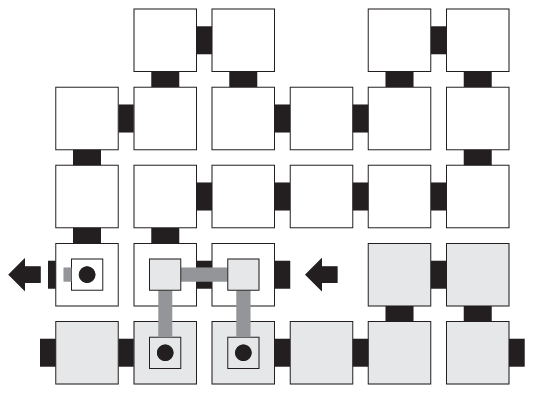}

(a) {\tt RL\_0\_0}: For each transition of the form $\delta(p,0) = (q,0,L)$, where $p \in Q - \{q_{halt}\}$, create one of these gadgets with input glue $\langle p,L\rangle$ and output glue $\langle q,L\rangle$.
\end{minipage}}
\fbox{
\begin{minipage}[t][55mm][b]{0.45\textwidth}
\centering
 \includegraphics[width=.55\textwidth]{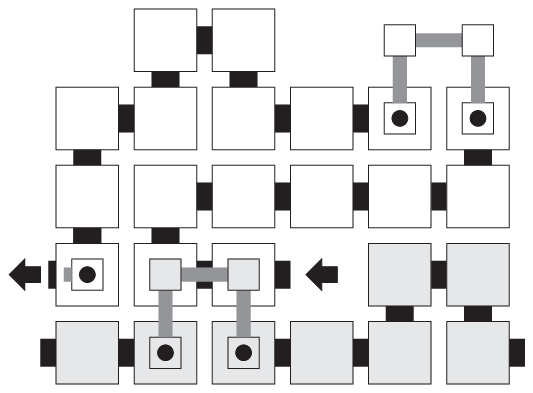}

(b) {\tt RL\_0\_1}: For each transition of the form $\delta(p,0) = (q,1,L)$, $p \in Q - \{q_{halt}\}$, create one of these gadgets with input glue $\langle p,L\rangle$ and output glue $\langle q,L\rangle$.
\end{minipage}}
\end{minipage}

\bigskip


\begin{minipage}[t]{\textwidth}
\centering
\fbox{
\begin{minipage}[t][55mm][b]{0.45\textwidth}
\centering
 \includegraphics[width=.55\textwidth]{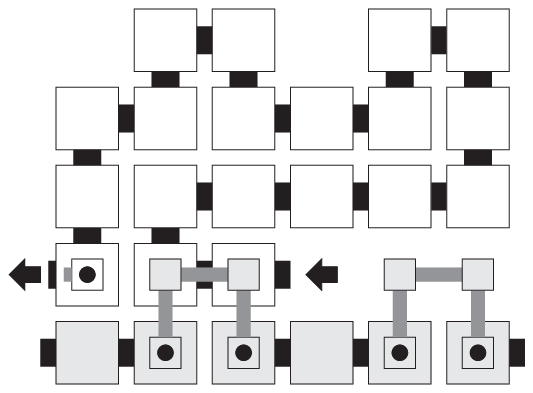}

(c) {\tt RL\_1\_0}: For each transition of the form $\delta(p,1) = (q,0,L)$, $p \in Q - \{q_{halt}\}$, create one of these gadgets with input glue $\langle p,L\rangle$ and output glue $\langle q,L\rangle$.
\end{minipage}}
\fbox{
\begin{minipage}[t][55mm][b]{0.45\textwidth}
\centering
 \includegraphics[width=.55\textwidth]{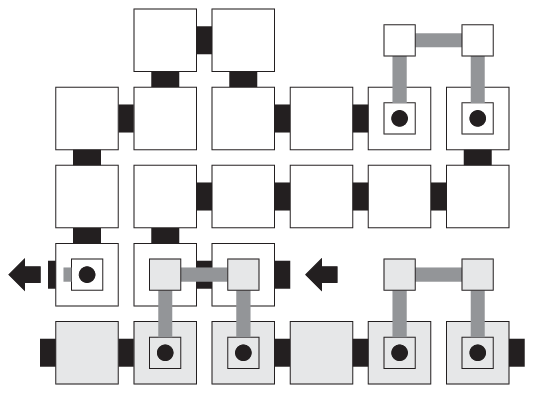}

(d) {\tt RL\_1\_1}: For each transition of the form $\delta(p,1) = (q,1,L)$, $p \in Q - \{q_{halt}\}$, create one of these gadgets with input glue $\langle p,L\rangle$ and output glue $\langle q,L\rangle$.
\end{minipage}}
\end{minipage}

\caption{The {\tt RL\_to\_left} gadgets. These gadgets execute the first half of a left-moving transition, i.e., they ``send'' the tape head to the left. The input glues of these gadgets bind to the output glues of {\tt RL\_guess} gadgets.}
\label{fig:RL_send_tape_head}
\vspace{0pt}
\end{figure}

\begin{figure}[htp] 
\centering
\begin{minipage}[t]{\textwidth}
\centering
\fbox{
\begin{minipage}[t][50mm][b]{.45\textwidth}
\vspace{0pt}
\centering
 \includegraphics[width=.55\textwidth]{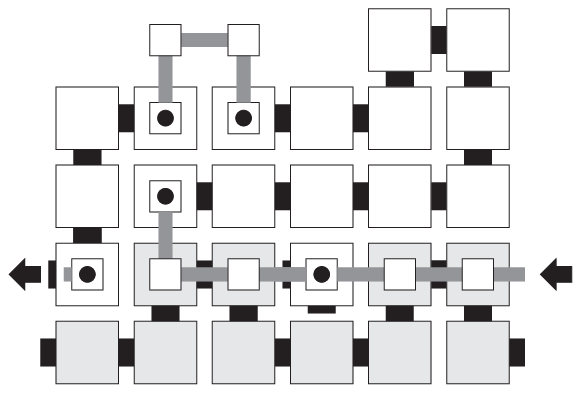}
\medskip
\begin{minipage}[t][][t]{\textwidth}
\centering\small
(a) {\tt RL\_from\_right\_0}: Receive the tape head from the right (now reading a 0). For each $q\in Q$, create one of these gadgets with input glue $\langle q,L\rangle$ and output glue $\langle q,move(q,0)\rangle$.
\end{minipage}
\end{minipage}}
\fbox{
\begin{minipage}[t][50mm][b]{0.45\textwidth}
\centering
 \includegraphics[width=.55\textwidth]{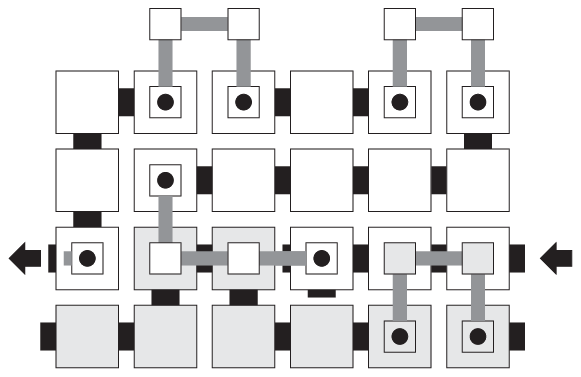}
\medskip
\begin{minipage}[t][][t]{\textwidth}
\centering\small
(b) {\tt RL\_from\_right\_1}: Receive the tape head from the right (now reading a 1). For each $q\in Q$, create one of these gadgets with input glue $\langle q,L\rangle$ and output glue $\langle q,move(q,1)\rangle$.
\end{minipage}
\end{minipage}}
\end{minipage}
\caption{The {\tt RL\_from\_right} gadgets. These gadgets execute the second half of a left-moving transition, i.e., they ``receive'' the tape head from the right. The input glues for these gadgets bind to the output glue of an {\tt RL\_to\_left} gadget.}
\label{fig:RL_receive_tape_head}
\end{figure}

\subsubsection{Copy}
The gadgets in this subsection are shown in Figures~\ref{fig:RL_copy},~\ref{fig:RL_copy_leftmost_no_halt} and~\ref{fig:RL_copy_leftmost_halt}.
\begin{figure}[htp]
\centering


\begin{minipage}[t]{\textwidth}
\centering

\fbox{
\begin{minipage}[t][55mm][b]{0.45\textwidth}
\centering
 \includegraphics[width=.55\textwidth]{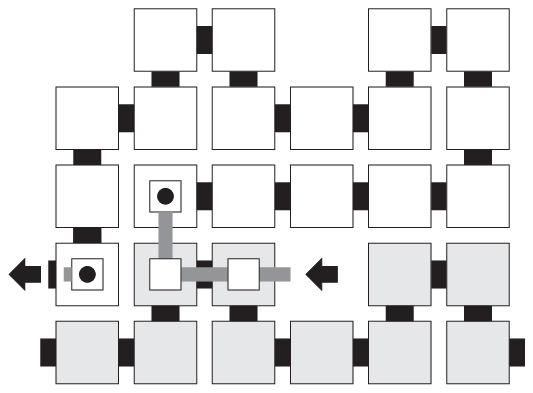}

(a) {\tt RL\_copy\_no\_state\_0}: Copy a 0 bit that is not being read by the tape head. For each $q \in Q$ and $d \in \{L,R,L^*\}$, create one of these gadgets with input glue $\langle q,d\rangle$ and output glue $\langle q,d\rangle$.
\end{minipage}}
\fbox{
\begin{minipage}[t][55mm][b]{0.45\textwidth}
\centering
 \includegraphics[width=.55\textwidth]{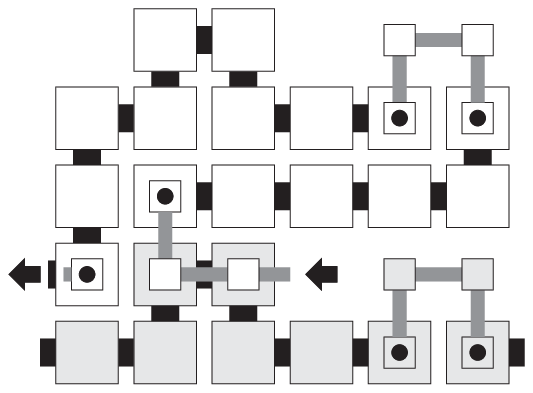}

(b) {\tt RL\_copy\_no\_state\_1}: Copy a 1 bit that is not being read by the tape head. For each $q \in Q$ and $d \in \{L,R,L^*\}$, create one of these gadgets with input glue $\langle q,d\rangle$ and output glue $\langle q,d\rangle$.
\end{minipage}}
\end{minipage}

\bigskip


\begin{minipage}[t]{\textwidth}
\centering
\fbox{
\begin{minipage}[t][55mm][b]{0.45\textwidth}
\centering
 \includegraphics[width=.55\textwidth]{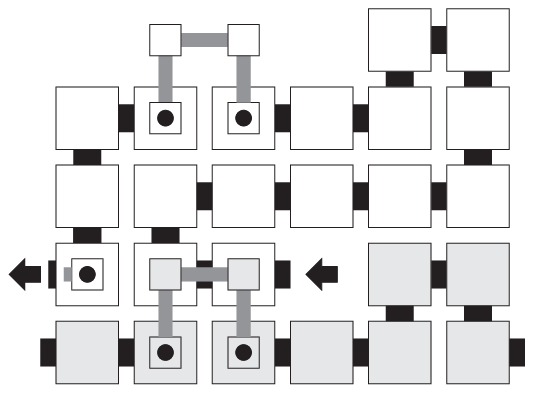}

(c) {\tt RL\_copy\_with\_state\_0}: Copy a 0 bit that is being read by the tape head. For each $q \in Q$ and $d \in \{R,L^*\}$, create one of these gadgets with input glue $\langle q,d\rangle$ and output glue $\langle q,d\rangle$.
\end{minipage}}
\fbox{
\begin{minipage}[t][55mm][b]{0.45\textwidth}
\centering
 \includegraphics[width=.55\textwidth]{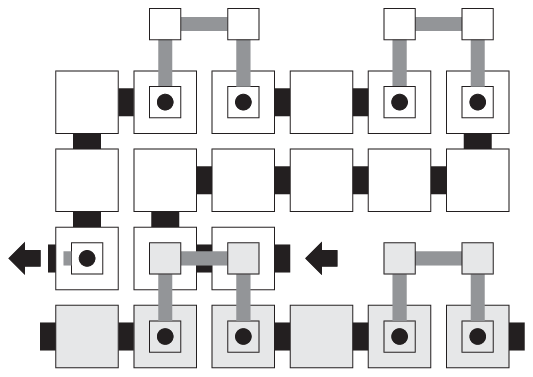}

(d) {\tt RL\_copy\_with\_state\_1}: Copy a 1 bit that is being read by the tape head. For each $q \in Q$ and $d \in \{R,L^*\}$, create one of these gadgets with input glue $\langle q,d\rangle$ and output glue $\langle q,d\rangle$.
\end{minipage}}
\end{minipage}

\caption{The {\tt RL\_copy} gadgets. These gadgets propagate the contents of a tape cell up to the next (left-to-right) transition row. The input glues of these gadgets bind to the output glues of {\tt RL\_guess} gadgets. }
\label{fig:RL_copy}
\vspace{0pt}
\end{figure}

\begin{figure}[htp] 
\centering
\begin{minipage}[t]{\textwidth}
\centering
\fbox{
\begin{minipage}[t][95mm][b]{.45\textwidth}
\vspace{0pt}
\centering
 \includegraphics[width=.55\textwidth]{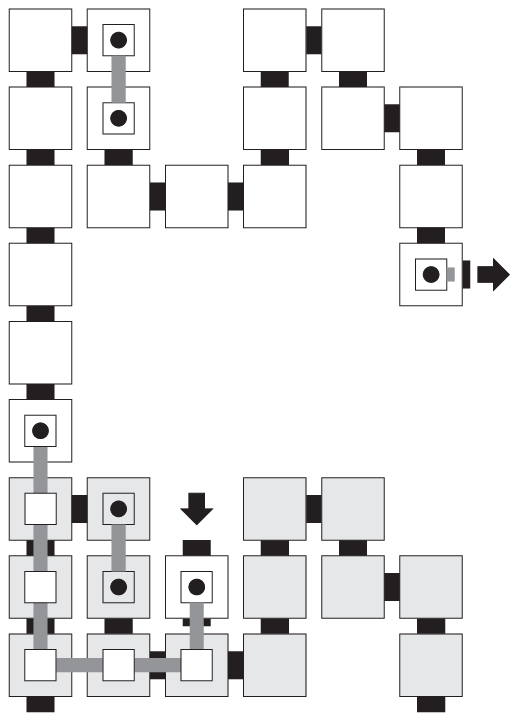}
\medskip
\begin{minipage}[t][][t]{\textwidth}
\centering\small
(a) {\tt RL\_copy\_leftmost\_even\_no\_halt}: The leftmost tape cell. Do not grow the tape to the left but propagate the parity of the length of the input by blocking the deep dent to the south in the $z=0$ plane.

\end{minipage}
\end{minipage}}
\fbox{
\begin{minipage}[t][95mm][b]{0.45\textwidth}
\centering
 \includegraphics[width=.55\textwidth]{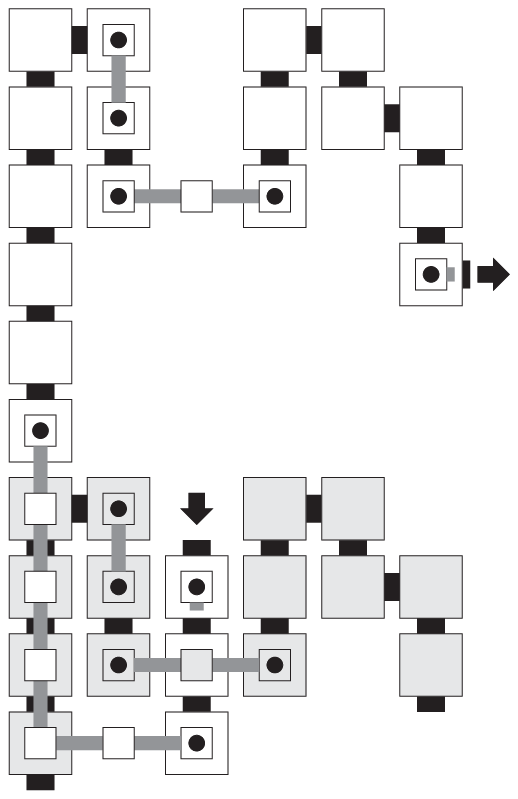}
\medskip
\begin{minipage}[t][][t]{\textwidth}
\centering\small
(b) {\tt RL\_copy\_leftmost\_odd\_no\_halt}: The leftmost tape cell. Do not grow the tape to the left but propagate the parity of the length of the input by blocking the deep dent to the south in the $z=1$ plane.
\end{minipage}
\end{minipage}}
\end{minipage}
\caption{The {\tt RL\_leftmost\_no\_halt} gadgets. These gadgets detect the leftmost tape cell. They geometrically propagate the parity of the length of the input and then initiate the next left-to-right row. Create one of each of these gadgets such that its input glue is $\langle q,L\rangle$, where $q \in Q - \{q_{halt}\}$, its output glue is $\left\langle \delta(q,\#),move(\delta(q,\#),\#)\right\rangle$, the former binds to the output glue of a {\tt RL\_to\_left} gadget. Also, create one of each of these gadgets such that, for any $q \in Q - \{q_{halt}\}$ and $d\in\{L,R,L^*\}$, its input glue is $\langle q,d \rangle$ and its output glue is $\langle q,d\rangle$ if $d \in \{L,R\}$ and $\langle q,L\rangle$ if $d = L^*$. Its input glue binds to the output glue of a {\tt RL\_copy} or {\tt RL\_from\_right} gadget.}
\label{fig:RL_copy_leftmost_no_halt}
\end{figure}

\begin{figure}[htp] 
\centering
\begin{minipage}[t]{\textwidth}
\centering
\fbox{
\begin{minipage}[t][80mm][b]{.475\textwidth}
\vspace{0pt}
\centering
 \includegraphics[width=.99\textwidth]{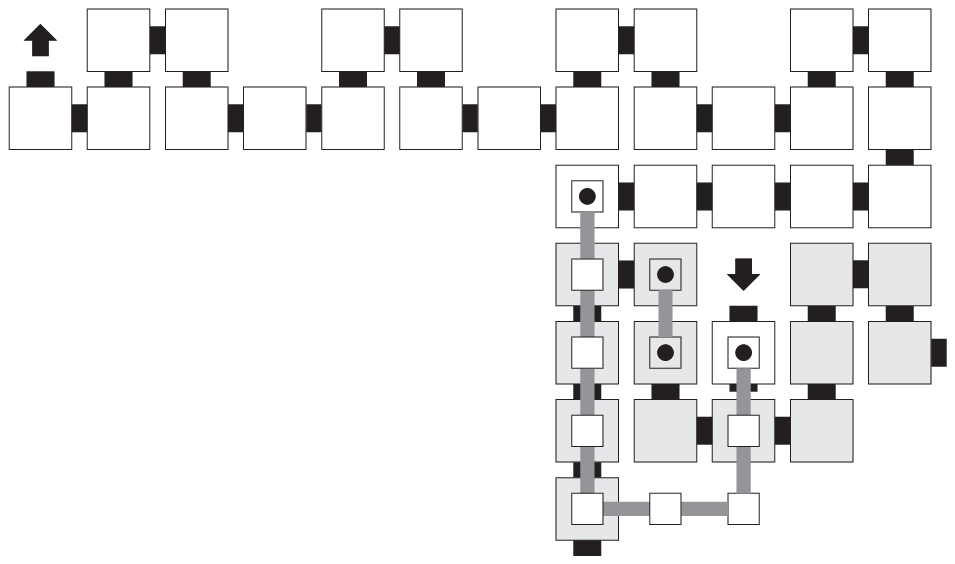}
\medskip
\begin{minipage}[t][][t]{\textwidth}
\centering\small
(a) {\tt RL\_copy\_leftmost\_even\_halt}: The leftmost tape cell and in a halting state (the input state is $q_{halt}$). The length of the input was even, so add one additional tape cell to the left and terminate the transition phase of the simulation (recall that the {\tt Pre\_leftmost} gadgets always added one additional tape cell to the left).
\end{minipage}
\end{minipage}}
\fbox{
\begin{minipage}[t][80mm][b]{0.425\textwidth}
\centering
 \includegraphics[width=.55\textwidth]{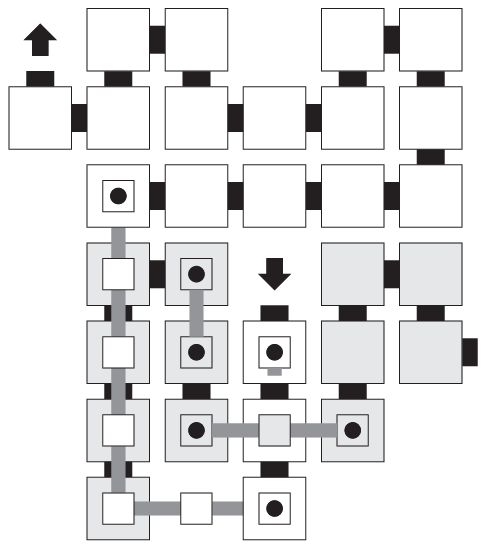}
\medskip
\begin{minipage}[t][][t]{\textwidth}
\centering\small
(b) {\tt RL\_copy\_leftmost\_odd\_halt}: The leftmost tape cell and in a halting state (the input state is $q_{halt}$). The length of the input was odd, so do not add an additional tape cell to the left and terminate the transition phase of the simulation.
\end{minipage}
\end{minipage}}
\end{minipage}
\caption{The {\tt RL\_leftmost\_halt} gadgets. These gadgets detect the leftmost tape cell and the input glue is $\langle q_{halt}, d\rangle$, where $d \in \{L,R\}$. They add either one(zero) additional tape cell(s) to the left if the length of the input is even(odd), respectively. The input glues of these gadgets bind to the output glue of an {\tt RL\_guess} gadget.}
\label{fig:RL_copy_leftmost_halt}
\end{figure}

\clearpage
\subsection{Post-processing}
In the post-processing stage, all of the bits of the gadgets, starting from the leftmost tape cell, to the tape cell that contains the tape head (including the bit contained in this gadget), are shifted to the right most position on the tape cell. The basic idea is that bits are shifted to the right in a zig-zag fashion, similar to how the diagonal signal is shifted to the right in the growth blocks (see Section~\ref{sec:appendix}). Bits are shifted in left-to-right-assembling \emph{shift rows}, where as bits are copied in right-to-left-assembling \emph{copy rows}. Although many of the gadgets in this subsection are similar to those defined the previous subsections, we list them here for the sake of completeness.

\begin{figure}[htp]
\centering
\includegraphics[width=\textwidth]{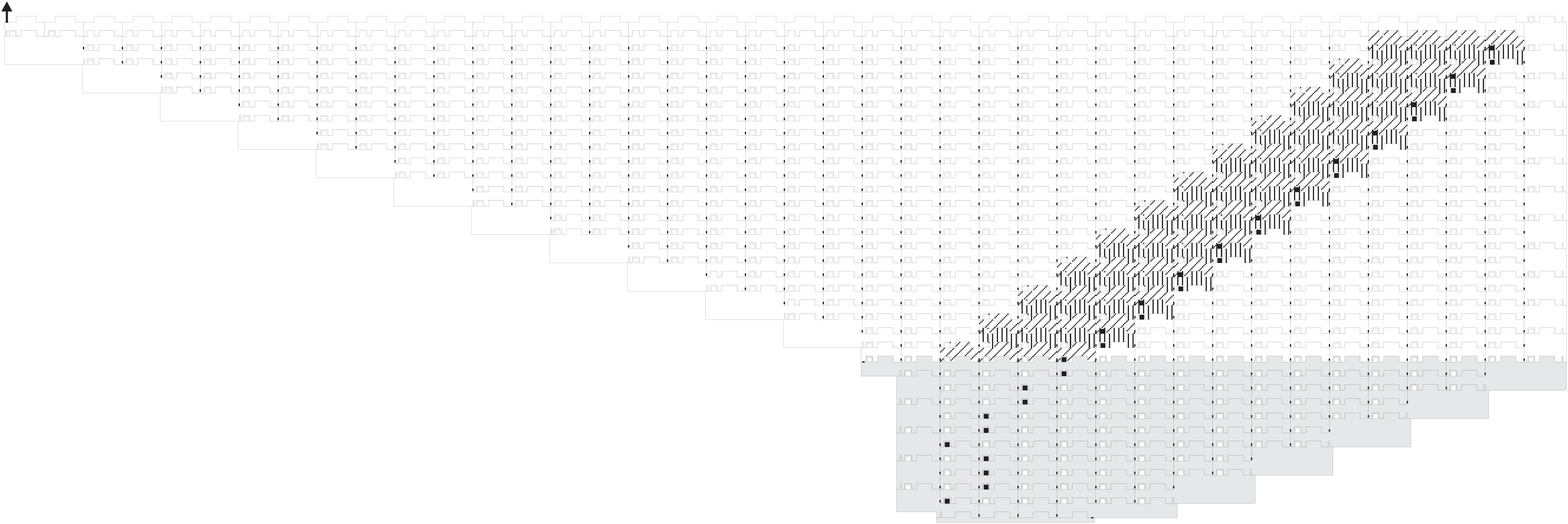}
\caption{Overview of the Turing machine simulation with post-processing stage shown. }
\label{fig:overview_post_process}
\vspace{0pt}
\end{figure}

\begin{figure}[htp]
\begin{minipage}{\textwidth}
\centering
\fbox{
 \includegraphics[width=0.45\textwidth]{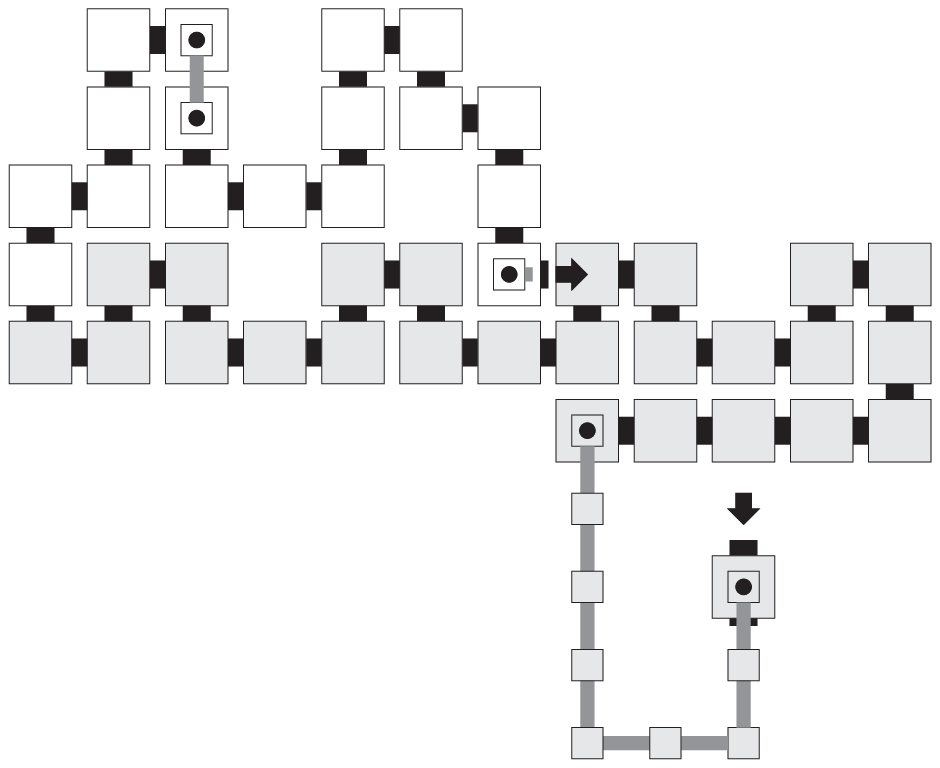}
}
\end{minipage}

\caption{{\tt Start}: The first gadget in the post-processing phase. The grey tiles are from a {\tt RL\_copy\_leftmost\_halt} gadget. The deep indentation toward the left side of the gadget indicates that this is the leftmost tape cell. The input, preprocessing and transition phases are shown in grey. The rest is the post-processing phase.}
\label{fig:Post_process_start}
\vspace{0pt}
\end{figure}

\begin{figure}[htp]
\centering


\begin{minipage}[t]{\textwidth}
\centering

\fbox{
\begin{minipage}[t][45mm][b]{0.45\textwidth}
\centering
 \includegraphics[width=.6\textwidth]{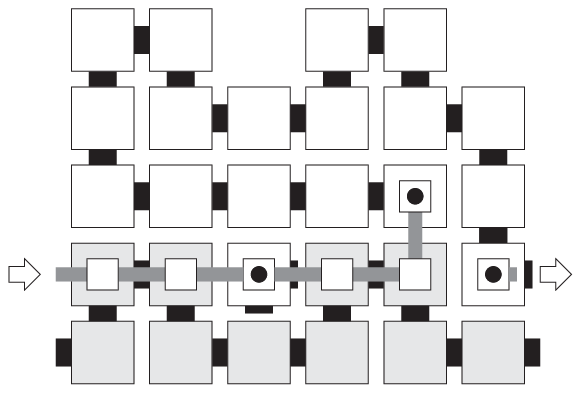}

(a) {\tt LR\_no\_state\_0\_0}: Accept a 0 bit from the left and shift a 0 bit to the right.
\end{minipage}}
\fbox{
\begin{minipage}[t][45mm][b]{0.45\textwidth}
\centering
 \includegraphics[width=.6\textwidth]{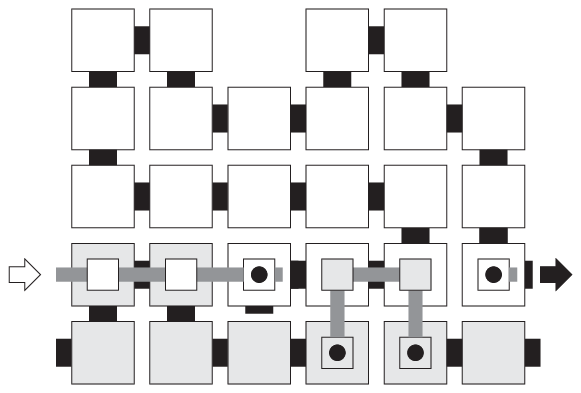}

(b) {\tt LR\_no\_state\_0\_1}: Accept a 0 bit from the left and shift a 1 bit to the right.
\end{minipage}}
\end{minipage}

\bigskip


\begin{minipage}[t]{\textwidth}
\centering
\fbox{
\begin{minipage}[t][45mm][b]{0.45\textwidth}
\centering
 \includegraphics[width=.6\textwidth]{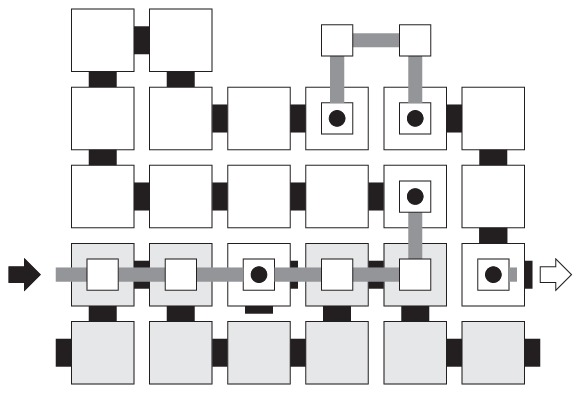}

(c) {\tt LR\_no\_state\_1\_0}: Accept a 1 bit from the left and shift a 0 bit to the right.
\end{minipage}}
\fbox{
\begin{minipage}[t][45mm][b]{0.45\textwidth}
\centering
 \includegraphics[width=.6\textwidth]{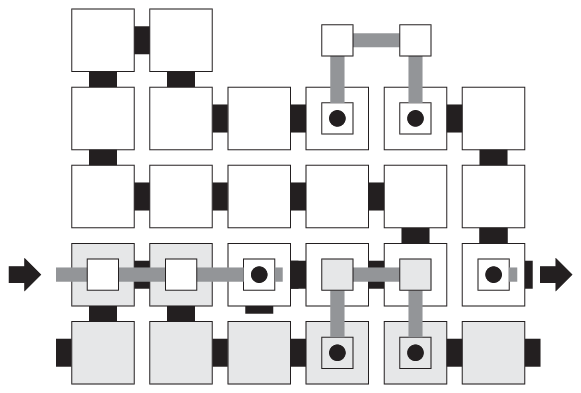}

(d) {\tt LR\_no\_state\_1\_1}: Accept a 1 bit from the left and shift a 1 bit to the right.
\end{minipage}}
\end{minipage}

\caption{The {\tt LR\_no\_state} gadgets. These gadgets accept a bit from the left and shift the bit currently being read to the right. The bit from the left is encoded geometrically along the top of the gadget and the geometrically-encoded bit being read is shifted to the right via the output glue. Here, a black arrow represents a glue that encodes a $1$ bit, while a white arrow represents a glue that encodes a $0$ bit. These gadgets are used to shift bits that are not being read by the tape head. One of these gadgets will initially bind to the {\tt Start} gadget. Then, these gadgets will bind to either other copies of {\tt LR\_no\_state} gadgets or a {\tt LR\_receive\_state} gadget. }
\label{fig:LR_no_state}
\vspace{0pt}
\end{figure}

\begin{figure}[htp]
\centering


\begin{minipage}[t]{\textwidth}
\centering

\fbox{
\begin{minipage}[t][45mm][b]{0.45\textwidth}
\centering
 \includegraphics[width=.63  \textwidth]{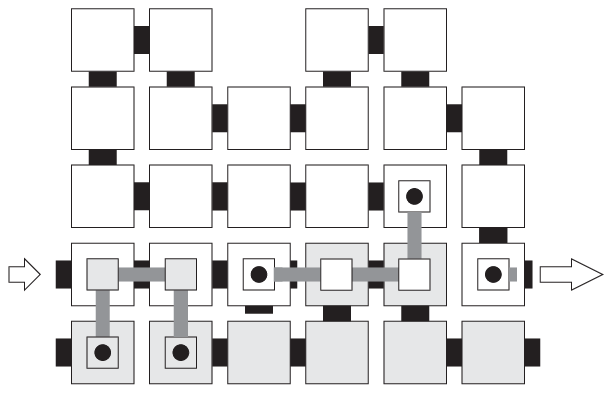}

(a) {\tt LR\_send\_state\_0\_0}: Accept a 0 bit from the left and shift a 0 bit, which is currently read by the tape head, to the right.
\end{minipage}}
\fbox{
\begin{minipage}[t][45mm][b]{0.45\textwidth}
\centering
 \includegraphics[width=.63\textwidth]{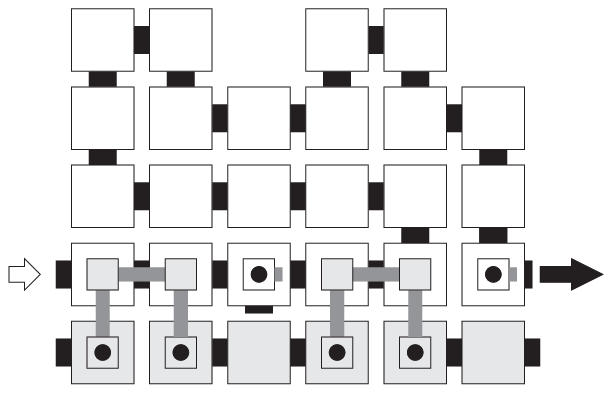}

(b) {\tt LR\_send\_state\_0\_1}: Accept a 0 bit from the left and shift a 1 bit, which is currently read by the tape head, to the right.
\end{minipage}}
\end{minipage}

\bigskip


\begin{minipage}[t]{\textwidth}
\centering
\fbox{
\begin{minipage}[t][45mm][b]{0.45\textwidth}
\centering
 \includegraphics[width=.63\textwidth]{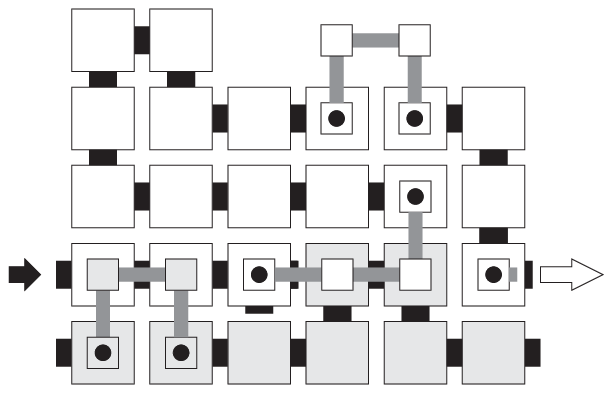}

(c) {\tt LR\_send\_state\_1\_0}: Accept a 1 bit from the left and shift a 0 bit, which is currently read by the tape head, to the right.
\end{minipage}}
\fbox{
\begin{minipage}[t][45mm][b]{0.45\textwidth}
\centering
 \includegraphics[width=.63\textwidth]{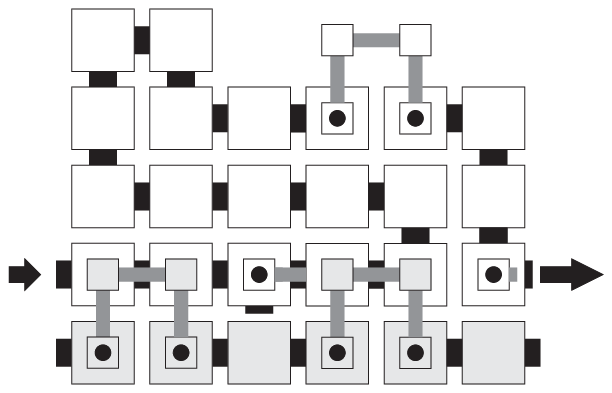}

(d) {\tt LR\_send\_state\_1\_1}: Accept a 1 bit from the left and shift a 1 bit, which is currently read by the tape head, to the right.
\end{minipage}}
\end{minipage}

\caption{The {\tt LR\_send\_state} gadgets. These gadgets work exactly the same as the {\tt LR\_no\_state} gadgets, but they also shift the tape head to the right. Here, a long arrow represents a glue that encodes both a bit value as well as the presence of the tape head. These gadgets will bind to an {\tt LR\_no\_state} gadget.  }
\label{fig:LR_send_state}
\vspace{0pt}
\end{figure}

\begin{figure}[htp]
\centering


\begin{minipage}[t]{\textwidth}
\centering

\fbox{
\begin{minipage}[t][45mm][b]{0.45\textwidth}
\centering
 \includegraphics[width=.63\textwidth]{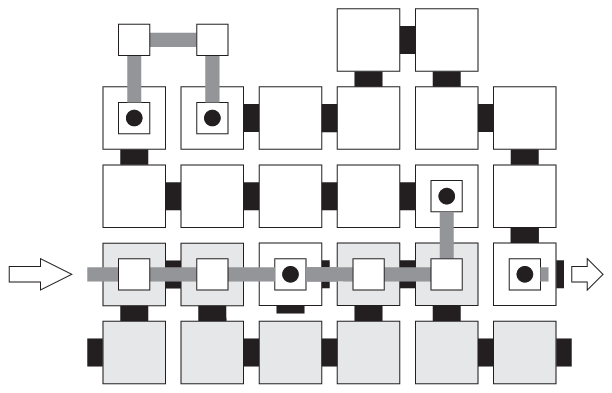}

(a) {\tt LR\_receive\_state\_0\_0}: Accept a 0 bit, which is currently being read by the tape head, from the left and shift a 0 bit to the right.
\end{minipage}}
\fbox{
\begin{minipage}[t][45mm][b]{0.45\textwidth}
\centering
 \includegraphics[width=.63\textwidth]{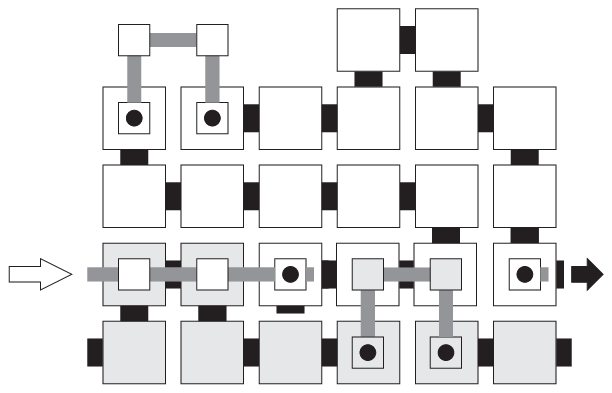}

(b) {\tt LR\_receive\_state\_0\_1}: Accept a 0 bit, which is currently being read by the tape head, from the left and shift a 1 bit to the right.
\end{minipage}}
\end{minipage}

\bigskip


\begin{minipage}[t]{\textwidth}
\centering
\fbox{
\begin{minipage}[t][45mm][b]{0.45\textwidth}
\centering
 \includegraphics[width=.63\textwidth]{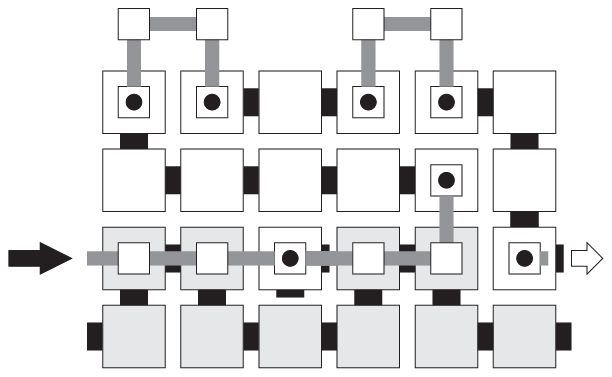}

(c) {\tt LR\_receive\_state\_1\_0}: Accept a 1 bit, which is currently being read by the tape head, from the left and shift a 0 bit to the right.
\end{minipage}}
\fbox{
\begin{minipage}[t][45mm][b]{0.45\textwidth}
\centering
 \includegraphics[width=.63\textwidth]{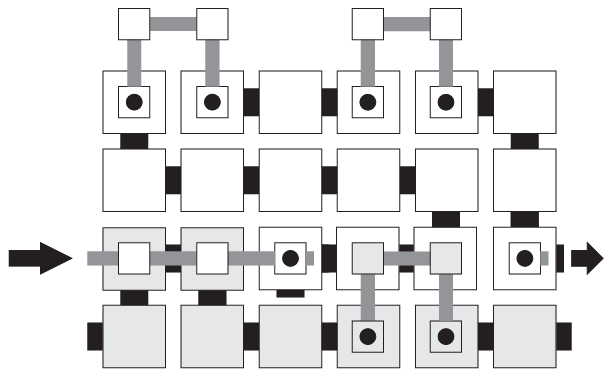}

(d) {\tt LR\_receive\_state\_1\_1}: Accept a 1 bit, which is currently being read by the tape head, from the left and shift a 1 bit to the right.
\end{minipage}}
\end{minipage}

\caption{The {\tt LR\_receive\_state} gadgets. These gadgets work exactly the same as the {\tt LR\_no\_state} and {\tt LR\_receive\_state} gadgets, but they receive the tape head from the left. Here, as is the case in Figure~\ref{fig:LR_send_state}, a long arrow represents a glue that encodes both a bit value as well as the presence of the tape head. These gadgets will bind to an {\tt LR\_send\_state} gadget.  }
\label{fig:LR_receive_state}
\vspace{0pt}
\end{figure}

\begin{figure}[htp]
\begin{minipage}{\textwidth}
\centering
\fbox{
 \includegraphics[width=0.27\textwidth]{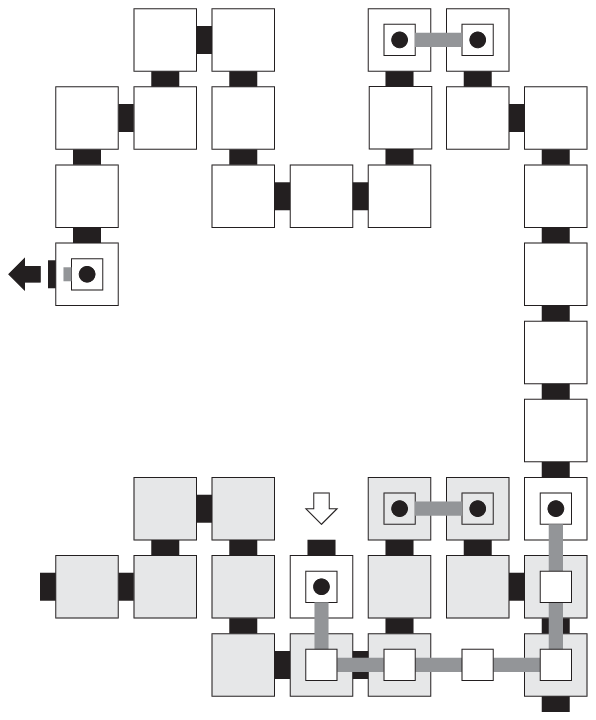}
}
\end{minipage}

\caption{{\tt LR\_rightmost\_no\_state}: terminate a left-to-right-assembling shift row and initiate a right-to-left-assembling copy row. Note that this gadget cannot receive a $1$ bit, thus, the input arrow is white. The fact that the output arrow is black serves no purpose. This gadget will bind to an {\tt LR\_no\_state} gadget.}
\label{fig:Post_process_lr_rightmost_no_state}
\vspace{0pt}
\end{figure}

\begin{figure}[htp] 
\centering
\begin{minipage}[t]{\textwidth}
\centering
\fbox{
\begin{minipage}[t][25mm][b]{.45\textwidth}
\vspace{0pt}
\centering
 \includegraphics[width=.55\textwidth]{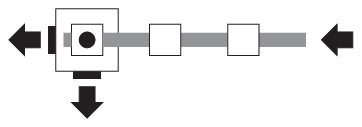}
\medskip
\begin{minipage}[t][][t]{\textwidth}
\centering\small
(a) {\tt RL\_guess\_0}: Guess a 0 bit.
\end{minipage}
\end{minipage}}
\fbox{
\begin{minipage}[t][25mm][b]{0.45\textwidth}
\centering
 \includegraphics[width=.55\textwidth]{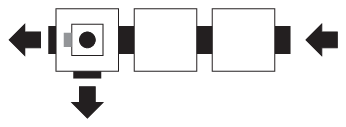}
\medskip
\begin{minipage}[t][][t]{\textwidth}
\centering\small
(b) {\tt RL\_guess\_1}: Guess a 1 bit.
\end{minipage}
\end{minipage}}
\end{minipage}
\caption{The {\tt RL\_guess} gadgets. These gadgets guess the binary value of a tape cell in a copy row of gadgets, which assemble from right-to-left. They work exactly the same as the {\tt RL\_guess} gadgets from Figure~\ref{fig:RL_guess}. In this figure, the color of the arrows serves no purpose. These gadgets will bind to either a {\tt RL\_copy} or {\tt LR\_rightmost\_no\_state} gadget.}
\label{fig:RL_post_process_guess}
\end{figure}

\begin{figure}[htp]
\centering


\begin{minipage}[t]{\textwidth}
\centering

\fbox{
\begin{minipage}[t][45mm][b]{0.45\textwidth}
\centering
 \includegraphics[width=.55\textwidth]{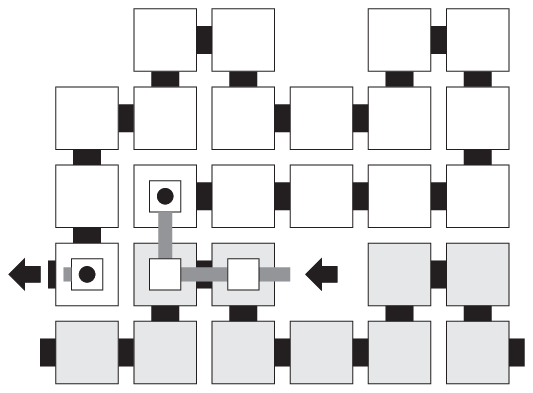}

(a) {\tt RL\_copy\_no\_state\_0}: Copy a 0 bit that is not currently being read by the tape head.
\end{minipage}}
\fbox{
\begin{minipage}[t][45mm][b]{0.45\textwidth}
\centering
 \includegraphics[width=.55\textwidth]{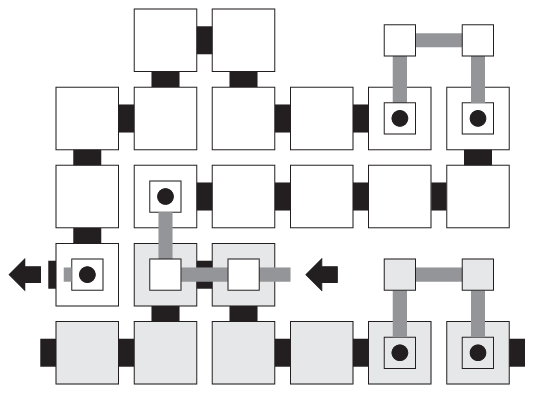}

(b) {\tt RL\_copy\_no\_state\_1}: Copy a 1 bit that is not currently being read by the tape head.
\end{minipage}}
\end{minipage}

\bigskip


\begin{minipage}[t]{\textwidth}
\centering
\fbox{
\begin{minipage}[t][45mm][b]{0.45\textwidth}
\centering
 \includegraphics[width=.55\textwidth]{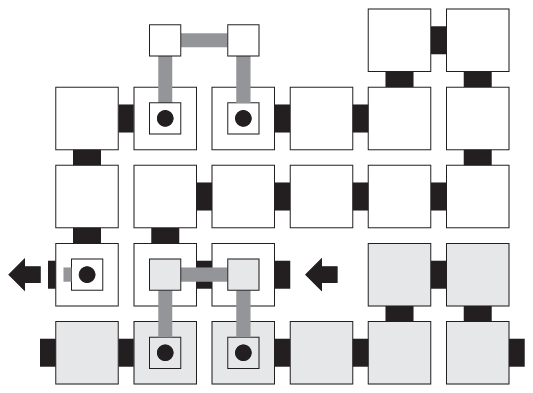}

(c) {\tt RL\_copy\_state\_0}: Copy a 0 bit that is currently being read by the tape head.
\end{minipage}}
\fbox{
\begin{minipage}[t][45mm][b]{0.45\textwidth}
\centering
 \includegraphics[width=.55\textwidth]{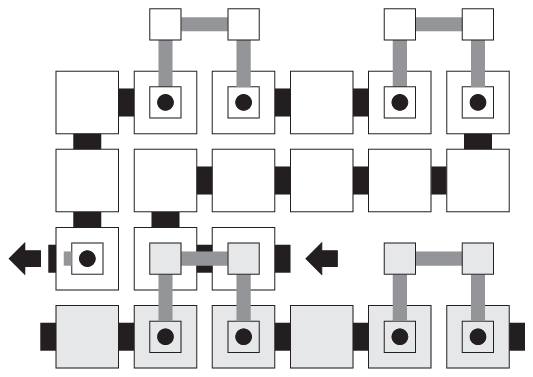}

(d) {\tt RL\_copy\_state\_1}: Copy a 1 bit that is currently being read by the tape head.
\end{minipage}}
\end{minipage}

\caption{The {\tt RL\_copy} gadgets. These gadgets assemble right-to-left while copying bits straight up without shifting. In this figure, the color of the arrows serves no purpose. These gadgets will bind to an {\tt RL\_guess} gadget.}
\label{fig:RL_copy}
\vspace{0pt}
\end{figure}

\begin{figure}[htp]
\begin{minipage}{\textwidth}
\centering
\fbox{
 \includegraphics[width=0.68\textwidth]{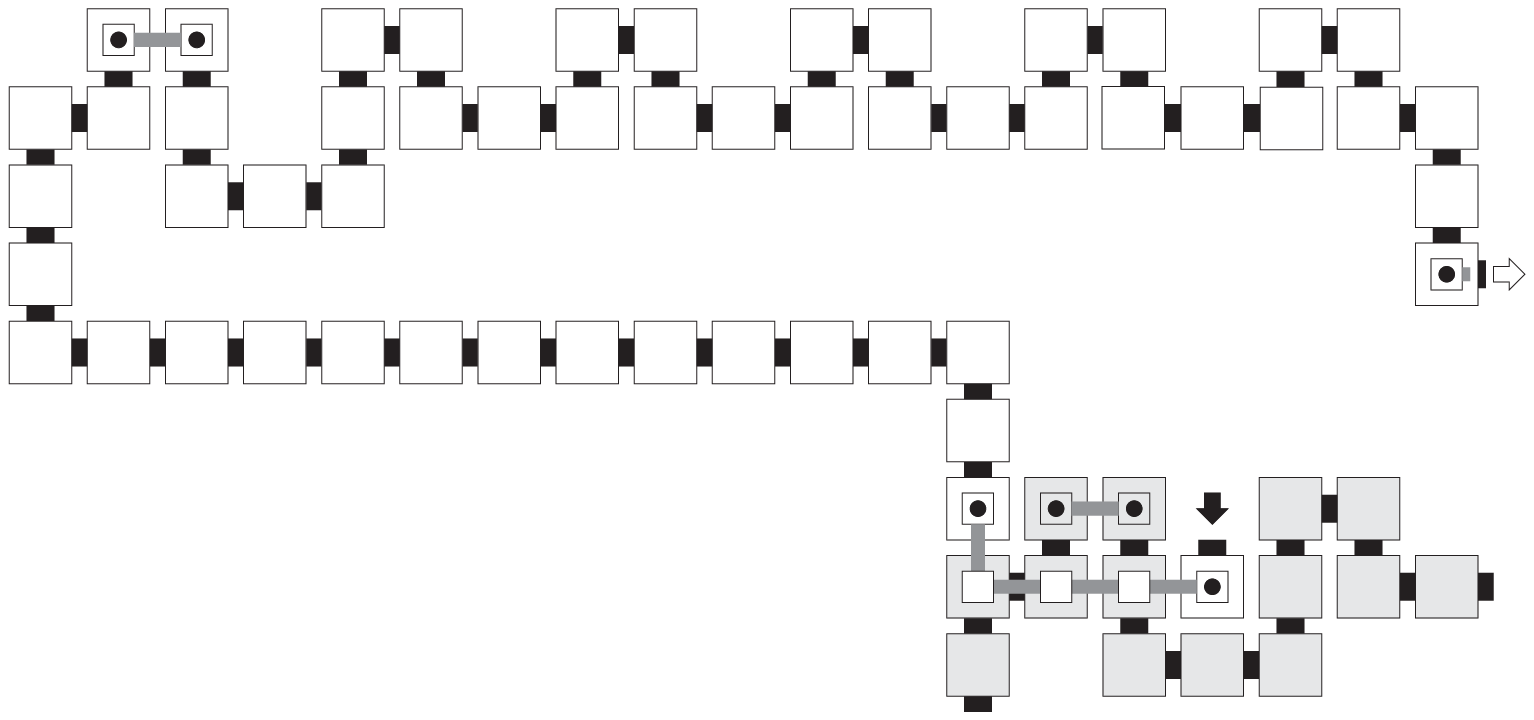}
}
\end{minipage}

\caption{The {\tt RL\_leftmost\_grow\_tape} gadget. This gadget grows the tape to the left by two blank tape cells. The color of the input arrow serves no purpose. However, the color of the output arrow is white because this gadget initiates the next shift row and therefore shifts a 0 bit to the right. This gadget will bind to a {\tt RL\_copy} gadget. The color of the input arrow of this gadget servers no purpose. However, the output arrow is white because a 0 bit is being shifted to the right. }
\label{fig:RL_leftmost_grow_tape}
\vspace{0pt}
\end{figure}

\begin{figure}[htp] 
\centering
\begin{minipage}[t]{\textwidth}
\centering
\fbox{
\begin{minipage}[t][70mm][b]{.45\textwidth}
\vspace{0pt}
\centering
 \includegraphics[width=.65\textwidth]{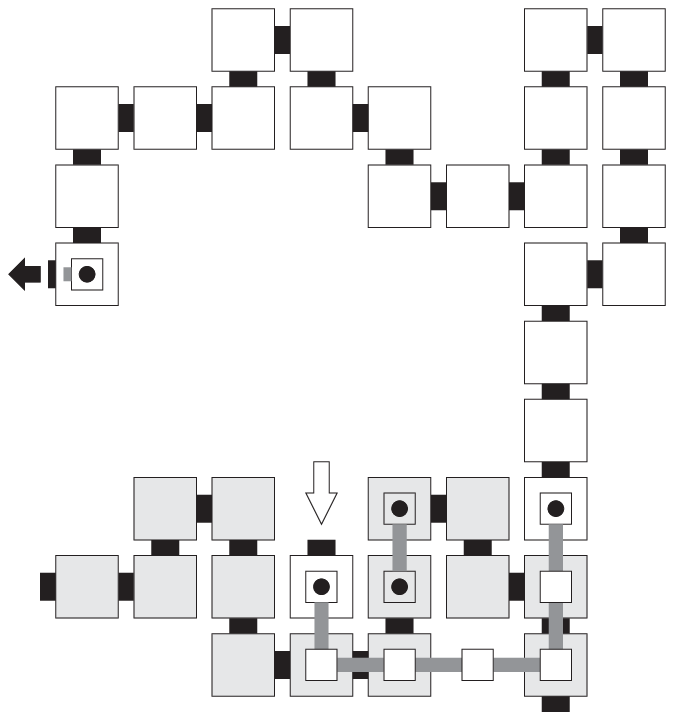}
\medskip
\begin{minipage}[t][][t]{\textwidth}
\centering\small
(a) {\tt LR\_right\_receive\_state\_0}: Receive a 0 bit, which is currently being read by the tape head, into the rightmost position.
\end{minipage}
\end{minipage}}
\fbox{
\begin{minipage}[t][70mm][b]{0.45\textwidth}
\centering
 \includegraphics[width=.65\textwidth]{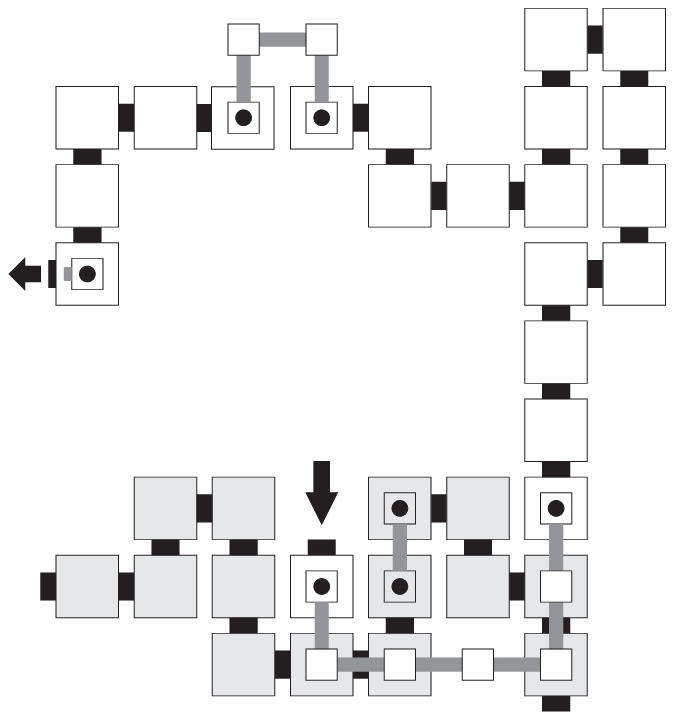}
\medskip
\begin{minipage}[t][][t]{\textwidth}
\centering\small
(b) {\tt LR\_right\_receive\_state\_1}: Receive a 1 bit, which is currently being read by the tape head, into the rightmost position.
\end{minipage}
\end{minipage}}
\end{minipage}
\caption{The {\tt LR\_from\_receive\_state} gadgets. These gadgets receive the tape head (along with the bit value that it is currently reading) into the rightmost gadget position. These gadgets initiate the final gadget row in the post-processing phase, which assembles right-to-left. These gadgets will initially bind to -- and receive the tape head from -- a {\tt LR\_send\_state} gadget. In this figure, a black arrow is carrying a 1 bit, a white arrow is carrying a 0 bit and a long arrow represents the presence of the tape head. The color/shape of the output arrow of this gadget serves no purpose. }
\label{fig:LR_rightmost_receive_state}
\end{figure}

\begin{figure}[htp]
\centering


\begin{minipage}[t]{\textwidth}
\centering

\fbox{
\begin{minipage}[t][45mm][b]{0.45\textwidth}
\centering
 \includegraphics[width=.55\textwidth]{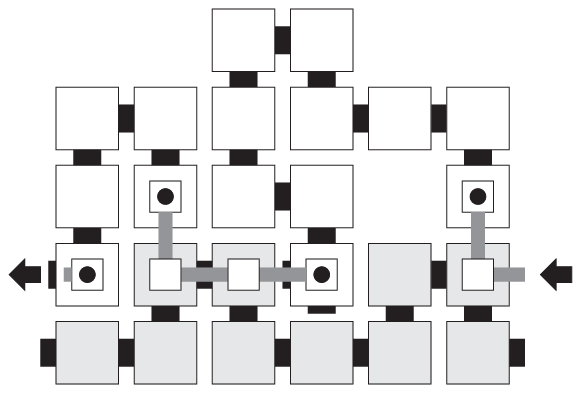}

(a) {\tt Two\_to\_one\_no\_state\_0}: Erase the tape head bit-bump from a gadget that represents a 0 bit that is not currently being read by the tape head.
\end{minipage}}
\fbox{
\begin{minipage}[t][45mm][b]{0.45\textwidth}
\centering
 \includegraphics[width=.55\textwidth]{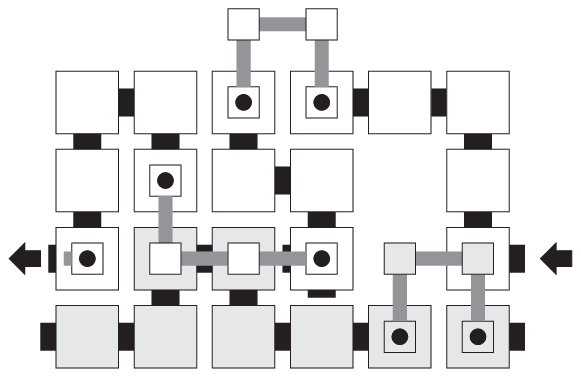}

(b) {\tt Two\_to\_one\_no\_state\_1}: Erase the tape head bit-bump from a gadget that represents a 1 bit that is not currently being read by the tape head.
\end{minipage}}
\end{minipage}

\bigskip


\begin{minipage}[t]{\textwidth}
\centering
\fbox{
\begin{minipage}[t][45mm][b]{0.45\textwidth}
\centering
 \includegraphics[width=.55\textwidth]{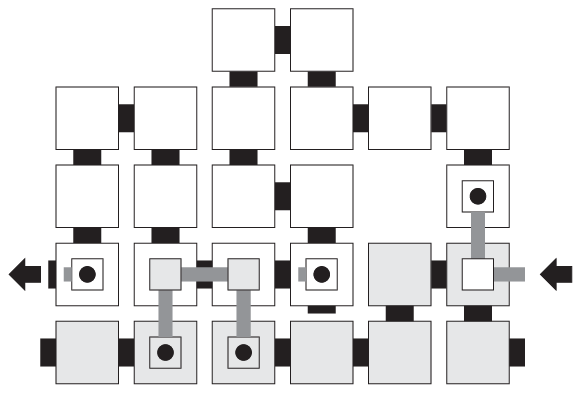}

(c) {\tt Two\_to\_one\_state\_0}: Erase the tape head bit-bump from a gadget that represents a 0 bit that is currently being read by the tape head.
\end{minipage}}
\fbox{
\begin{minipage}[t][45mm][b]{0.45\textwidth}
\centering
 \includegraphics[width=.55\textwidth]{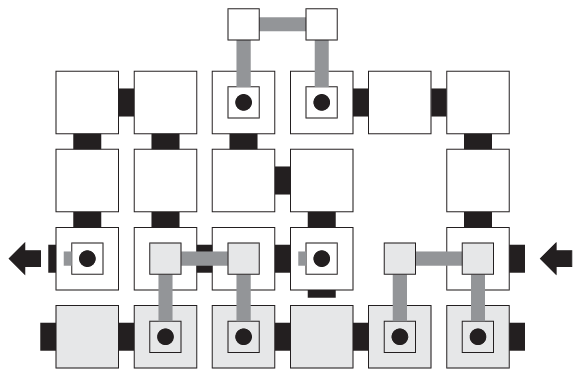}

(d) {\tt Two\_to\_one\_state\_1}: Erase the tape head bit-bump from a gadget that represents a 1 bit that is currently being read by the tape head.
\end{minipage}}
\end{minipage}

\caption{The {\tt Two\_to\_one} gadgets.  The final (topmost) row of gadgets in the post-processing phase erases the tape head indicator from each tape cell gadget, leaving only the bit value in the middle of the gadget. One of these gadgets will initially bind to a {\tt LR\_rightmost\_receive\_state} gadget. Then, these gadgets will bind to other copies of {\tt Two\_to\_one} gadgets. The color of the arrows serves no purpose in this figure. }
\label{fig:Two_to_one}
\vspace{0pt}
\end{figure}

\begin{figure}[htp]
\begin{minipage}{\textwidth}
\centering
\fbox{
 \includegraphics[width=0.2\textwidth]{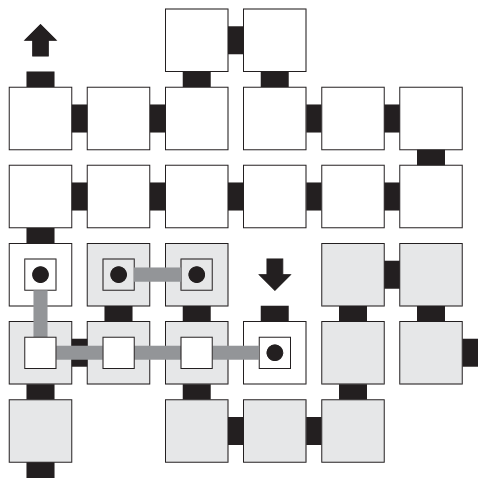}
}
\end{minipage}

\caption{The {\tt Two\_to\_one\_leftmost} gadget. Erase the tape head bit-bump from the leftmost gadget, which will always contain a 0 bit that is not being read by the tape head. The color of the arrows serves no purpose in this figure. }
\label{fig:Two_to_one_leftmost}
\vspace{0pt}
\end{figure}

%% file: appendix.tex
\clearpage
\section{Appendix: Growth blocks}
\label{sec:appendix}

In this section, we provide details of our construction for the growth
blocks. For all gadgets described in this section, we assume that the
input side of the growth block to which the gadget belongs is its
south side. Therefore, for completeness, we have to build three
additional copies of each gadget, naemly after a 90-degree, a
180-degree, and a 270-degree rotation, respectively.  Throughout this
section, we assume that the glues on each tile are implicitly defined
to ensure deterministic assembly.

\subsection{Gadgets used by all three types of growth blocks}\label{app:shared_gadgets}

Recall that the assembly of each growth block starts in its
bottom-left corner. The gadgets in this section all assemble from left
to right on the bottom row of the growth block until the first actual
move of the remaining path (that is, the moves in the Hamiltonian
cycle that have not been performed and erased by previous growth
blocks) is found. Each move in the Hamiltonian path (as well as each
one of the leading no-moves) is encoded in two bits, each of which
takes up 6 consecutive tiles. Therefore, most gadgets in this
construction are 12 tiles wide. However, the ``edge gadgets'' (i.e.,
those on the left and right edges of the growth block) take up two
additional tiles.

Figure~\ref{fig:gadget_BL} depicts the {\tt BL} (for ``bottom-left'')
gadget whose assembly is initiated by the only output glue in the
top-left corner of the previous growth block, or the seed block if the
current block is the first growth block in the Hamiltonian cycle.  The
{\tt BL} gadget, like most gadgets in this subsection, is 6 tiles
high.  However, this gadget is the only gadget in our growth-block
construction that is 26 tiles wide, since it encodes two no-moves,
which we assume are always present at the beginning of the Hamiltonian
cycle and it is an edge gadget. This gadget advertises two no-moves on
its top edge and its output glue encodes the fact that the first
actual move in the remaining path has not yet been found. Finally,
this gadget also advertises a diagonal marker in the rightmost no-move
on its top side. This diagonal marker will be shifted to the right by
later gadgets in order to guarantee that each growth block has a
square shape.

\begin{figure}[htp]
\begin{minipage}{\textwidth}
\centering
\fbox{
 \includegraphics[width=\textwidth]{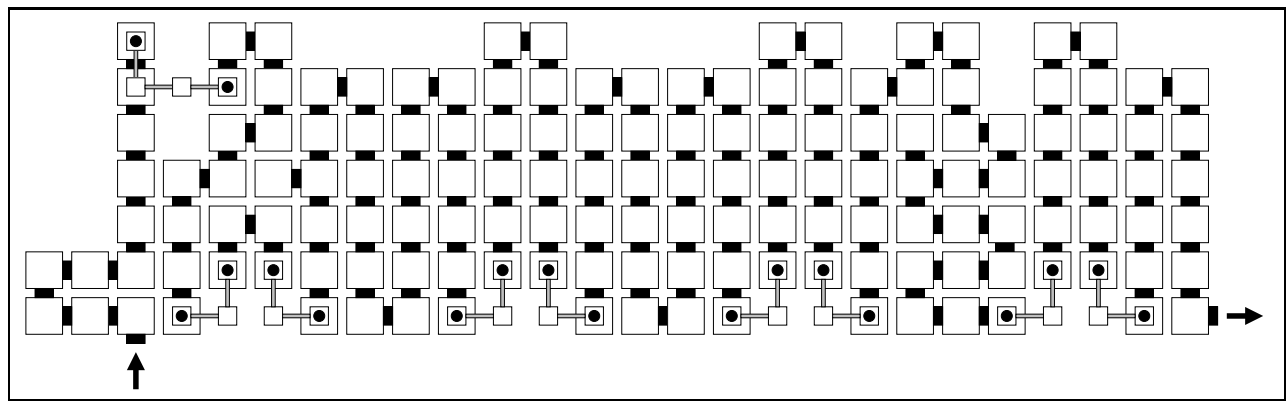} }
\end{minipage}
\caption{Bottom-left gadget ({\tt BL})}
\label{fig:gadget_BL}
\vspace{0pt}
\end{figure}

Figure~\ref{fig:gadget_before_first_move} depicts the gadget that
copies all of the additional no-moves that are included in front of
the Hamiltonian cycle. Since this gadget detects two
geometrically-encoded 0 bits on its bottom side, it also advertises
a no-move on its top side. Its input and output glues encode the fact that
the first actual move in the remaining path has not yet been found.

\begin{figure}[htp]
\begin{minipage}{\textwidth}
\centering
\fbox{
 \includegraphics[width=0.6\textwidth]{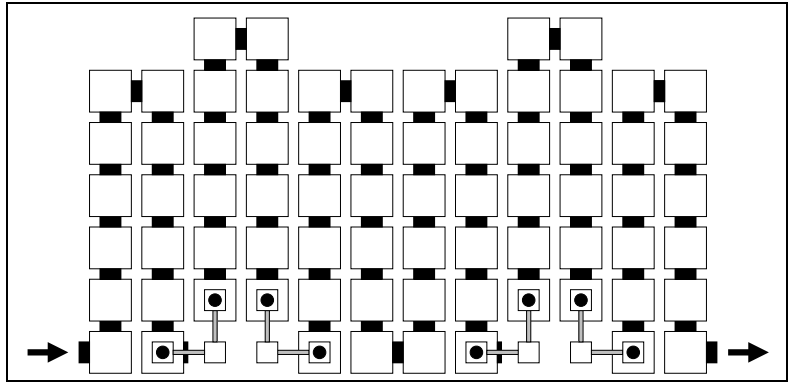}
}
\end{minipage}
\caption{Bottom-row gadget ({\tt before\_first\_move}) that looks for the first non-empty move in the remaining path}
\label{fig:gadget_before_first_move}
\vspace{0pt}
\end{figure}

Figure~\ref{fig:gadget_first_move} depicts the gadget that finds the
first move in the remaining path. This gadget erases this first move
by advertising a no-move on its top side, since this move is being
implemented in the current growth block. Its output glue, as well as
all of the input and output glues in the remaining gadgets will encode
the type of move (straight, left or right) that this gadget detects.

\begin{figure}[htp]
\begin{minipage}{\textwidth}
\centering
\fbox{
 \includegraphics[width=0.6\textwidth]{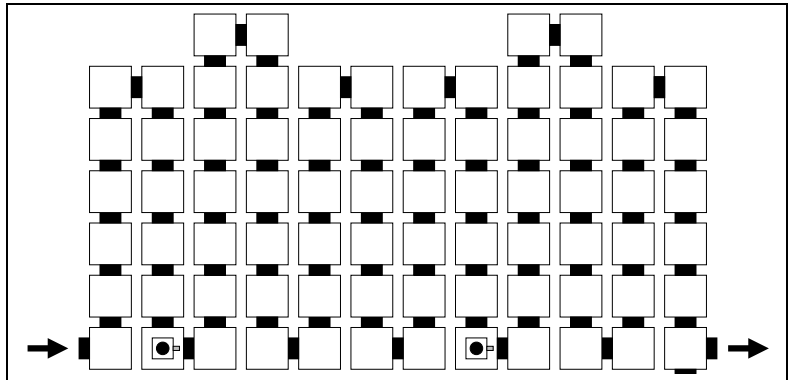}
}

a) The first move is a straight move
\end{minipage}\medskip

\begin{minipage}{\textwidth}
\centering
\fbox{
 \includegraphics[width=0.6\textwidth]{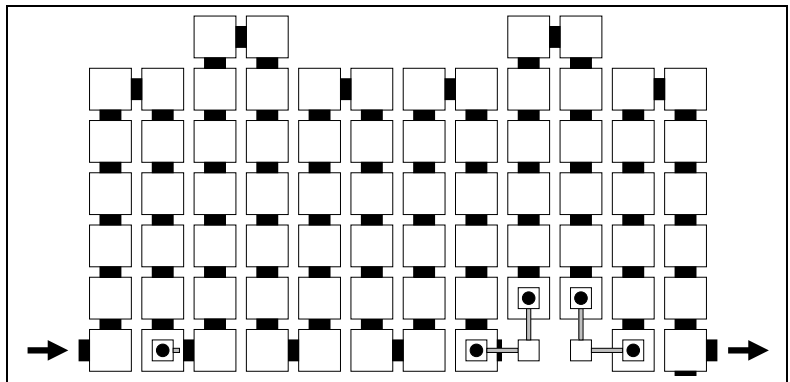}
}

b) The first move is a left turn
\end{minipage}\medskip

\begin{minipage}{\textwidth}
\centering
\fbox{
 \includegraphics[width=0.6\textwidth]{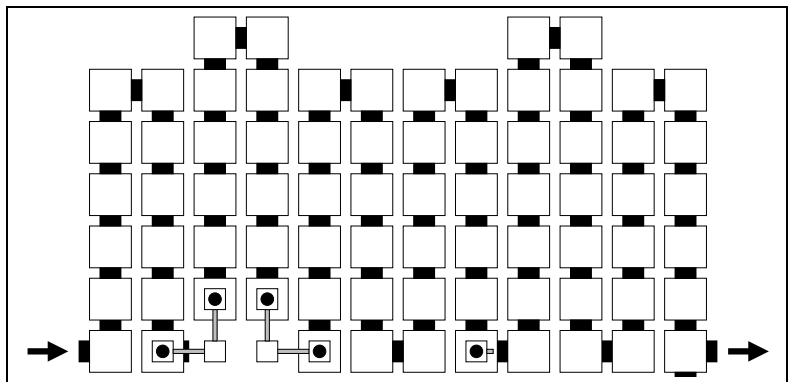}
}

c) The first move is a right turn
\end{minipage}

\caption{The three configurations of the bottom-row gadget ({\tt
    first\_move}) that finds and erases the first non-empty move in
  the remaining path}
\label{fig:gadget_first_move}
\vspace{0pt}
\end{figure}

\subsection{Growth block with a straight move}

The set of gadgets described in this section take over after the {\tt
first\_move} gadget and assemble to complete the growth block in the
case where the first move in the remaining path is a straight
move. Figure~\ref{fig:Sgrowth_construction} depicts the overall
construction for such a block.


\begin{figure}[htp]
\begin{minipage}{\textwidth}
\centering
 \includegraphics[width=0.9\textwidth]{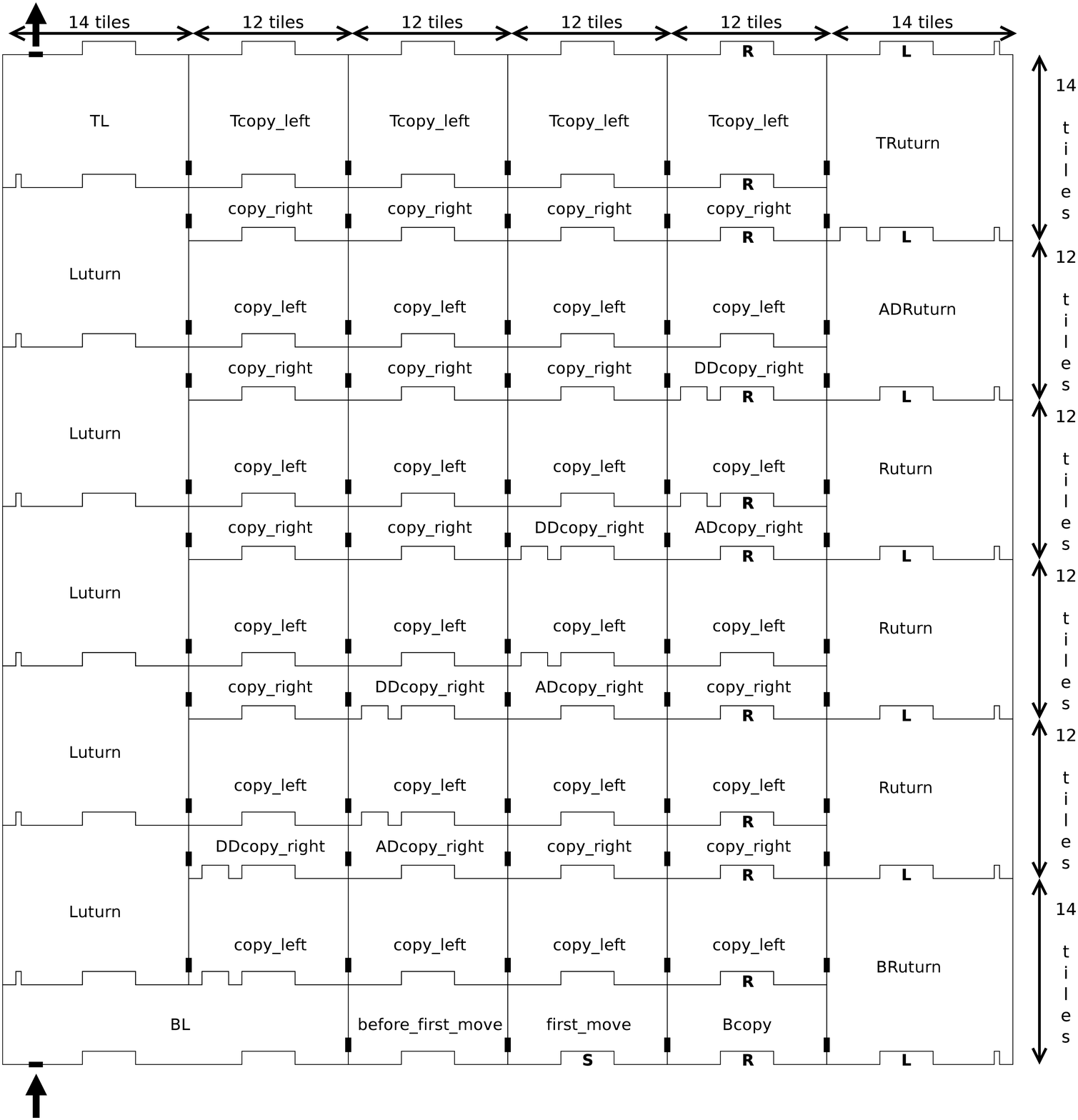}
\end{minipage}
\caption{Overall construction of a growth block whose input path
  starts with a straight move}
\label{fig:Sgrowth_construction}
\vspace{0pt}
\end{figure}

\subsubsection{Gadgets used in the bottom row of straight-move blocks}\label{app:Smove_gadgets_bottom_row}

Figure~\ref{fig:gadget_Sblock_Bcopy} depicts the bottom-row gadget
that copies all of the remaining moves, except for the first one,
which was erased by the {\tt first\_move} gadget, and the last one,
which is handled by the next gadget (see
Figure~\ref{fig:gadget_Sblock_BRuturn}). This gadget simply advertises
on its top side the move that is geometrically encoded on its bottom
side. Recall that its output glue encodes the type of the first move
and therefore the type of the growth block as a whole. Even though
this information is encoded in the glues of all remaining gadgets, we
will not repeat that fact in the following descriptions.

\begin{figure}[htp]
\begin{minipage}{\textwidth}
\centering
\fbox{
 \includegraphics[width=0.6\textwidth]{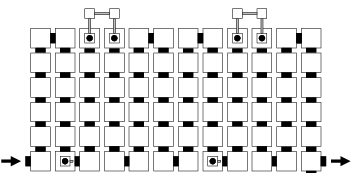}
}

a) The current move is a straight move
\end{minipage}\medskip

\begin{minipage}{\textwidth}
\centering
\fbox{
 \includegraphics[width=0.6\textwidth]{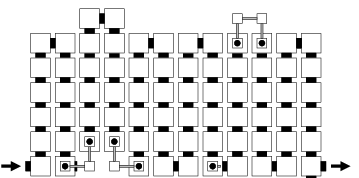}
}

b) The current move is a right turn
\end{minipage}\medskip

\begin{minipage}{\textwidth}
\centering
\fbox{
 \includegraphics[width=0.6\textwidth]{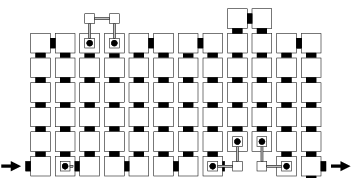}
}

c) The current move is a left turn
\end{minipage}

\caption{The three configurations of the bottom-row gadget ({\tt
    Bcopy}) that copies the remaining path, one move at a time from
  left to right, except for its first and last move}
\label{fig:gadget_Sblock_Bcopy}
\vspace{0pt}
\end{figure}

Figure~\ref{fig:gadget_Sblock_BRuturn} depicts the bottom-row gadget
that detects the right edge of the block or, equivalently the last
move in the remaining path. This {\tt BRuturn} gadget, which appears
in the bottom-right corner of the block, concludes the left-to-right
assembly of the bottom row and initiates the right-to-left assembly of
the next row of gadgets. Since the current block is a straight move,
the right edge of the block is a straight line. The top side of this
(copy) gadget advertises the move that it decodes on its bottom
side. While Figure~\ref{fig:gadget_Sblock_BRuturn} shows a left turn
being copied, two similar gadgets are needed for the cases where the
last move is a right turn or a straight move. This gadget is
differentiated from the {\tt Bcopy} gadget by the one-tile deep dent
in its bottom-right corner that is the marker we use for the right edge
of the block and the reason why we use two extra tiles on the edge
gadgets.

\begin{figure}[htp]
\begin{minipage}{\textwidth}
\centering
\fbox{
 \includegraphics[width=0.55\textwidth]{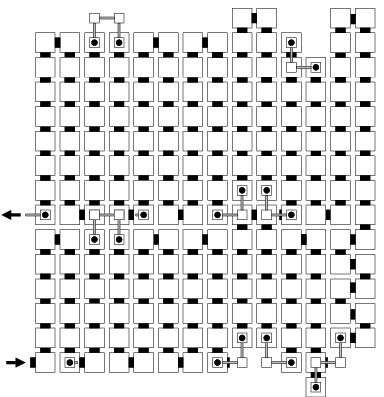}
}
\end{minipage}

\caption{The bottom-row rightmost gadget {\tt BRuturn} that copies the
  last move in the remaining path (in this case, a left turn)  and
  initiates the first row of {\tt copy\_left} gadgets}
\label{fig:gadget_Sblock_BRuturn}
\vspace{0pt}
\end{figure}

Note that some of the gadgets described in this section, even though
they have different names, are really different configurations of the
same gadget in that they are identical in the first ``half'' of the
gadget and then take different paths based on a geometric feature that
is detected later on during the assembly of the gadget. For example,
the {\tt before\_first\_move} and {\tt first\_move} gadgets are really
configurations of the same gadget that differ based on whether a
no-move or an actual move is detected, respectively. Similarly, the
{\tt BRuturn} and {\tt Bcopy} gadgets are really
configurations of the same gadget that differ based on whether the
right edge marker is detected or not, respectively. In fact, the same
is true for all four of the previous gadgets, since the first move is
detected in the same gadget as the edge marker in the last growth
block. This observation applies to other groups of gadgets in this
section.

\subsubsection{Gadgets used in the right-to-left rows of straight-move blocks}\label{app:Smove_gadgets_RL}

Figures~\ref{fig:gadget_Sblock_copy_left}~and~\ref{fig:gadget_Sblock_copy_left2}
depict the gadget that copies the remaining moves (or the leading
no-moves) of the Hamiltonian cycle, one at a time, in right-to-left
order. This gadget has four configurations, namely one for each type
of move (or no-move) being copied. This gadget also detects that the
current move is not positioned on the bottom-left to top-right
diagonal of the current block.

\begin{figure}[htp]
\begin{minipage}{\textwidth}
\centering
\fbox{
 \includegraphics[width=0.55\textwidth]{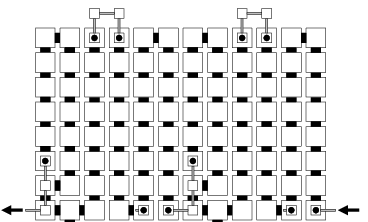}
}

a) The copied move is a straight move
\end{minipage}\bigskip

\begin{minipage}{\textwidth}
\centering
\fbox{
 \includegraphics[width=0.55\textwidth]{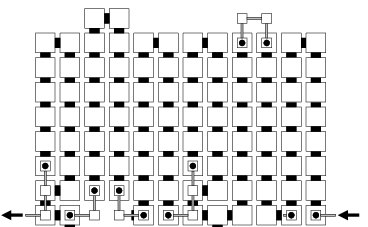}
}

b) The copied move is a right turn
\end{minipage}\medskip

\caption{The first two non-diagonal configurations of the {\tt
    copy\_left} gadget that copies the remaining path, one move at a
  time from right to left, except for the first and last moves}
\label{fig:gadget_Sblock_copy_left}
\vspace{0pt}
\end{figure}

\begin{figure}[htp]
\begin{minipage}{\textwidth}
\centering
\fbox{
 \includegraphics[width=0.55\textwidth]{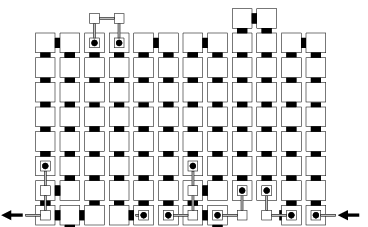}
}

c) The copied  move is a left turn
\end{minipage}\bigskip

\begin{minipage}{\textwidth}
\centering
\fbox{
 \includegraphics[width=0.55\textwidth]{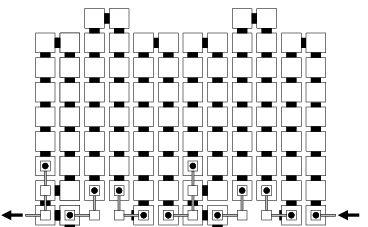}
}

d) There is no move to copy
\end{minipage}

\caption{The last two non-diagonal configurations of the {\tt
    copy\_left} gadget that copies the remaining path, one move at a
  time from right to left, except for the first and last moves}
\label{fig:gadget_Sblock_copy_left2}
\vspace{0pt}
\end{figure}

Figures~\ref{fig:gadget_Sblock_copy_leftD}~and~\ref{fig:gadget_Sblock_copy_leftD2}
depict the gadget that acts like the previous gadget but also detects
that the current move is positioned on the bottom-left to top-right
diagonal of the current block and therefore also advertises the
diagonal marker on its top side.

\begin{figure}[htp]
\begin{minipage}{\textwidth}
\centering
\fbox{
 \includegraphics[width=0.55\textwidth]{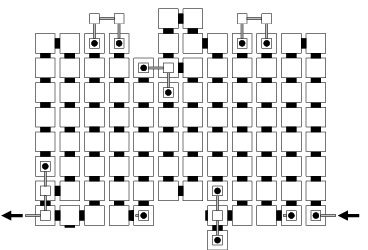}
}

a) The copied move is a straight move
\end{minipage}\bigskip

\begin{minipage}{\textwidth}
\centering
\fbox{
 \includegraphics[width=0.55\textwidth]{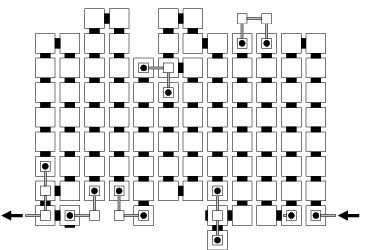}
}

b) The copied move is a right turn
\end{minipage}

\caption{The first two diagonal configurations of the {\tt
    copy\_left} gadget that copies the remaining path, one move at a
  time from right to left, except for the first and last moves}
\label{fig:gadget_Sblock_copy_leftD}
\vspace{0pt}
\end{figure}

\begin{figure}[htp]
\begin{minipage}{\textwidth}
\centering
\fbox{
 \includegraphics[width=0.55\textwidth]{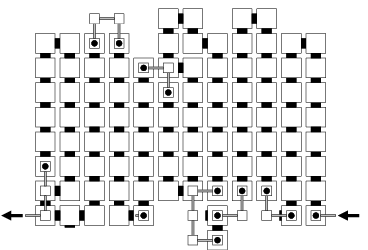}
}

c) The copied  move is a left turn
\end{minipage}\bigskip

\begin{minipage}{\textwidth}
\centering
\fbox{
 \includegraphics[width=0.55\textwidth]{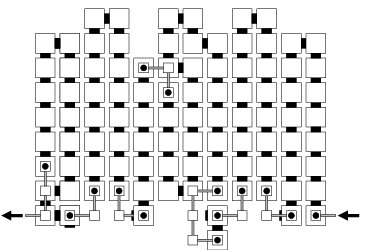}
}

d) There is no move to copy
\end{minipage}

\caption{The last two diagonal configurations of the {\tt
    copy\_left} gadget that copies the remaining path, one move at a
  time from right to left, except for the first and last moves}
\label{fig:gadget_Sblock_copy_leftD2}
\vspace{0pt}
\end{figure}

Figures~\ref{fig:gadget_Sblock_LuturnLowest}~and~\ref{fig:gadget_Sblock_Luturn}
depict the ``uturn'' gadget that is the counterpart on the left edge
of the block to the gadget depicted in
Figure~\ref{fig:gadget_Sblock_BRuturn}.  This gadget concludes the
assembly of the current right-to-left row and initiates the assembly
of the next left-to-right row of gadgets. Since the current block is a
straight move, the left edge of the block, and thus of this gadget, is
a straight line. The top side of this (copy) gadget advertises a
no-move since the move on the left edge of any block is always a
no-move. The reason this gadget has two configurations is because the
{\tt BL} gadget has a different top side than all of the other gadgets
on the left edge of the block (this is needed to accommodate the case
where the current block is a left turn block).
Figure~\ref{fig:gadget_Sblock_LuturnLowest} depicts the {\tt Luturn}
gadget that assembles right above the {\tt BL} gadget, while
Figure~\ref{fig:gadget_Sblock_Luturn} depicts the {\tt Luturn} gadget
that assembles in all higher rows of the block. Since the {\tt BL
gadget} has a deeper, two-tile-high dent in its top-left corner than
all other gadgets on the left edge of the block (whose dent in their
top-left corner is only one tile high), the gadgets in
Figures~\ref{fig:gadget_Sblock_LuturnLowest}~and~\ref{fig:gadget_Sblock_Luturn}
have slightly different configurations in their bottom-left corners.

\begin{figure}[htp]
\begin{minipage}{\textwidth}
\centering
\fbox{
 \includegraphics[width=0.55\textwidth]{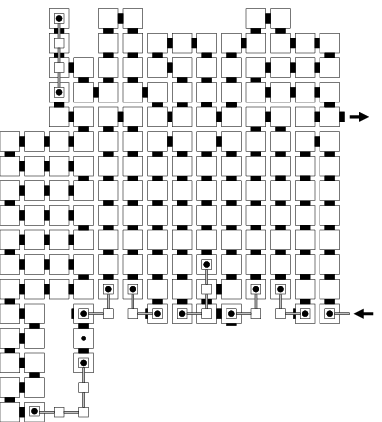}
}
\end{minipage}

\caption{Configuration of the leftmost-column gadget ({\tt Luturn}) that appears right above the {\tt BL} gadget}
\label{fig:gadget_Sblock_LuturnLowest}
\vspace{0pt}
\end{figure}

\begin{figure}[htp]
\begin{minipage}{\textwidth}
\centering
\fbox{
 \includegraphics[width=0.55\textwidth]{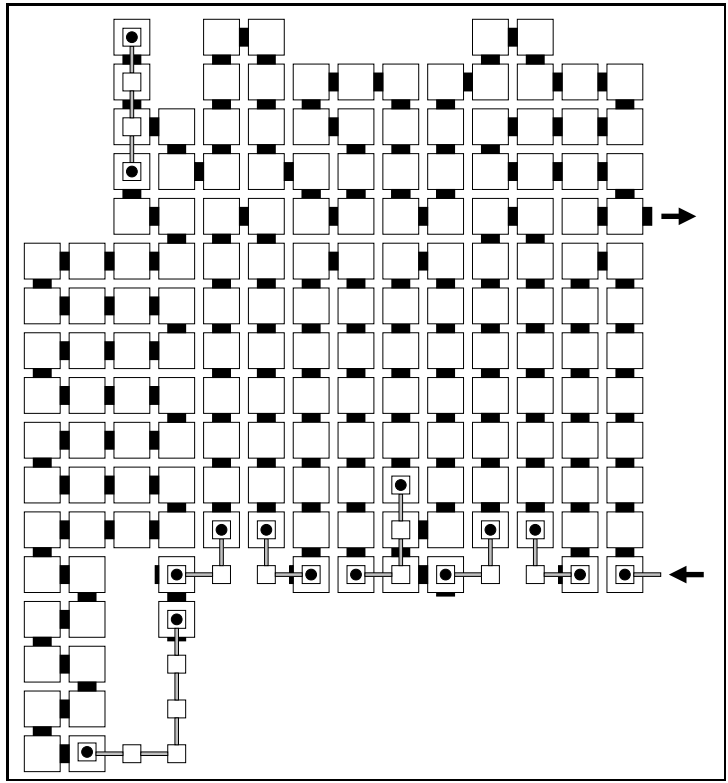}
}
\end{minipage}

\caption{Configuration of the leftmost-column gadget ({\tt Luturn}
  gadget {\it not} directly above the {\tt BL} gadget) that always
  copies the no-move marker and initiates the next row of {\tt
    copy\_right} gadgets}
\label{fig:gadget_Sblock_Luturn}
\vspace{0pt}
\end{figure}

\subsubsection{Gadgets used in the left-to-right rows of straight-move blocks}\label{app:Smove_gadgets_LR}

Figure~\ref{fig:gadget_Sblock_copy_right} depicts the gadget that
copies the moves and no-moves, one at a time, from left to right
between the edges of the block, as long as the current move is not
located on the bottom-left to top-right diagonal of the block.

\begin{figure}[htp]
\begin{minipage}{\textwidth}
\centering
\fbox{
 \includegraphics[width=0.55\textwidth]{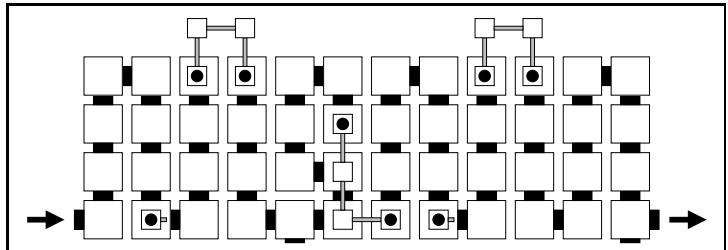}
}

a) The copied move is a straight move\bigskip

\fbox{
 \includegraphics[width=0.55\textwidth]{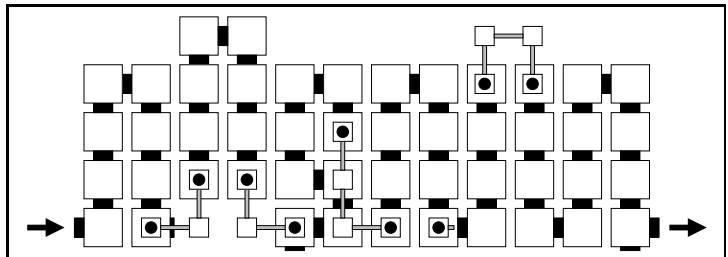}
}

b) The copied move is a right turn\bigskip

\fbox{
 \includegraphics[width=0.55\textwidth]{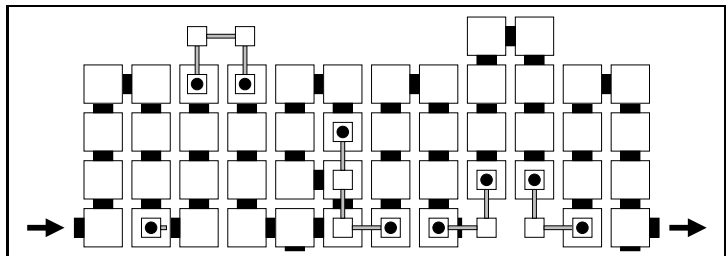}
}

c) The copied move is a left turn\bigskip

\fbox{
 \includegraphics[width=0.55\textwidth]{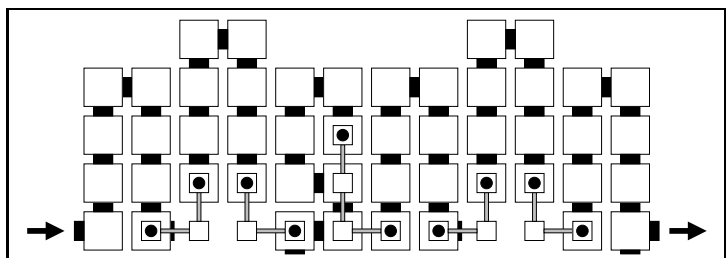}
}

d) There is no move to copy

\end{minipage}\medskip

\caption{The four configurations of the {\tt
    copy\_right} gadget that copies the remaining path, one move at a
  time from left to right, except for the first and last moves}
\label{fig:gadget_Sblock_copy_right}
\vspace{0pt}
\end{figure}

Figure~\ref{fig:gadget_Sblock_DDcopy_right} depicts the gadget that
behaves like the previous one but also detects the diagonal marker in
the current position. This gadget not only copies the current move to
its top side but also encodes in its output glue the fact that the
diagonal marker was just detected.

\begin{figure}[htp]
\begin{minipage}{\textwidth}
\centering
\fbox{
 \includegraphics[width=0.55\textwidth]{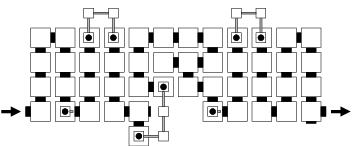}
}

a) The copied move is a straight move\bigskip

\fbox{
 \includegraphics[width=0.55\textwidth]{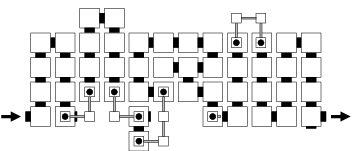}
}

b) The copied move is a right turn\bigskip

\fbox{
 \includegraphics[width=0.55\textwidth]{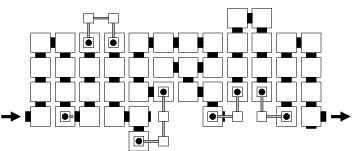}
}

c) The copied move is a left turn\bigskip

\fbox{
 \includegraphics[width=0.55\textwidth]{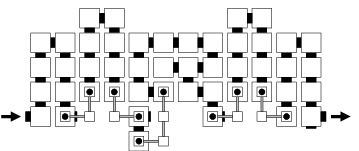}
}

d) There is no move to copy

\end{minipage}\medskip

\caption{The four configurations of the {\tt DDcopy\_right} gadget
  that detects the diagonal marker and copies the remaining path, one
  move at a time from left to right, except for the first and last
  moves}
\label{fig:gadget_Sblock_DDcopy_right}
\vspace{0pt}
\end{figure}

Figure~\ref{fig:gadget_Sblock_ADcopy_right} depicts the gadget that
not only copies the current move but also advertises the diagonal
marker on its top side. Therefore this gadget, together with the
previous one, makes sure that the diagonal marker is shifted by one
gadget position to the right in each left-to-right row.

\begin{figure}[htp]
\begin{minipage}{\textwidth}
\centering
\fbox{
 \includegraphics[width=0.55\textwidth]{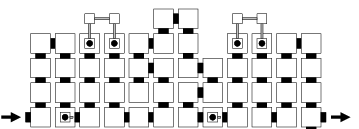}
}

a) The copied move is a straight move\bigskip

\fbox{
 \includegraphics[width=0.55\textwidth]{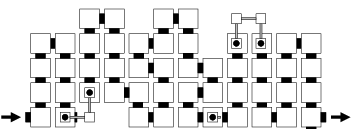}
}

b) The copied move is a right turn\bigskip

\fbox{
 \includegraphics[width=0.55\textwidth]{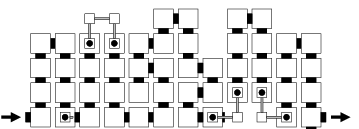}
}

c) The copied move is a left turn\bigskip

\fbox{
 \includegraphics[width=0.55\textwidth]{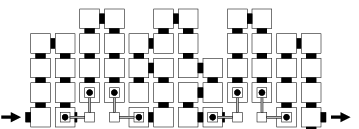}
}

d) There is no move to copy

\end{minipage}\medskip

\caption{The four configurations of the {\tt ADcopy\_right} gadget
  that advertises the diagonal marker and copies the remaining path, one
  move at a time from left to right, except for the first and last
  moves}
\label{fig:gadget_Sblock_ADcopy_right}
\vspace{0pt}
\end{figure}

Figure~\ref{fig:gadget_Sblock_Ruturn} depicts the {\tt Ruturn} gadget
that detects the right-edge marker and thus terminates the current
left-to-right row of gadgets and initiates the next right-to-left
row. The only difference between this gadget and the one shown in
Figure~\ref{fig:gadget_Sblock_BRuturn} is that the latter belongs to
the bottom row of the block and is thus two tiles higher.

\begin{figure}[htp]
\begin{minipage}{\textwidth}
\centering
\fbox{
 \includegraphics[width=0.55\textwidth]{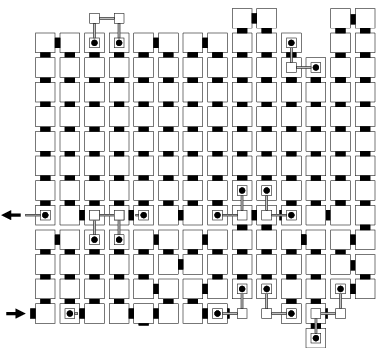}
}
\end{minipage}

\caption{The rightmost gadget {\tt Ruturn} that copies the
  last move (in this case, a left turn) in the remaining path and initiates
  the second and following rows of {\tt copy\_left} gadgets}
\label{fig:gadget_Sblock_Ruturn}
\vspace{0pt}
\end{figure}

Figure~\ref{fig:gadget_Sblock_ADRuturn} is similar to the previous
gadget except that its input glue binds to the output glue of the {\tt
DDcopy\_right} gadget. Therefore, this gadget must advertise the
diagonal marker on its top side, thereby detecting the fact that the
diagonal marker has reached the right edge of the block. As result,
the next ``uturn'' gadget that will bind to the top side of this gadget
will ``know'' to initiate the last, topmost row of the current block
(see next subsection).

While the last two figures show a left turn being copied (as the last
move in the Hamiltonian cycle), two similar configurations are needed for
each gadget for the cases where the last move is a right turn or a
straight move.

\begin{figure}[htp]
\begin{minipage}{\textwidth}
\centering
\fbox{
 \includegraphics[width=0.55\textwidth]{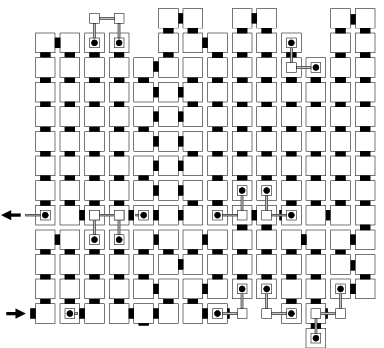}
}
\end{minipage}

\caption{The rightmost gadget {\tt ADRuturn} that copies the last move
  (in this case, a left turn) in the remaining path, advertises the diagonal
  marker and initiates the last row of {\tt
    copy\_left} gadgets}
\label{fig:gadget_Sblock_ADRuturn}
\vspace{0pt}
\end{figure}

\subsubsection{Gadgets used in the topmost row of straight-move blocks}\label{app:Smove_gadgets_topmost}

Figure~\ref{fig:gadget_Sblock_TRuturn} depicts the top-right ``uturn''
gadget that is identical to the one in
Figure~\ref{fig:gadget_Sblock_Ruturn}, except that it is two tiles
higher (since it belongs to the topmost row) and its output glue
encodes the fact that the topmost row needs to be constructed.

\begin{figure}[htp]
\begin{minipage}{\textwidth}
\centering
\fbox{
 \includegraphics[width=0.55\textwidth]{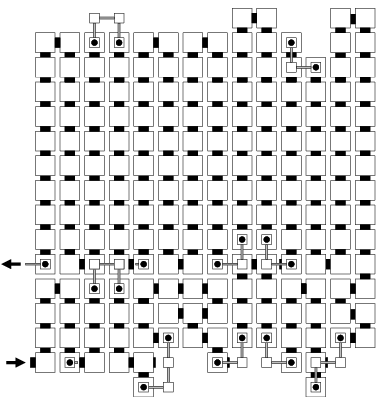}
}
\end{minipage}

\caption{The top-right gadget {\tt TRuturn} that copies the last move
  (in this case, a left turn) in the remaining path and initiates the row of {\tt
    Tcopy\_left} gadgets}
\label{fig:gadget_Sblock_TRuturn}
\vspace{0pt}
\end{figure}

We omitted figures for the topmost row, copy-left gadget (called {\tt
Tcopy\_left} in Figure~\ref{fig:Sgrowth_construction}) since it is
almost identical to the {\tt copy\_left} gadget shown in
Figures~\ref{fig:gadget_Sblock_copy_left}~and~\ref{fig:gadget_Sblock_copy_left2}
(it is two tiles higher and its input and output glues encode the fact
that it belongs to the topmost row).

Figure~\ref{fig:gadget_Sblock_TL} depicts the last gadget in the
construction of a straight move. This gadget assembles in the top-left
corner of the current block and has an output glue on its top side that will initiate the
assembly of the next growth block.

\begin{figure}[htp]
\begin{minipage}{\textwidth}
\centering
\fbox{
 \includegraphics[width=0.55\textwidth]{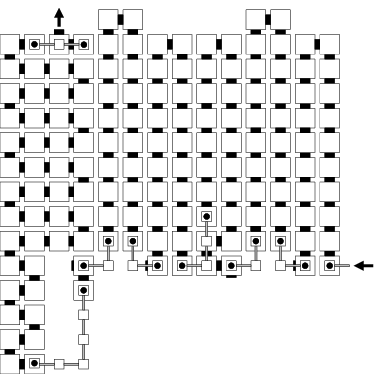}
}
\end{minipage}

\caption{The top-left gadget ({\tt TL}) that completes the current
  (straight) growth  block and initiates the next growth block}
\label{fig:gadget_Sblock_TL}
\vspace{0pt}
\end{figure}

\subsection{Growth block with a left turn}

The set of gadgets described in this section take over after the first
move gadget and assemble to complete the growth block in the case
where the first move in the remaining path is a left turn.
Figure~\ref{fig:Lgrowth_construction} depicts the overall
construction for such a block.

\begin{figure}[htp]
\begin{minipage}{\textwidth}
\centering
 \includegraphics[width=0.9\textwidth]{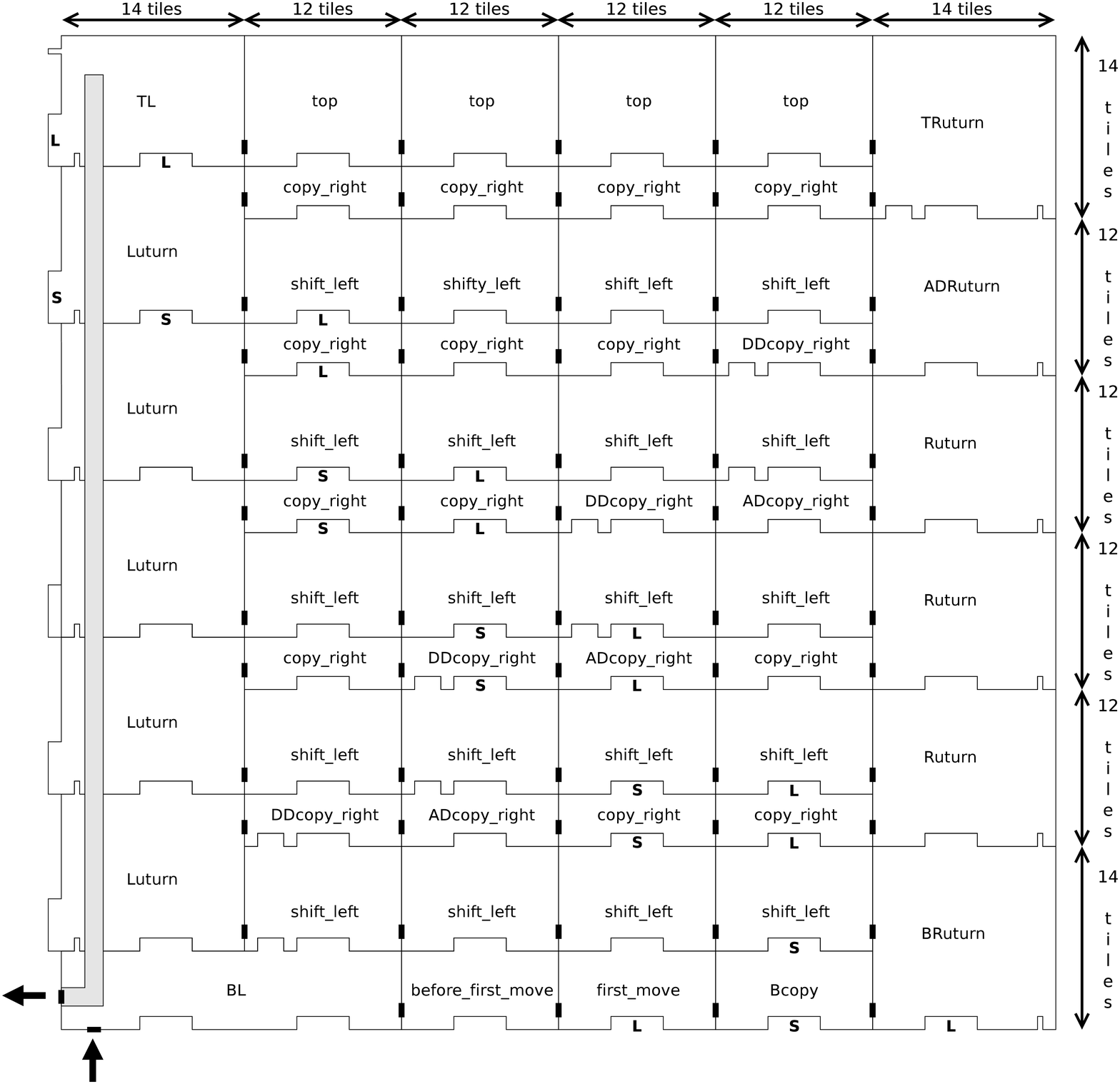}
\end{minipage}
\caption{Overall construction of a growth block whose input path
  starts with a left turn}
\label{fig:Lgrowth_construction}
\vspace{0pt}
\end{figure}

This set of gadgets differs from the one for the straight
move in two essential ways.

First, the moves in the remaining path need to be copied from the
bottom side to the left side of the block. This is accomplished by
replacing the {\tt copy\_left} gadget with a {\tt shift\_left} gadget,
whose generic function is depicted in
Figure~\ref{fig:gadget_Lblock_generic_shift_left}. Figure~\ref{fig:gadget_Lblock_shift_left} depicts two of the many
possible configurations for the {\tt shift\_left} gadget.

\begin{figure}[htp]
\begin{minipage}{\textwidth}
\centering
 \includegraphics[width=0.4\textwidth]{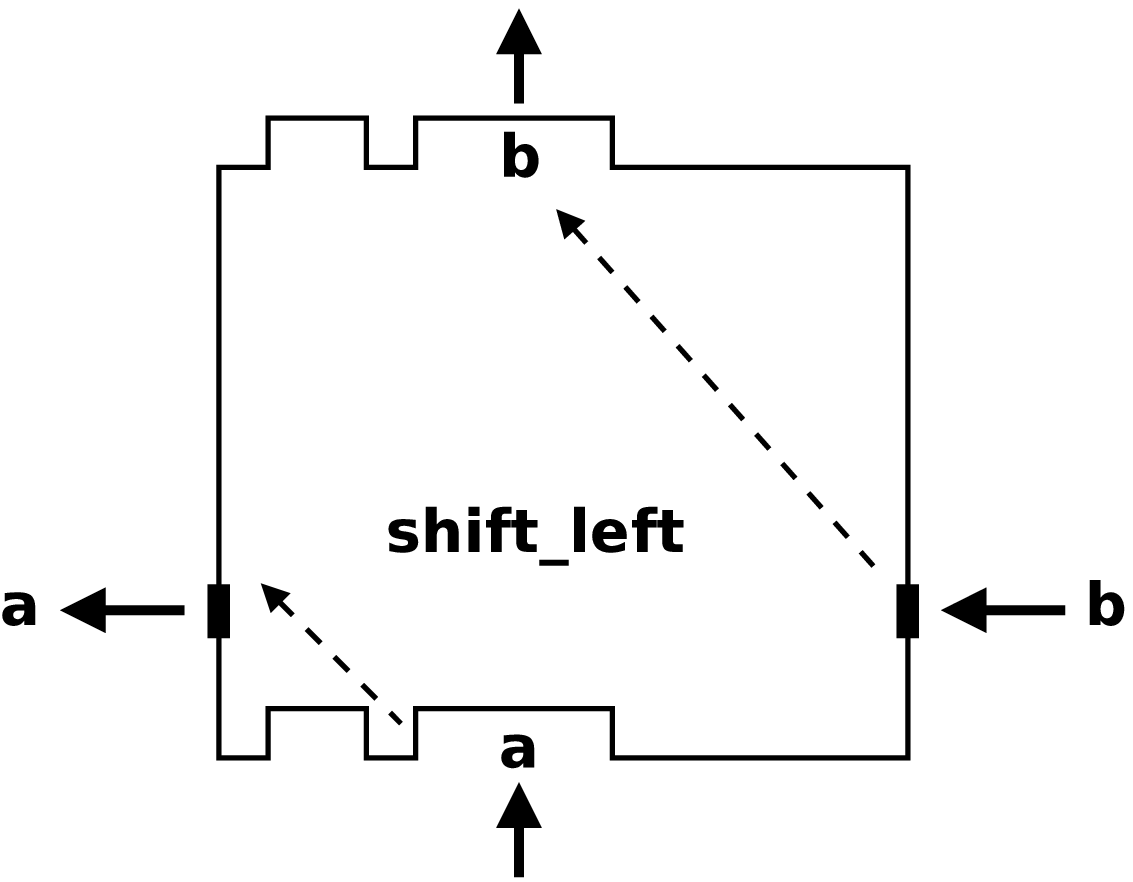}
\end{minipage}
\caption{Generic {\tt shift\_left} gadget where {\tt a} and {\tt b}
  each represent a type of move (straight move, left turn, right turn
  or no move) and the diagonal marker is optional. There are $26=13*2$
  configurations of the {\tt shift\_left} gadget, that is, a total of
  $4\times 4 = 16$ pairs of move types except for the three pairs in
  which {\tt b} is a ``no move'' and {\tt a} is not (since a ``no
  move'' cannot follow an actual move in the remaining path), each one with  or without a diagonal marker. Figure~\ref{fig:gadget_Lblock_shift_left} depicts two of these configurations.}
\label{fig:gadget_Lblock_generic_shift_left}
\vspace{0pt}
\end{figure}

\begin{figure}[htp]
\begin{minipage}{\textwidth}
\centering
\fbox{
 \includegraphics[width=0.55\textwidth]{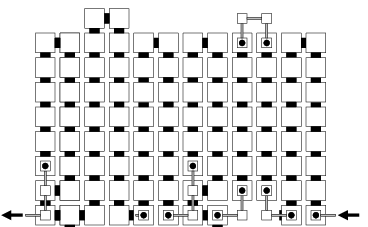}
}

a) Configuration in which {\tt a} is a left turn,  {\tt b} is a right turn, and the diagonal marker is absent
\end{minipage}\bigskip

\begin{minipage}{\textwidth}
\centering
\fbox{
 \includegraphics[width=0.55\textwidth]{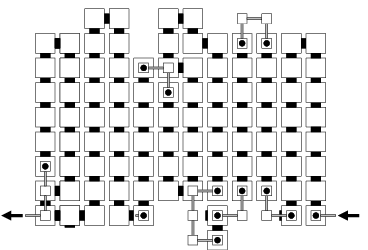}
}

b) Configuration in which {\tt a} is a left turn,  {\tt b} is a right turn, and the diagonal marker is present
\end{minipage}

\caption{Two of the 26 possible configurations of the {\tt shift\_left} gadget}
\label{fig:gadget_Lblock_shift_left}
\vspace{0pt}
\end{figure}

Second, left turn blocks differ significantly from the other kinds of
blocks on their left edge. They must not only advertise the remaining
path on this edge, they must also allow for the current block's output
glue (the one that will initiate the assembly of the next growth
block) to be put in place at the bottom of the left edge after the
current block is completely assembled. To achieve this goal, all of
the gadgets on the left edge of left turn blocks have been built as
follows. Figures~\ref{fig:gadget_Lblock_LuturnLowest}, \ref{fig:gadget_Lblock_Luturn}
and~\ref{fig:gadget_Lblock_Luturn2} show how the ``Luturn'' gadgets
now contain a one-tile-wide vertical path that is left empty until the
top-left gadget has completed. Finally,
Figure~\ref{fig:gadget_Lblock_TL} shows how the top-left gadget not
only completes the topmost row, as usual, but also leaves a dent at
the top of the left side of the current block (to mark the right edge
of the next growth block) and then builds a one-tile-wide path of
tiles (shown in light gray) that will position the current block's output glue to initiate
the assembly of the next growth block.

\begin{figure}[htp]
\begin{minipage}{\textwidth}
\centering
\fbox{
 \includegraphics[width=0.55\textwidth]{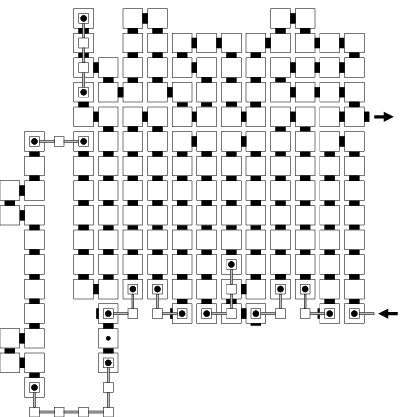}
}
\end{minipage}

\caption{The leftmost-column gadget ({\tt Luturn}) that appears right above the {\tt BL} gadget (the bottom pattern must be a ``no-move'' and the input glue must carry the ``no-move'' as well)}
\label{fig:gadget_Lblock_LuturnLowest}
\vspace{0pt}
\end{figure}

\begin{figure}[htp]
\begin{minipage}{\textwidth}
\centering
\fbox{
 \includegraphics[width=0.55\textwidth]{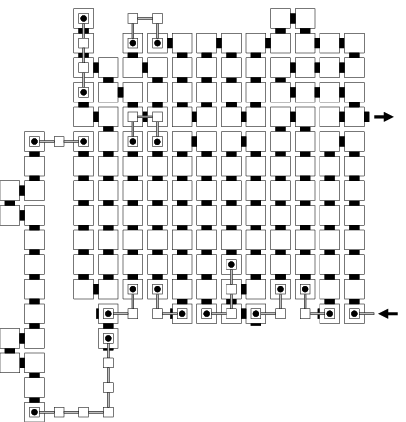}
}

a) There is no move to copy
\end{minipage}\bigskip

\begin{minipage}{\textwidth}
\centering
\fbox{
 \includegraphics[width=0.55\textwidth]{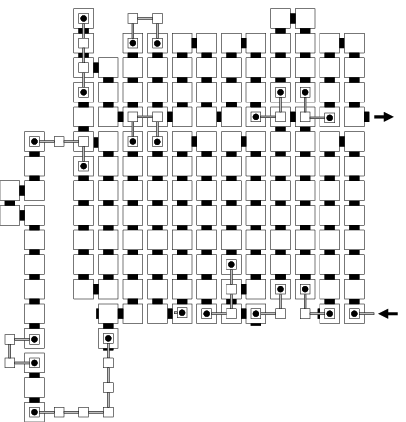}
}

b) The copied move is a left turn
\end{minipage}\medskip

\caption{The first two configurations of the leftmost-column gadget
  {\tt Luturn} ({\it not} appearing right above the {\tt BL} gadget)
when the input glue represents a left turn}
\label{fig:gadget_Lblock_Luturn}
\vspace{0pt}
\end{figure}

\begin{figure}[htp]
\begin{minipage}{\textwidth}
\centering
\fbox{
 \includegraphics[width=0.55\textwidth]{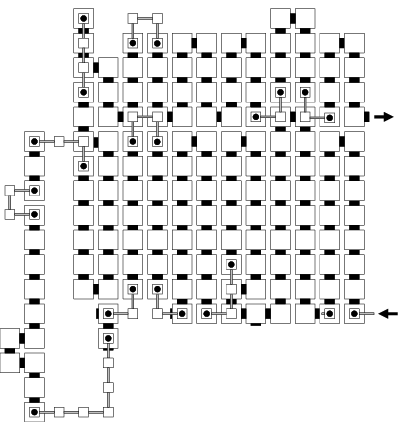}
}

c) The copied move is a right turn
\end{minipage}\bigskip

\begin{minipage}{\textwidth}
\centering
\fbox{
 \includegraphics[width=0.55\textwidth]{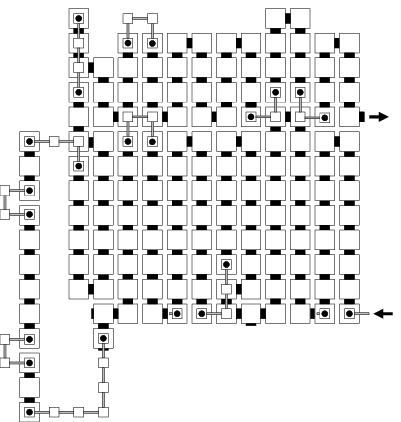}
}

d) The copied move is a straight move
\end{minipage}\medskip

\caption{The last two configurations of the leftmost-column gadget
  {\tt Luturn} ({\it not} appearing right above the {\tt BL} gadget)
when the input glue represents a left turn}
\label{fig:gadget_Lblock_Luturn2}
\vspace{0pt}
\end{figure}

\begin{figure}[htp]
\begin{minipage}{\textwidth}
\centering
\fbox{
 \includegraphics[width=0.55\textwidth]{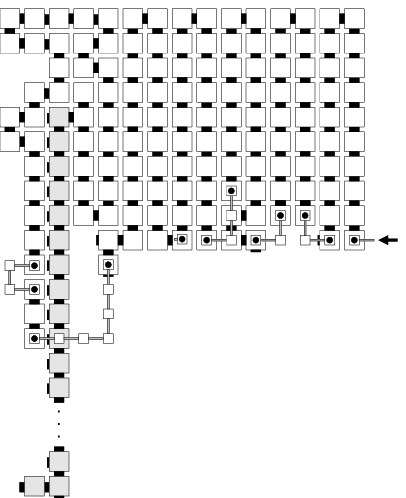}
}
\end{minipage}

\caption{The {\tt TL} gadget that first fills up the top-left corner
  of the current block and then puts the glue in its bottom-left corner that
  will initiate the next block}
\label{fig:gadget_Lblock_TL}
\vspace{0pt}
\end{figure}

\subsection{Growth block with a right turn}\label{app:Rmove_gadgets}

The overall construction for this type of growth block is shown in
Figure~\ref{fig:Rgrowth_construction} in the main body of this
paper. We omit figures for the individual gadgets in this type of
block because they are similar to the gadgets for the straight-move
block, except that the remaining path is advertised on the right side
instead of the top side and thus the {\tt copy\_right} gadgets are
replaced with {\tt shift\_right} gadgets which, in turn, are similar to
the {\tt shift\_left} gadgets in the left-move blocks (see Figures~\ref{fig:gadget_Lblock_generic_shift_left} and ~\ref{fig:gadget_Lblock_shift_left}).

\subsection{Last growth block}\label{app:last_block}

Recall that the growth block corresponding to the last move in the
Hamiltonian cycle encodes a move toward the seed block that will
complete it as a full square. This last growth block is distinguished
from all other growth blocks in its bottom row,
when the first move in the remaining input path is found on the
right edge of the block.

We omit figures for this set of gadgets since they only differ from
the gadgets for the regular left, right or straight moves in their
topmost row. This last row of the last block simply needs to be a
straight line in which each tile contains an output glue that will
bind to a single filler tile. This filler tile will assemble with
copies of itself in a comb-like structure with teeth of variable
length defined by the number of filler tiles it takes to ``bump'' into
the existing seed block assembly, as shown in
Figure~\ref{fig:putting_it_all_together} in the main body of this
paper.